\documentclass[journal,onecolumn]{IEEEtran}
\usepackage{liustyle}
\usepackage{amsmath,amsfonts}
\usepackage{bm}
\usepackage{algorithmic}
\usepackage{algorithm}
\usepackage{array}
\usepackage[caption=false,font=normalsize,labelfont=sf,textfont=sf]{subfig}
\usepackage{textcomp}
\usepackage{stfloats}
\usepackage{url}
\usepackage{verbatim}
\usepackage{graphicx}

\usepackage{cite}
\begin{document}

\title{Vectorized Generalized  Nearest Neighbor Decoding for In-block Memory Channel}

\author{Yuhao Liu, Xinwei Li, Shuqin Pang, Hao Wu, and Wenyi Zhang,~\IEEEmembership{Senior Member,~IEEE}
\thanks{Corresponding author: Wenyi Zhang.}
\thanks{Yuhao Liu is with the School of Mathematical Sciences, Peking University, Beijing 100871, China (e-mail: yhliu@math.pku.edu.cn).}
\thanks{Xinwei Li and Hao Wu are with the Department of Mathematical Sciences, Tsinghua University, Beijing 100084, China (e-mail: lxw22@mails.tsinghua.edu.cn; hwu@tsinghua.edu.cn).}
\thanks{Shuqin Pang and Wenyi Zhang are with the Department of Electronic Engineering and Information Science, University of Science and Technology of China, Hefei 230027, China (e-mail: shuqinpa@mail.ustc.edu.cn; wenyizha@ustc.edu.cn).}}

%



\maketitle

\begin{abstract}
This work extends the generalized nearest neighbor decoding (GNND), originally developed as a receiver architecture for memoryless channels, to a vectorized GNND (Vec-GNND) suitable for in-block memory (IBM) channels. Leveraging the generalized mutual information (GMI) as an operational lower bound on the mismatch capacity, an analytical characterization of the optimal Vec-GNND is obtained for general IBM channels with Gaussian codebooks. The formalism further provides closed-form optimality conditions and achievable GMIs for restricted variants of the receiver architecture. Furthermore, we formulate a GMI-based joint design viewpoint for Gaussian codebook covariance and decoding metrics. Since the metric optimization admits a closed-form solution for each fixed covariance, the joint design is reduced to an input-covariance optimization problem; for the diagonal covariance family, we derive first-order self-consistent optimality conditions. 
Numerical evaluations on block noncoherent AWGN and phase-noise channels demonstrate consistent GMI gains over conventional scaling-based GNND schemes, suggesting potential performance advantages of Vec-GNNDR in practical communication systems. 
\end{abstract}


\section{Introduction}
The modern theory of reliable communication begins with Shannon’s channel coding theorem \cite{shannon1948mathematical}, which establishes the operational meaning of channel capacity and shows that, over additive white Gaussian noise (AWGN) channels,  Gaussian codebooks together with maximum-likelihood decoding achieve the Shannon limit. In this canonical setting, the maximum-likelihood  rule coincides with nearest-neighbor (NN) decoding in Euclidean distance, providing a simple geometric interpretation fundamental to both theory and practice of communication systems \cite{jacobs1965principles}.

While nearest-neighbor decoding is optimal for  AWGN and conditionally AWGN channels, many practical communication environments deviate from this idealized model. Real-world channels are often non-Gaussian, uncertain, or even entirely unknown due to fading phenomena and non-ideal transceiver distortions \cite{bjornson2014massive}. In such settings, identifying the capacity-achieving input distribution together with the corresponding optimal decoding rule becomes highly intractable and computationally burdensome \cite{lapidoth2002reliable,scarlett2014mismatched}, hindering their deployment in practice. This tension between information-theoretic optimality and implementability has long motivated the study of low-complexity decoding rules that trade a calibrated loss in rate for robustness and simplicity.

Mismatched decoding provides a principled framework that employs a fixed decoding metric not necessarily aligned with the true channel characteristics (i.e., conditional probability distributions) \cite{merhav1994information, csiszar2002channel, scarlett2020information}, offering a realistic model for communication over uncertain or time-varying channels. Within this framework, the mismatch capacity denotes the supremum of achievable rates with arbitrarily small error probability, given the specified mismatched decoding metric but optimized over all possible codebooks. However, the exact characterization of mismatch capacity is notoriously intractable for general channels and decoding metrics \cite{csiszar2002channel}. To address this challenge,  two widely-considered achievable rates, namely the generalized mutual information (GMI) \cite{fischer1978some,695156}, derived from random-coding analysis, and the LM rate \cite{merhav1994information}, obtained through constant-composition coding, provide tractable lower bounds on the mismatch capacity. Building upon these formulations, subsequent research has explored GMI and LM rate as principled objectives for the design of constellations and codebooks \cite{8240991,2017Constellation, Jones2019, koike2016gmi,chen2024double,kramer2023information}.

Beyond optimizing codebooks to enhance achievable rates under mismatched decoding, some parallel efforts have focused on improving these rates by developing more effective and implementable decoding rules for fixed codebooks through GMI optimization. For  Gaussian codebooks, \cite{zhang2011general} maintains the NN decoding rule while introducing and optimizing a scaling coefficient on the channel input to enhance the achievable GMI for general memoryless channels. In \cite{zhang2016remark,zhang2019regression}, an additional output processing function is introduced alongside the scaling coefficient, and both are jointly optimized to further enhance the GMI. The optimal processing function serves to linearize the nonlinear channel response effectively, rendering it to behave as closely as possible to an additive-noise channel, consistent with the Bussgang decomposition formalism \cite{bussgang1952crosscorrelation}. Furthermore, a generalized nearest-neighbor decoding rule (GNNDR) was proposed in \cite{wang2022generalized} to further boost the GMI without sacrificing the simplicity of NN decoding principles. The GNNDR augments the NN decoding with two tunable ingredients: a processing function applied to the observation and a scaling function applied to codewords, both potentially dependent on the received signal and the channel state information (CSI). Within this formalism, the optimal processing and scaling functions, as well as the corresponding GMI, can be derived analytically for Gaussian codebooks, and the formalism remains applicable and extendable when the processing and scaling functions are subject to restrictions. The GNNDR formalism has subsequently been extended and applied to a variety of scenarios, including general input constellations \cite{pang2025generalized}, channels with imperfect CSI \cite{pang2021generalized}, and quasi-static fading channels\cite{shi2022linear} .


Despite the theoretical and practical appeal of the GNNDR scheme, existing studies have largely been confined to symbol-wise memoryless channels. Many practical communication systems, however, exhibit statistical dependence within finite blocks, arising from bandlimitation and dispersion \cite{proakis2001digital}, multipath propagation and mobility \cite{biglieri2002fading}, phase noise \cite{ghozlan2017models}, and related effects. We model such scenarios through in-block memory (IBM) channels \cite{2642,kramer2014information}, where different blocks are statistically independent while the symbols within each block may be dependent.

We emphasize that the IBM model is not introduced here as a new channel model in itself: mathematically, one may regard each block as a vector-valued channel use. The key point is instead architectural. A direct element-wise application of GNNDR ignores the statistical structure shared by symbols within the same block, whereas a block-level nearest-neighbor metric can exploit this structure through matrix-valued scaling and vector-valued processing. This distinction is essential, for example, in block noncoherent AWGN channels \cite{lapidoth2002phase,nuriyev2005capacity}, where all symbols in a block share a common unknown phase, and in phase-noise channels \cite{zou2007compensation}, where the phase process is correlated within a block. Motivated by this observation, the present work introduces the vectorized GNNDR (Vec-GNNDR) and characterizes its GMI-optimal metric for Gaussian codebooks. The main contributions of this work are summarized as follows:
\begin{itemize}
    \item We formulate Vec-GNNDR as a block-level mismatched nearest-neighbor architecture for IBM channels. Unlike element-wise GNNDR, the proposed metric allows vector-valued processing and matrix-valued codeword scaling, thereby preserving the statistical structure within each block. For Gaussian codebooks, we derive an analytical expression for the GMI of Vec-GNNDR using the moment-generating function technique.
    \item We characterize the supremum GMI in closed form and construct $\epsilon$-optimal Vec-GNNDR metrics without the conditional-variance upper bound imposed in~\cite{wang2022generalized}. The analysis combines spectral majorization and convexity to establish global optimality of the posterior-covariance eigenbasis, and introduces a truncation construction that preserves conditional-mean contributions at vanishing scaling coefficients. The resulting spectral characterization may be of independent interest for the optimization of quadratic decoding metrics.
    \item  Building upon the derivation of the optimal Vec-GNNDR, we formulate a joint codebook-metric design problem over Gaussian input covariance matrices and Vec-GNNDR metrics. For every fixed covariance, the inner metric optimization is solved in closed form, reducing the joint design to an input-covariance optimization problem. As a tractable first instance, we derive first-order self-consistent optimality conditions for the diagonal covariance family.
    \item The optimal Vec-GNNDR is further characterized under several restricted configurations of the processing and scaling functions, including constant scalar scaling, constant matrix scaling, CSI-dependent scalar scaling, CSI-dependent matrix scaling, and linear processing. The corresponding optimal GMI for each restricted scenario is obtained analytically.
    \item We numerically evaluate the optimal GMI achieved by the optimal Vec-GNNDR over block noncoherent AWGN channels and phase-noise channels, demonstrating the performance superiority of the proposed Vec-GNNDR scheme.
\end{itemize}

The remainder of this paper is organized as follows. Section~\ref{sec:set-up} presents the system model. Section~\ref{sec:main} derives the GMI for the proposed Vec-GNNDR, provides an analytical characterization of the optimal Vec-GNNDR and its corresponding GMI, and develops a joint codebook-metric design perspective through first-order conditions for diagonal Gaussian covariance matrices. Section~\ref{sec:restricted_vec_GNNDR} investigates a number of restricted or suboptimal configurations of the Vec-GNNDR and establishes their achievable GMIs. Section~\ref{sec:simulation} numerically validates the performance of the Vec-GNNDR in two representative IBM channel scenarios. Finally, Section~\ref{sec:conclusion} concludes the paper and highlights directions for future work.

\paragraph*{Notation}  For a complex scalar $a$, $|a|$ denotes its modulus and $\mathrm{Arg}(a)$ its principal argument. For a complex square matrix $\bm{A}$, $\det(\bm{A})$ and $\mathrm{Tr}(\bm{A})$ denote its determinant and trace, respectively, and $\bm{A}^{\dagger}$ denotes its Moore-Penrose pseudoinverse.   The symbol $\|\cdot\|_2$ denotes the Euclidean ($\ell_2$) norm on complex vector spaces. For a complex matrix $\bm{A}$ or vector $\bm{v}$, we write $\bm{A}^H$ and $\bm{v}^H$ for their conjugate transposes.  For a Hermitian matrix $\bm{A} \in \mathbb{C}^{n \times n}$ admitting the eigendecomposition $\bm{A} = \bm{W} \mathrm{diag}\{\lambda_1, \dots, \lambda_n\} \bm{W}^H$, and a real-valued function $f: \mathbb{R} \to \mathbb{R}$, the corresponding matrix function is defined as $f(\bm{A}) = \bm{W} \mathrm{diag}\{f(\lambda_1), \dots, f(\lambda_n)\} \bm{W}^H$. The unitary group of degree $n$ is denoted by $\mathbb{U}(n)$. We use $\mathcal{N}(\bm{\mu}, \bm{\Sigma})$ and $\mathcal{CN}(\bm{\mu}, \bm{\Sigma})$ to denote, respectively, the real multivariate normal distribution and the circularly symmetric complex normal distribution with mean vector $\bm{\mu}$ and covariance matrix $\bm{\Sigma}$. The notation $\chi^2(n)$ denotes the chi-squared distribution with $n$ degrees of freedom.  
Finally, for a set $\Omega$, we use $\mathbb{I}\{ x \in \Omega\}$ to denote the indicator function of $\Omega$, 
and we define the positive-part function as $(x)_{+} = \max\{x, 0\}$.

\section{System Model} \label{sec:set-up}
We consider a general in-block memory channel with input $\bm{X} \in \cX = \mathbb{C}^{B_x}$, output $\bm{Y} \in \mathcal{Y} = \CC^{B_y}$,
state $\bS \in \mathcal{S}$, and receiver channel state information (CSI) $\bV \in \cV$. The sets $\cS$ and $\cV$ are general, not necessarily scalar-valued. We assume that the state is independent of the input, and that the channel is used without feedback. A code of length  $N:= LB_x$  is divided into $L$ blocks, each associated with an input of length $B_x$. Formally, we have
\begin{equation} \label{eq:in-channel}
    \begin{aligned}
        P_{\bS, \bV}\left(\bs^{L}, \bv^{L}\right) & =\prod_{l=1}^{L} P_{\bS, \bV}\left(\bs_{l}, \bv_{l}\right), \\
        P_{\bY \mid \bX, \bS}\left(\by^{L} \mid \bx^{L}, \bs^{L}\right) & =\prod_{l=1}^{L} P_{\bY \mid \bX, \bS}\left(\by_{l} \mid \bx_{l}, \bs_{l}\right).
\end{aligned}
\end{equation}
This formulation indicates that the in-block memory channel is memoryless across blocks while potentially exhibiting memory within each block, as captured by the joint conditional probability $P_{\bY | \bX, \bS}$. We further assume that,  throughout each block transmission, the relation $\bV \leftrightarrow \bS \leftrightarrow (\bX,\bY)$ forms a Markov chain, which implies that the receiver CSI $\bV$ is generated through a mechanism independent of the current channel realization. The channel model is illustrated in Figure \ref{fig:channel}.
 
\begin{remark}
    For an in-block memory channel, the received vector $\by_l$ of length $B_y$ need not have the same length $B_x$ as the transmitted vector $\bx_l$. 
    Examples with $B_x>B_y$  include the multiple-input single-output (MISO) systems \cite{shafiee2007achievable}. Examples with $B_x<B_y$ include  the single-input multi-output (SIMO) channel \cite{parada2005secrecy}, fractionally-spaced channel \cite{johnson1995fractionally}.
\end{remark}

\begin{figure}[htbp]
    \centering
    \includegraphics[width=18cm]{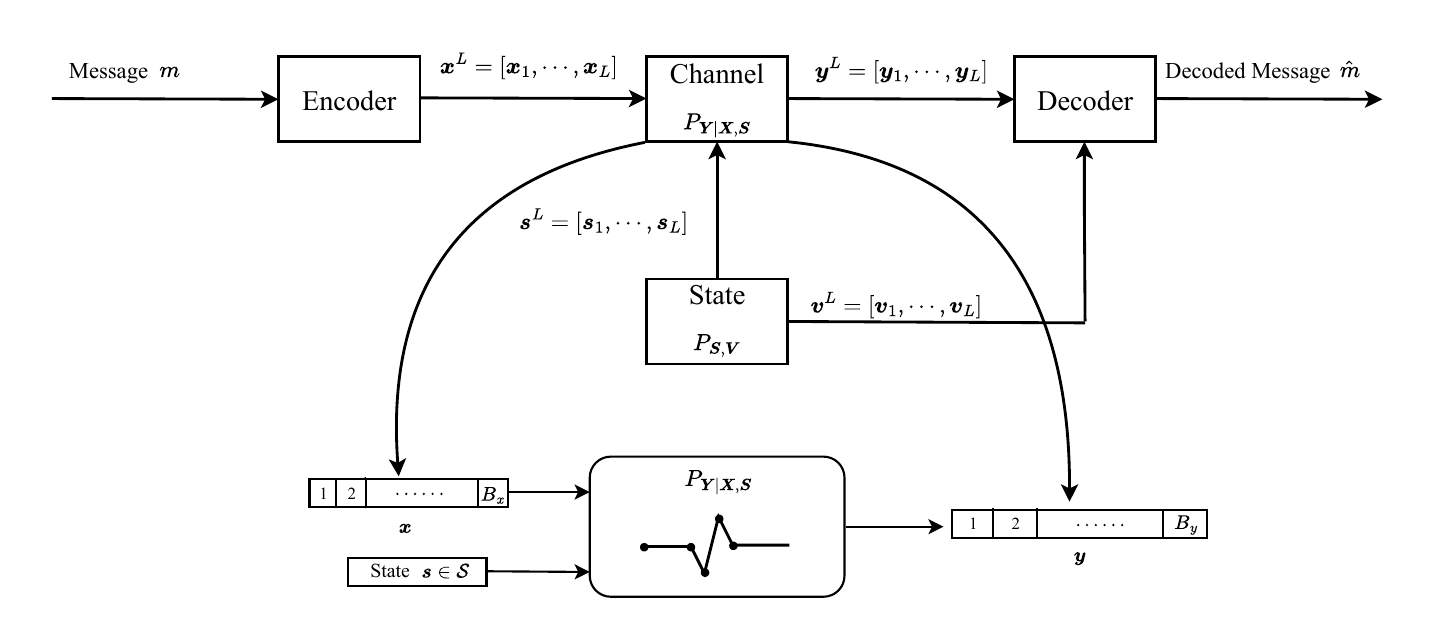}
    \caption{Illustration of an in-block memory channel. The message $m$ is encoded into a sequence $\bx^L$, where each input block $\bx_l$ has dimension $B_x$. The channel is characterized by the conditional distribution $P_{\bY|\bX,\bS}$, which maps the input $\bx_l$ and state $\bs_l$ to an output vector $\by_l$ of length $B_y$. The receiver utilizes the received sequence $\by^L$ and the CSI $\bv^L$ to produce the estimate $\hat{m}$.}
    \label{fig:channel}
\end{figure}

We employ an i.i.d. Gaussian random codebook ensemble, where the codebook is composed of mutually independent codewords drawn from $\mathcal{CN}(0, P \bI_N)$ under an average power constraint $P$. For a code rate $R$, the encoder selects a message $m$ uniformly at random from $\mathcal{M}=\left\{1, \ldots,\left\lceil e^{N R}\right\rceil\right\}$ for transmission. The encoding function $\mathcal{E}: \mathcal{M} \mapsto \mathcal{X}^{L}=(\mathbb{C}^{B_x})^{\otimes L}$ assigns each $m \in \mathcal{M}$ to an $N$-symbol codeword $\bx^{L}(m)$ comprising $L$ blocks, with $\bx^{L}(m) \sim \mathcal{CN}(0, P \boldsymbol{I}_N)$.

\begin{remark}
	Indeed, the codebook can be generated such that each block follows a colored Gaussian distribution $\bx_l(m) \sim \mathcal{CN}(\mathbf{0}, P \bm{\Sigma})$ with power constraint $\mathrm{Tr}(\bm{\Sigma}) = B_x$, explicitly allowing for intra-block correlation. However, via the whitening transformation $\hat{\bx}_l(m) = \bm{\Sigma}^{-\frac{1}{2}} \bx_l(m)$, we can transform the inputs into the white Gaussian form $\tilde{\bx}_l(m) \sim \mathcal{CN}(\mathbf{0}, P \mathbf{I}_{B_x})$. Consequently, the problem reduces to transmission through an equivalent channel with transition probability $P_{\bY \mid \tilde{\bX}, \bS}(\by_l \mid \hat{\bx}_l, \bs_l) = P_{\bY \mid \bX, \bS} (\by_l \mid \bm{\Sigma}^{\frac{1}{2}} \hat{\bx}_l, \bs_l)$ with white Gaussian random codebook.  This observation motivates a joint codebook-metric design viewpoint, in which $\bm{\Sigma}$ specifies the Gaussian codebook ensemble while the Vec-GNNDR metric is optimized conditionally on that ensemble; see Section~\ref{sec:cb_opt}.
	\label{remark:codebook}
\end{remark}

Upon receiving the channel output $\by^{L}$ and  the corresponding CSI $\bv^{L}$, the decoder employs a decoding function $\mathcal{D}: \mathcal{X}^{L} \times \mathcal{Y}^{L} \mapsto \mathcal{L}$ to map the received signals $(\by^{L}, \bm{v}^{L})$  to a decoded message $\hat{m}$. As discussed in the introduction,  the optimal  decoding strategy is the maximum likelihood (ML) decoding. For Gaussian channels,  ML decoding simplifies to the computationally efficient nearest-neighbor decoding rule (NNDR) based on Euclidean distance. However, for general in-block memory channels, ML decoding  can be computationally intractable. Motivated by the efficiency of NNDR, we adopt the following Vec-GNNDR:
\begin{equation} \label{eq:dec_rule}
    \hat{m}=\arg \min _{m \in \mathcal{M}} \sum_{l=1}^{L}\left\|\bg\left(\by_{l}, \bv_{l}\right)-\bm{f}\left(\by_{l}, \bv_{l}\right) \bx_{l}(m)\right\|_2^{2},
\end{equation}
where $\bg: \cY \times \cV \mapsto \CC^{B}$ is a processing function and $\bm{f}:\cY \times \cV \mapsto \CC^{B \times B_x}$ is a scaling matrix. For the convenience of subsequent theoretical analysis, we assume 
\begin{equation} \label{def:g}
    \bg(\by, \bv) = [g_1(\by, \bv), g_2(\by, \bv), \cdots, g_{B}(\by, \bv)]^{\top},
\end{equation}
and let $\bm{f}(\by, \bv)$ admit the following singular value decomposition (SVD):
\begin{equation} \label{def:f}
    \bm{f}(\by, \bv) = \sum_{i=1}^{B \wedge B_x} \sigma_{i}(\by, \bv) \bu_i(\by, \bv) \bw^H_i(\by, \bv),  
\end{equation}
where $\sigma_i: \cY \times \cV \mapsto \RR_{\geq 0}$, $\bu_i:  \cY \times \cV \mapsto \CC^{B}$ for $1\leq i \leq B$ and $\bw_i:  \cY \times \cV \mapsto \CC^{B_x}$ for $1 \leq i \leq B_x$. We further define  
\begin{equation*}
    \begin{aligned}
        \bU(\by, \bv) &= [\bu_1(\by, \bv), \bu_2(\by, \bv), \cdots, \bu_{B}(\by, \bv)], \\
        \bW(\by, \bv) &= [\bw_1(\by, \bv), \bw_2(\by, \bv), \cdots, \bw_{B_x}(\by, \bv)],
    \end{aligned}
\end{equation*}
hence $\bU: \cY \times \cV \mapsto \mathbb{U}(B)$ and $\bW: \cY \times \cV \mapsto \mathbb{U}(B_x)$.

Noting that $\bm{U}(\by, \bv) \in \mathbb{U}(B)$ is an isometric transformation, it preserves distances, as demonstrated by the following equality:
\begin{equation}
    \left \|\bU^H(\by,\bv) \left[ \bg(\by, \bv) - \bm{f}(\by, \bv) \bx\right] \right \|_2 = \left \|\bg(\by, \bv) - \bm{f}(\by, \bv) \bx \right \|_2,
\end{equation}
without changing the decoding rule~\eqref{eq:dec_rule}. Thus, without loss of generality, we may assume $\bm{U}(\by, \bv) = \bm{I}_{B}$. 

For the case of $B<B_x$, the scaling matrix projects $\bx_l(m)$ onto a lower-dimensional subspace, which may discard information relevant to decoding and consequently reduce the achievable rate. For this reason, this case is not considered further in the paper. For the case of $B \geq B_x$, the Vec-GNNDR exhibits an inherent redundancy: the additional dimensions beyond \(B_x\) do not contribute to the decoding decision, as established in the following lemma.

\begin{lemma}
    \label{lem:reduction}
    The Vec-GNNDR with $B\geq B_x$ can be equivalently reduced to the case of $B= B_x$ without altering the decoding result.
\end{lemma}
\begin{proof}
    Without loss of generality, we set $\bU = \boldsymbol{I}_{B_x}$. 
    When $B \geq B_x$, by applying the SVD in~\eqref{def:f}, the Vec-GNNDR can be rewritten as
    \begin{equation*}
        \hat{m} = \arg \min _{m \in \mathcal{M}} \sum_{l=1}^{L} \sum_{i=1}^{B_x} \left| g_i(\by_l, \bv_l) - \sigma_i(\by_l, \bv_l)  (\bw_i^H \bx_l(m))\right|^2 + L \sum_{i=B_x+1}^{B} \left|g_i(\by_l, \bv_l)\right|^2.
    \end{equation*}
    The last $B - B_x$ components are independent of $\bx(m)$ and thus have no influence on the minimization in the decoding rule. Consequently, the Vec-GNNDR effectively reduces to the case $B = B_x$. This completes the proof.
\end{proof}
In view of Lemma~\ref{lem:reduction}, subsequent analyses are henceforth restricted to the case $B=B_x$, without loss of generality.

Although the Vec-GNNDR is not generally optimal, it is a tractable solution, and the inclusion of carefully designed processing and scaling functions further boosts its performance. In the next section, we employ the GMI approach to  characterize the achievable rate with decoding rule \eqref{eq:dec_rule} and derive the optimal processing-scaling pair by maximizing the GMI.


\section{GMI and Optimal Vec-GNNDR} \label{sec:main}
In this section, we derive the GMI for arbitrary Vec-GNNDR schemes and determine the optimal Vec-GNNDR by maximizing the GMI with respect to both the processing and scaling functions.

\subsection{GMI of Vec-GNNDR}
We extend the standard terminology and notation of mismatched decoding from memoryless channels to in-block memory channels. Consider an in-block memory channel depicted by the conditional probability density $P_{\bY \mid \bX}$ with input $\bX \in \cX$ and $\bY \in \cY$. For a target rate $R$ and block length $B_x$,  a codeword of length $N = LB_x$ is partitioned into $L$ blocks. A codebook $\mathcal{C}$ contains $\lceil e^{NR}\rceil$ codewords. Each message $m \in \mathcal{M} \triangleq\left\{1, \ldots,\left\lceil e^{N R}\right\rceil\right\}$ is mapped to $\bx^{L}(m) = (\bx_1(m), \bx_2(m), \cdots, \bx_L(m))$ with $\bx_l(m) \in \cX$, $1 \leq l \leq L$. Let $d: \cX \times \cY \to R$ be a so-called "decoding metric". The induced mismatched decoder selects
\begin{equation*}
    \mathcal{D}_{d}: \hat{m}=\arg \min _{m \in \mathcal{M}} \sum_{l=1}^{L} d\left(\bx_{l}(m), \by_{l}\right),
\end{equation*}
with ties resolved arbitrarily. A rate $R$ is achievable if there exists a sequence of $(L,R)$-codebooks for which the maximal decoding-error probability approaches zero asymptotically as $L \to \infty$. The supremum of all achievable rates is called the mismatch capacity of the in-block memory channel.

In \cite{2642}, the Shannon capacity of in-block memory channels is characterized from a game-theoretic perspective. However, for mismatched decoding, the exact mismatch capacity  remains an open problem even for memoryless channels, with only upper and lower bounds currently available \cite{lapidoth2002reliable,scarlett2020information}. Determining the mismatch capacity for in-block memory channels is therefore also intractable in general. We resort to the GMI, which provides a convenient lower bound on the mismatch capacity for in-block memory channels.

To characterize the GMI, we consider the Vec-GNNDR with processing function $\bg$ and scaling function $\bm{f}$ satisfying the following assumption.
\begin{assumption} \label{ass:strong_law}
    The processing function $\bg$ and scaling function $\bm{f}$ satisfy 
    \begin{equation*}
        \mathbb{E}_{\bY, \bV} \left[ | g_i(\bY, \bV)|^2 \right]< +\infty, \quad \mathbb{E}_{\bY, \bV} \left[\sigma_i^2(\bY, \bV) \right]< +\infty, \quad \mathrm{for} \; 1\leq i \leq B_x
    \end{equation*}
    and $\mathbb{E}_{\bY, \bV} \left[ \left\| \bm{f}(\bY, \bV) \bX\right\|_2^2\right]< + \infty$.
\end{assumption}

\begin{proposition} \label{prop:GMI}
    Consider the in-block memory channel defined in \eqref{eq:in-channel}. 
    Let $\bg$ and $\bm{f}$ be fixed functions as specified in \eqref{def:g} and \eqref{def:f}, respectively, satisfying Assumption \ref{ass:strong_law}. 
    Under this setting, the decoding metric is
    \begin{equation*}
        d(\bx, (\by, \bv)) = \left\|\bg(\by, \bv)  - \bm{f}(\by, \bv) \bx\right\|_2^2.
    \end{equation*}
    The resulting GMI is given by 
    \begin{equation} \label{eq:I_GMI_gf}
        I_{\mathrm{GMI}, \bg, \bm{f}}=B_x^{-1}\max _{\theta<0} \left\{ \theta \EE \left[\|\bg(\bY, \bV) - \bm{f}(\bY, \bV) \bX\|_2^{2}\right] -\sum_{i=1}^{B_x} \EE \left[\frac{\theta|g_i(\bY, \bV)|^2}{1-\theta P \sigma_i^2(\bY, \bV) } - \log \left(1-\theta P \sigma_i^2(\bY, \bV) \right)\right]\right\}.
    \end{equation}
\end{proposition}

\begin{proof}
    We begin by analyzing the average decoding error probability over messages and codebook ensembles, denoted by $P(\hat{m} \neq m)$. Due to the inherent symmetry of the i.i.d. Gaussian codebook ensemble, we assume without loss of generality that the transmitted codeword corresponds to the message $m = 1$. Consequently, we have
    \begin{equation*}
        P(\hat{m} \neq m) = P(\hat{m} \neq 1 | m= 1).
    \end{equation*}
    Under the assumption $m=1$, the normalized decoding metric becomes 
    \begin{equation} \label{eq:converge_D1}
        \begin{aligned}
            \mathrm{D}(1) & =\frac{1}{L} \sum_{l=1}^{L} d\left(\bX_{l}(1),\left(\bY_{l}, \bV_l \right)\right)  =\frac{1}{L} \sum_{l=1}^{L}\|\bg(\bY, \bV) - \bm{f}(\bY, \bV) \bX\|_2^{2}  \rightarrow \EE \left[\|\bg(\bY, \bV) - \bm{f}(\bY, \bV) \bX\|_2^{2}\right],
        \end{aligned}
    \end{equation}
    where the convergence holds almost surely by the strong law of large numbers (Lemma~\ref{lem:SLLN}),  which is ensured by  Assumption~\ref{ass:strong_law}. For any $\delta>0$, we define the event
    \begin{equation*}
        \mathcal{A}_{\delta}=\left\{\mathrm{D}(1) \geq \EE \left[\|\bg(\bY, \bV) - \bm{f}(\bY, \bV) \bX\|_2^{2}\right]+\delta\right\}.
    \end{equation*}
    Thus, the average decoding error probability can be decomposed as follows:
    \begin{equation*}
        \begin{aligned}
        P(\hat{m} \neq 1 \mid m=1) & =P\left(\hat{m} \neq 1 \mid m=1, \mathcal{A}_{\delta}\right) P\left(\mathcal{A}_{\delta}\right)+P\left(\hat{m} \neq 1 \mid m=1, \mathcal{A}_{\delta}^{c}\right) P\left(\mathcal{A}_{\delta}^{c}\right) \\
        & \leq P\left(\mathcal{A}_{\delta}\right)+P\left(\hat{m} \neq 1 \mid m=1, \mathcal{A}_{\delta}^{c}\right) P\left(\mathcal{A}_{\delta}^{c}\right).
        \end{aligned}
    \end{equation*}
    The first term on the right-hand side vanishes asymptotically for sufficiently large $L$. The second term can be bounded as 
    \begin{equation*}
        \begin{aligned}
            P\left(\hat{m} \neq 1 \mid m=1, \mathcal{A}_{\delta}^{c}\right) P\left(\mathcal{A}_{\delta}^{c}\right) & \leq P\left(\exists m^{\prime} \neq 1, \mathrm{D}\left(m^{\prime}\right)<\EE \left[\|\bg(\bY, \bV)-\bm{f}(\bY, \bV) \bX\|_2^{2}\right]+\delta \mid \mathcal{A}_{\delta}^{c}\right) P\left(\mathcal{A}_{\delta}^{c}\right) \\
            & \leq e^{L B_x R} P\left(\mathrm{D}(2)<\EE \left[\|\bg(\bY, \bV)-\bm{f}(\bY, \bV) \bX \|_2^{2}\right]+\delta \mid \mathcal{A}_{\delta}^{c}\right) P\left(\mathcal{A}_{\delta}^{c}\right) \\
            & =e^{L B_x R} P\left(\mathrm{D}(2)<\EE\left[\|\bg(\bY, \bV)-f(\bY, \bV) \bX \|_2^{2}\right]+\delta, \mathcal{A}_{\delta}^{c}\right) \\
            & \leq e^{L B_x R} P\left(\mathrm{D}(2)<\EE\left[\|\bg(\bY, \bV)-\bm{f}(\bY, \bV) \bX \|_2^{2}\right]+\delta\right).
        \end{aligned}
    \end{equation*}
    Using the law of total expectation, we have
    \begin{equation} \label{eq:cond_proB}
        \begin{aligned}
         P\left(\mathrm{D}(2)<\EE \left[\|\bg(\bY, \bV)-\bm{f}(\bY, \bV) \bX\|_2^{2}\right]+\delta\right) 
        = \EE \left[P\left(\mathrm{D}(2)<\EE\left[\|\bg(\bY, \bV)-\bm{f}(\bY, \bV) \bX \|_2^{2}\right]+\delta \mid \bY^{L}, \bV^{L}\right)\right],
        \end{aligned}
    \end{equation}
    where the expectation is taken over $(\bY^{L}, \bV^{L}) = \{(\bY_l, \bV_l)_{l=1}^L\}$, a collection of $L$ \iid random variables $(\bY, \bV)$.
    
    Conditioned upon $(\bY^L, \bV^L)$, the normalized decoding metric 
    \begin{equation*}
        \mathrm{D}(2) = \frac{1}{L} \sum_{l=1}^L \underbrace{\|\bg(\bY_l, 
        \bV_l) - \bm{f}(\bY_l, \bV_l) \bX(2)\|_2^2}_{\tilde{\chi}^2_l},
    \end{equation*}
    is the  empirical mean of $L$ independent generalized $\chi^2$ random variables $\{\tilde{\chi}^2_l\}_{l=1}^L$. Specifically, 
    \begin{equation*}
        \tilde{\chi}_l^2 = \sum_{l=1}^{B_x} \frac{1}{2} \sigma_i^2(\bY_l, \bV_l) \chi^2 \left(2, \frac{|g_i(\bY_l, \bV_l)|^2}{\sigma_i^2(\bY_l, \bV_l) }\right),
    \end{equation*}
    where $\chi^{2}(k,\lambda)$ denotes a non-central chi-square variable with $k$ degrees of freedom and non-centrality parameter $\lambda$. The moment generating function of $\tilde{\chi}^2_l$ is given by 
    \begin{equation*}
        M_{l}(\theta) = \EE e^{\theta \tilde{\chi}_l^2} = \exp{\left\{ \sum_{i=1}^{B_x}\frac{|g_i(\bY_l, \bV_l)|^2\theta}{1- \sigma_i^2(\bY_l, \bV_l)\theta P}\right\}} \prod_{i=1}^{B_x} \left[1 - \sigma^2_i(\bY_l, \bV_l) \theta P  \right]^{-1}
    \end{equation*}
    for any $\theta<0$. 
    Applying Markov's inequality, we obtain
    \begin{equation*}
        \begin{aligned}
         P\Big(D(2) < \mathbb{E} \big[\|\bg(\bY, \bV)  - \bm{f}(\bY, & \bV) \bX \|_2^{2}\big]\Big) 
         = P\left(\exp\left\{\theta L D(2)\right\} > \exp\left\{\theta L \mathbb{E}\left[\|\bg(\bY, \bV) - \bm{f}(\bY, \bV) \bX\|_2^{2}\right]\right\}\right) \\
        & \leq \min_{\theta<0} \mathbb{E}_{(\bY^L, \bV^L)}\left[\exp\left\{\theta L D(2) - \theta L \mathbb{E}\left[\|\bg(\bY, \bV) - \bm{f}(\bY, \bV) \bX\|_2^{2}\right]\right\}\right] \\
        & = \min_{\theta<0}  \mathbb{E}_{(\bY^L, \bV^L)} \exp{ \left\{ L\left(\frac{1}{L} \sum_{l=1}^L \log M_l(\theta) - \theta \mathbb{E}\left[\|\bg(\bY, \bV) - \bm{f}(\bY, \bV) \bX\|_2^{2}\right] \right)\right\}}.
        \end{aligned}
    \end{equation*}
    With Assumption~\ref{ass:strong_law}, we have 
    \begin{equation*}
        \begin{aligned}
            & \mathbb{E}_{\bY, \bV} \left| \left[ \frac{|g_i(\bY, \bV)|^2\theta}{1- \sigma_i^2(\bY, \bV)\theta} \right]\right|  \leq \mathbb{E}_{\bY, \bV} |\theta| \left| g_i(\bY,\bV) \right|^2< +\infty,  \quad 1 \leq i \leq B_x ,\\
            & \mathbb{E}_{\bY, \bV} \left| \log(1- \theta \sigma_i^2(\bY, \bV))\right| \leq \mathbb{E}_{\bY, \bV} |\theta| \sigma_i^2(\bY, \bV) < +\infty, \quad 1 \leq i \leq B_x.
        \end{aligned}    
    \end{equation*}
    Therefore, by the strong law of large numbers (Lemma~\ref{lem:SLLN}),
    \begin{equation} \label{eq:converge_M}
        \frac{1}{L} \sum_{l=1}^L\log M_l(\theta) \to \sum_{i=1}^{B_x} \EE_{\bY,\bV} \left[ \frac{|g_i(\bY, \bV)|^2\theta}{1- \sigma_i^2(\bY, \bV)\theta} \right]- \sum_{i=1}^{B_x} \EE_{\bY, \bV} \left[ \log \left(1 - \theta \sigma^2_i(\bY, \bV)  \right) \right],
    \end{equation}
    and thus, the conditional probability in \eqref{eq:cond_proB} exponentially decays to zero at a rate given by
    \begin{equation} \label{eq:max_GMI_opt}
        B_x^{-1}\max_{\theta<0} \theta\left(\EE\left[\|\bg(\bY, \bV) - \bm{f}(\bY, \bV) \bX\|_2^{2}\right]\right)-\sum_{i=1}^{B_x} \left(\EE  \left[\frac{\theta|g_i(\bY, \bV)|^{2}}{1-\theta  P \sigma_i^2(\bY, \bV)}\right]-\EE\left[\log \left(1-\theta P \sigma_i^2\left(\bY, \bV\right) \right)\right] \right).
    \end{equation}
    Consequently, for any rate $R < I_{\text{GMI}, \bg, \bm{f}}$, we can find sufficiently small $\delta>0$ and sufficiently large $L$ such that the average decoding error probability approaches zero.
\end{proof}
\subsection{GMI-maximizing Vec-GNNDR} \label{subsec:GMI-max}
In this subsection, we investigate the optimal choices of the functions $\bg(\by, \bv)$ and $\bm{f}(\by, \bv)$ that maximize the GMI $I_{\text{GMI}, \bg, \bm{f}}$ given in Proposition.~\ref{prop:GMI}.  Since rescaling both $\bg(\by, \bv)$ and $\bm{f}(\by, \bv)$  by the same factor does not affect the decoded message $\hat{m}$, the scalar parameter $\theta$  can be absorbed into both processing and scaling functions.
For simplicity, we henceforth set $\theta = -1$.
To facilitate the optimization for GMI, we impose the following assumption on the channel:
\begin{assumption}
    For the in-block memory channel~\eqref{eq:in-channel}, the conditional covariance matrix  $\mathrm{Cov}[\bX \mid \bY, \bV]$ is positive definite almost surely.

    \label{ass:psd_cond_cov}
\end{assumption}
\begin{remark}
    This is both a meaningful and technically essential assumption.  If $\operatorname{Cov}[\bX \mid \bY, \bV]$ has zero eigenvalues with nonzero probability, then there exists a measurable set $A$ with probability mass $p = P\{(Y,V) \in A\} >0 $, such that for any $(\by,\bv) \in A$, the support of the conditional distribution of $\bX \mid \by, \bv$ is contained in an affine subspace of dimension strictly smaller than $B_x$.  However, since the prior distribution $\mathcal{CN}(0, P \boldsymbol{I}_{B_x})$ is absolutely continuous and has strictly positive density everywhere on $\mathbb{R}^{B_x}$, the corresponding mutual information is given by
    \begin{equation*}
        \begin{aligned}
            I(\bX; \bY, \bV) & = \mathbb{E}_{\bY, \bV} \left[ D(P_{\bX \mid \bY ,\bV}(\cdot \mid \bY, \bV) || \mathcal{CN}(0, P\boldsymbol{I}_{B_x}))\right] \\
            & \geq \mathbb{E}_{\bY, \bV} \left[ D(P_{\bX \mid \bY ,\bV}(\cdot \mid \bY, \bV) || \mathcal{CN}(0, P\boldsymbol{I}_{B_x})) \mathbb{I}\{(\bY, \bV) \in A\}\right] \geq p \cdot (+\infty) = +\infty,
        \end{aligned}
    \end{equation*}
    where $D(\cdot\|\cdot)$ denotes the Kullback--Leibler divergence. Consequently, the achievable rate of the in-block memory channel is unbounded. To illustrate, consider the special case $B_x = B_y = 1$ with no CSI $V$.  If the channel output $y \in A$, then $\operatorname{Var}[X \mid y] = 0$, which means the channel input can be exactly recovered via $\mathbb{E}[X \mid y]$.  Suppose each codeword is transmitted $n$ times independently.  Then the probability of exact recovery is $1 - (1-p)^n \to 1$ as $n \to \infty$. This exact recovery capability allows the minimum codeword distance to be chosen arbitrarily small, thereby enabling the codebook to contain an arbitrarily large number of codewords. Consequently, the achievable rate can be made arbitrarily large. It is of course not meaningful to consider scenarios where the achievable rate is infinite.

    Meanwhile, this assumption also serves a technical role: it guarantees that all eigenvalues of the conditional covariance matrix are strictly positive, thereby avoiding singularity issues that would otherwise arise in the optimization of the processing function $\bg$ and the scaling function $\bm{f}$ in the subsequent analysis.
\end{remark}
The Vec-GNNDR with its optimal processing and scaling functions, together with the corresponding GMI, is  characterized in the following theorem.
\begin{theorem} \label{thm:main_result}
    Consider the in-block memory channel described in Section \ref{sec:set-up}, which satisfies Assumption \ref{ass:psd_cond_cov}. For a given channel output $\by$ and receiver CSI $\bv$,
    define the conditional mean and the conditional covariance matrix respectively as
    \begin{equation} \label{eq:cond_mean_cov}
        \bm{\mu}(\by, \bv) = \mathbb{E}[\bX \mid \by, \bv], \qquad \bm{\Sigma}(\by, \bv) = \EE \left[ \bX \bX^H | \by, \bv \right] - \bm{\mu}(\by, \bv) \bm{\mu}^H(\by, \bv).
    \end{equation}
    Let $\lambda_i(\by, \bv)$ denote the $i$-th largest eigenvalue of $\bm{\Sigma}(\by, \bv)$, and let $\bW(\by, \bv)$ be the corresponding diagonalizing matrix such that
    \begin{equation*}
        \bW^H(\by,\bv) \bm{\Sigma}(\by, \bv) \bW(\by, \bv) = \operatorname{diag}\left( \lambda_1(\by, \bv), \cdots, \lambda_{B_x}(\by, \bv) \right).
    \end{equation*}
    Define the truncated eigenvalues for any $\epsilon > 0$ by
    \begin{equation*}
        [\lambda_i(\by, \bv)]_{\varepsilon} 
        = \lambda_i(\by, \bv)\, \mathbb{I}\{\lambda_i(\by, \bv) < P\} 
        + \frac{P}{1 + P\delta_\varepsilon(\by, \bv)}\, \mathbb{I}\{\lambda_i(\by, \bv) \geq P\},
    \end{equation*}
    where $\delta_\varepsilon(\by, \bv) = 0$ if $\lambda_1(\by, \bv) < P$, and otherwise $\delta_\varepsilon(\by, \bv)$ is uniquely determined as the solution to
    \begin{equation} 
        \log(1 + P \delta) - \delta \lambda_1(\by, \bv) = - \varepsilon.
        \label{eq:del_eps_main}
    \end{equation}
    Then the following statements hold:
    \begin{itemize}
        \item The optimal value of the GMI  in Proposition~\ref{prop:GMI} is given by 
        \begin{equation} \label{eq:IGMI}
           I_{\mathrm{GMI,opt}}= \mathbb{E}_{\bY, \bV} \left[ \frac{1}{B_x} \sum_{i=1}^{B_x} \psi \left(\lambda_i(\bY, \bV) \right) + \frac{1}{B_x P} \|\bm{\mu}(\bY, \bV)\|_2^2\right],
        \end{equation}
        where the function $\psi:(0,\infty) \to \mathbb{R}$ is given by:
        \begin{equation} 
            \psi(q)=
            \begin{cases}
            \displaystyle \log\!\frac{P}{q} - 1 + \frac{q}{P}, & 0<q<P,\\[6pt]
             0, & q\ge P.
            \end{cases}
            \label{eq:psi_thm1}
        \end{equation}
        \item With the Vec-GNNDR employing the processing function $\bm{g}^\varepsilon$ and the scaling function $\bm{f}^\varepsilon$ defined by
        \begin{equation} \label{eq:opt_g_eps}
            \bg^\varepsilon(\by, \bv) = \mathrm{diag}\left( \sqrt{\frac{P}{(P - [\lambda_i(\by, \bv)]_\varepsilon) [\lambda_i(\by, \bv)]_\varepsilon }}\right) \bW^H(\by, \bv) \bm{\mu}(\by, \bv),
        \end{equation}
        \begin{equation} \label{eq:opt_f_eps}
            \bm{f}^\varepsilon(\by, \bv) = \mathrm{diag}\left( \sqrt{\frac{P - [\lambda_i(\by, \bv)]_\varepsilon}{P [\lambda_i(\by, \bv)]_\varepsilon}}\right) \bW^H(\by, \bv),
        \end{equation}
        the  corresponding GMI is at least $I_{\mathrm{GMI,opt}} - \varepsilon$.
    \end{itemize}

\end{theorem}

\begin{proof}
    Setting $\theta=-1$ in \eqref{eq:I_GMI_gf} and applying the law of total expectation, the first term expands as
    \begin{equation*}
    \begin{aligned}
        &  \mathbb{E} \left[\|\bg(\bY, \bV) - \bm{f}(\bY, \bV) \bX\|_2^{2} \right] = \mathbb{E}_{\bY,\bV} \left[ \mathbb{E}_{\bX|\bY,\bV} \left[\|\bg(\bY, \bV) - \bm{f}(\bY, \bV) \bX\|_2^{2} \right] \right] \\
        & \quad \quad \quad \quad \quad = \sum_{i=1}^{B_x} \EE_{\bY, \bV} \left| g_i(\bY, \bV)\right|^2 + \sum_{i=1}^{B_x} \EE_{\bY,\bV} \left[ \sigma_i^2(\bY,\bV) \EE_{\bX|\bY,\bV} \left[ \bX^{H} \bw_{i}(\bY, \bV) \bw_{i}^H(\bY,\bV) \bX \right] \right]\\
        & \quad \quad \quad \quad \quad - 2 \sum_{i=1}^{B_x} \EE_{\bY,\bV} \left[\sigma_i(\bY,\bV)\Re\left(g_{i}(\bY,\bV)^H \bw_i^H (\bY,\bV)\bm{\mu}(\bY, \bV)\right) \right].
    \end{aligned}
\end{equation*}
Hence, we can express $I_{\text{GMI}, \bg, \bm{f}}$ as 
\begin{equation} 
    \begin{aligned}
         &\quad \max_{\bg, \bm{f}} I_{\text{GMI}, \bg, \bm{f}} 
        = B_x^{-1} \mathbb{E}_{\bY, \bV} \sum_{i=1}^{B_x} \bigg[ \max_{\bg, \bm{f}} \Big \{ 
        \left[ -|g_i(\bY, \bV)|^2 - \sigma_i^2(\bY, \bV) \mathbb{E}_{\bX|\bY, \bV} \left[ \bX^{H} \bw_{i}(\bY, \bV) \bw_{i}^H(\bY, \bV) \bX \right] \right] \\
        & \quad \quad \quad \quad \quad + 2\sigma_i(\bY,\bV)\Re\left(g_{i}(\bY,\bV)^H \bw_i^H(\bY,\bV) \bm{\mu}(\bY, \bV)\right) + \frac{|g_i(\bY, \bV)|^{2}}{1+ P \sigma_i^2(\bY, \bV) } + \log (1 + P \sigma^2_i(\bY,\bV)) \Big \} \bigg ].
    \end{aligned}
    \label{eq:max_GMI_gf}
\end{equation}
For given $\by, \bv$, we now explore the optimal values of $\{g_i(\by,\bv), \sigma_i(\by,\bv), \bw_i({\by,\bv})\}_{i=1,\cdots, B_x}$. To proceed, we first optimize the phase of $g_i(\by,\bv)$ by aligning it with $\bw_i^H(\by,\bv)\bm{\mu}(\by,\bv)$:
\begin{equation} \label{eq:argument}
    \text{Arg}(g_i(\by,\bv)) = \text{Arg}(\bw_i^H(\by,\bv) \bm{\mu}(\by,\bv)),
\end{equation}
so that the third term in \eqref{eq:max_GMI_gf} becomes
\begin{equation*}
    2 \sigma_i(\by,\bv) |g_i(\by,\bv)| \left|\bw_i^H(\by,\bv) \bm{\mu}(\by,\bv) \right|.
\end{equation*}
Next, to optimize $|g_i(\by, \bv)|$, we face a quadratic maximization problem. The optimal value of $|g_i(\by, \bv)|$ is given by 
\begin{equation} \label{def:opti_norm_g}
    |g_i(\by, \bv)| = \frac{1+P\sigma_i^2(\by,\bv)}{P \sigma_i(\by,\bv)} \left| \bw_i^H(\by,\bv) \bm{\mu}(\by,\bv)\right|.
\end{equation}
Thus, the maximization problem for GMI reduces to the following form:
\begin{equation} \label{eq:max_2}
    \begin{aligned}
        \max_{\{\sigma_i, \bw_i\}_{i=1}^{B_x}} \frac{1}{B_x}\sum_{i=1}^{B_x} \bigg \{ & \log(1+P\sigma_i^2(\by, \bv)) + \frac{1+P\sigma_i^2(\by,\bv)}{P} \left | \bw_i^H(\by,\bv) \bm{\mu}(\by, \bv) \right|^2  \\
        & - \sigma_i^2(\by,\bv) \EE [\bX^H \bw_i(\by,\bv) \bw^H_i(\by,\bv) \bX |\by, \bv]\bigg \}.
    \end{aligned}
\end{equation}
With the identity:
\begin{equation*}
    \begin{aligned}
        \mathbb{E}[\bX^H \bw_i(\by, \bv) \bw_i^H(\by, \bv) \bX \mid \by, \bv] &= \bw_i^H(\by, \bv) \mathbb{E}[\bX^H \bX \mid \by, \bv] \bw_i(\by, \bv) \\
        &  = \bw_i^H(\by, \bv) \left[ \bm{\Sigma}(\by, \bv) + \bm{\mu}(\by, \bv) \bm{\mu}^H(\by, \bv)\right] \bw_i(\by, \bv),
    \end{aligned}
\end{equation*}
we can rewrite the optimization problem in \eqref{eq:max_2} as follows:
\begin{equation} \label{eq:max_3}
    \max_{\{\sigma_i, \bw_i\}_{i=1}^{B_x}} \frac{1}{B_x} \sum_{i=1}^{B_x} \left\{ \log(1+P \sigma_i^2(\by,\bv)) - \sigma_i^2(\by, \bv)\bw_i^H(\by, \bv)  \bm{\Sigma}(\by, \bv) \bw_i(\by, \bv)\right\} + \frac{1}{B_xP} \|\bm{\mu(\by, \bv)}\|_2^2,
\end{equation}
where we have used the identity $\sum_{i=1}^{B_x} \bw_i(\by, \bv) \bw_i^H(\by, \bv) = \bI_{B_x}$. Next, we determine the optimal value of $\sigma_i(\by,\bv)$. 
Differentiating the objective function in \eqref{eq:max_3} with respect to $\sigma_i^2(\by,\bv)$ and setting the derivative equal to zero yields
\begin{equation*}
    \frac{P}{1+P \sigma^2_i(\by, \bv)} - \bw_i^H(\by, \bv)\bm{\Sigma}(\by, \bv) \bw_i(\by, \bv)= 0,
\end{equation*}
whence the stationary point
\begin{equation} 
    \sigma_i^2(\by, \bv) = \frac{1}{\bw_i^H(\by, \bv)\bm{\Sigma}(\by, \bv) \bw_i(\by, \bv)} - \frac{1}{P}.
    \label{eq:sigma_solu}
\end{equation}
We must enforce the constraint $\sigma_i^2(\by,\bv) \geq 0$, but a naive truncation at zero when $\bw_i^H(\by, \bv) \bm{\Sigma}(\by, \bv) \bw_i(\by, \bv) \geq P$
would result in a loss of GMI. 
Specifically, setting $\sigma_i(\by,\bv)=0$ directly in \eqref{eq:max_GMI_gf} forces the GMI contribution of index $i$ to vanish. 
In contrast, taking the limit $\sigma_i^2(\by,\bv) \to 0$ in \eqref{eq:max_2} yields a non-negative contribution of  $P^{-1} |\bw_i^{H}(\by, \bv) \bm{\mu}(\by,\bv)|^2$. To capture this contribution, we truncate $\sigma_i^2(\by,\bv)$ at a sufficiently small $\delta(\by,\bv) > 0$ and then let $\delta(\by,\bv)\to 0$.

To choose an appropriate $\delta(\by, \bv)$, we consider the following function
\begin{equation*}
    \varphi_{\by,\bv}(\sigma^2) 
    = \log\!\big(1 + P\sigma^2\big) - \sigma^2 \lambda_{\max}(\by,\bv),
\end{equation*}
where $\lambda_{\max}(\by,\bv)$ denotes the largest eigenvalue of $\bm{\Sigma}(\by,\bv)$. 
By the Rayleigh quotient (Lemma~\ref{lem:ray_quoti}), 
if there exists $i$ with $\bw_i^H(\by,\bv)\bm{\Sigma}(\by,\bv)\bw_i(\by,\bv) \geq P$, then $\lambda_{\max}(\by,\bv) \geq P$, and hence $\varphi_{\by,\bv}$ is monotonically decreasing on $(0,\infty)$. 
Therefore, for any $\varepsilon > 0$, there exists $\delta_\varepsilon(\by,\bv) > 0$ such that 
\begin{equation*}
    \varphi_{\by,\bv}\big(\delta_\varepsilon(\by,\bv)\big) = - \varepsilon.
\end{equation*}
Using the Rayleigh quotient (Lemma~\ref{lem:ray_quoti}) again, we have 
\begin{equation*}
    \bw_i^H(\by,\bv)\bm{\Sigma}(\by,\bv)\bw_i(\by,\bv) \leq \lambda_{\max}(\bm{\Sigma}(\by, \bv)).
\end{equation*}
Consequently, we obtain
\begin{equation*}
    \log(1+P(\delta_\varepsilon(\by,\bv))) - \delta_\varepsilon(\by,\bv) \bw_i^H(\by,\bv)\bm{\Sigma}(\by,\bv)\bw_i(\by,\bv) \geq -\varepsilon,
\end{equation*}
which implies that the overall GMI loss induced by truncation at $\delta_\varepsilon(\by,\bv)$ is bounded above by $\varepsilon$. We thus truncate $\sigma_i^2(\by, \bv)$ as follows:
\begin{equation}
    \left[ \sigma_i^2(\by, \bv)\right]_\varepsilon = \sigma_i^2(\by, \bv) \mathbb{I} \{\sigma_i^2(\by, \bv)>0\} + \delta_\varepsilon(\by,\bv) \mathbb{I} \{\sigma_i^2(\by, \bv)\leq0\},
    \label{eq:sigma_truncated}
\end{equation}
where $\sigma^2_i(\by,\bv)$ is given by \eqref{eq:sigma_solu}.
Plugging \eqref{eq:sigma_truncated} into \eqref{eq:max_3} and letting $\varepsilon \to 0$, the maximization problem for GMI is transformed into:
\begin{equation}
    \sup_{\bg,\bm{f}}I_{\mathrm{GMI}, \bg, \bm{f}} = \mathbb{E}_{\bY, \bV} \left\{ 
    \max_{\{\bw_i\}_{i=1}^{B_x}} \frac{1}{B_x} \sum_{i=1}^{B_x} \psi(\bw_i^H(\bY, \bV)\bm{\Sigma}(\bY, \bV) \bw_i(\bY, \bV)) + \frac{1}{B_x P} \|\bm{\mu}(\bY, \bV)\|_2^2 \right\}, 
\end{equation}
where $\psi$ is given by \eqref{eq:psi_thm1}.
Then by spectral majorization inequality (Lemma~\ref{lem:spec_maj_inequal}) and the convexity of $\psi$ (Lemma~\ref{lem:psi_convex}), we obtain 
\begin{equation*}
    I_{\mathrm{GMI,opt}} = \sup_{\bm{g},  \bm{f}} I_{\mathrm{GMI}, \bg, \bm{f}} = \mathbb{E}_{\bY, \bV} \left[ \frac{1}{B_x} \sum_{i=1}^{B_x} \psi \left(\lambda_i(\bY, \bV) \right) + \frac{1}{B_x P} \|\bm{\mu}(\bY, \bV)\|_2^2\right].
\end{equation*}
Equality is attained if $\bW(\by, \bv)$ diagonalizes $\bm{\Sigma}(\by, \bv)$, i.e.,
\begin{equation} \label{eq:w_opt_proof}
    \bW^H(\by,\bv) \bm{\Sigma}(\by, \bv) \bW(\by, \bv) = \text{diag}\left( \lambda_1(\by, \bv), \cdots, \lambda_{B_x}(\by, \bv) \right).
\end{equation}
Substituting \eqref{eq:w_opt_proof} into \eqref{eq:sigma_solu} and
\eqref{eq:sigma_truncated}, we obtain 
\begin{equation}
    [\sigma_i(\by, \bv)]_{\varepsilon} = \sqrt{\frac{P - [\lambda_i(\by, \bv)]_{\varepsilon}}{P [\lambda_i(\by, \bv)]_\varepsilon}},
    \label{eq:sigma_truncated_2}
\end{equation}
where $[\lambda_i(\by, \bv)]_{\varepsilon}$ is defined by
\begin{equation*}
    [\lambda_i(\by, \bv)]_{\varepsilon} = \lambda_i(\by, \bv) \mathbb{I} \{\lambda_i(\by, \bv) < P\} +  \frac{P}{1 +P \delta_\varepsilon(\by, \bv)}\mathbb{I} \{\lambda_i(\by, \bv) \geq P\}.
\end{equation*}
Next, substituting \eqref{eq:sigma_truncated_2} and \eqref{eq:w_opt_proof} into  \eqref{def:f}, we get 
\begin{equation*}
    \bm{f}^\varepsilon(\by, \bv) = \mathrm{diag}\left( \sqrt{\frac{P - [\lambda_i(\by, \bv)]_\varepsilon}{P [\lambda_i(\by, \bv)]_\varepsilon}}\right) \bW^H(\by, \bv).
\end{equation*}
Then, substituting \eqref{eq:sigma_truncated_2} into \eqref{def:opti_norm_g}, we obtain the optimal value of $|g^\varepsilon_i(\bv,\bv)|$: 
\begin{equation}
    |g_i^\varepsilon(\by, \bv)| = \sqrt{\frac{P}{(P - [\lambda_i(\by, \bv)]_\varepsilon) [\lambda_i(\by, \bv)]_\varepsilon }} |\bw_i^H(\by, \bv) \bm{\mu}(\by, \bv)|.
    \label{eq:opt_normal_g}
\end{equation}
Finally, combining \eqref{eq:argument} and \eqref{eq:opt_normal_g}, we get the optimal form for $\bg^\varepsilon(\by, \bv)$:
\begin{equation*}
    \bg^\varepsilon(\by, \bv) = \mathrm{diag}\left( \sqrt{\frac{P}{(P - [\lambda_i(\by, \bv)]_\varepsilon) [\lambda_i(\by, \bv)]_\varepsilon }}\right) \bW^H(\by, \bv) \bm{\mu}(\by, \bv),
\end{equation*}
which completes the proof.

\end{proof}

\begin{remark}[Technical distinction from scalar GNNDR]
The vector optimization involves an orthonormal-basis choice absent
from the scalar derivation in ~\cite{wang2022generalized}; scalar first-order conditions
alone do not resolve this coupled optimization. Spectral majorization
and convexity establish global optimality of the posterior-covariance
eigenbasis. The truncation argument further handles posterior
eigenvalues at or above $P$ without discarding conditional-mean
contributions, extending the analysis beyond the variance restriction
in~\cite{wang2022generalized}.
\end{remark}

\begin{remark}
    The optimal Vec-GNNDR can be interpreted as a two-stage procedure that combines principal component analysis (PCA) with GNNDR. In the first stage, PCA is applied to the conditional covariance matrix $\bm{\Sigma}(\by, \bv)$ to obtain eigenvectors $\{\bw_i(\by, \bv)\}_{i=1}^{B_x}$, ordered according to their associated eigenvalues. In the second stage, the input $\bx$ is projected onto each eigenspace spanned by $\bw_i(\by, \bv)$, and the GNNDR is then applied sequentially to the resulting projections $\bw_i(\by, \bv)^H \bx$. When the conditional variance along the $i$-th eigen-direction, given by $\mathrm{Var}[\bX \mid \by, \bv] = \lambda_i(\by, \bv)$, exceeds $P$, the eigenvalue is first truncated to $[\lambda_i(\by, \bv)]_{\varepsilon}$ prior to applying GNNDR.  Although $[\lambda_i(\by, \bv)]_{\varepsilon}$ is close to $P$, it cannot be directly set equal to $P$. Setting $\lambda_i(\by, \bv) = P$ would imply $\sigma_i(\by, \bv) = 0$, causing the scaling function $\bm{f}$ to project $\bx_l(m)$ onto an affine subspace of inherent dimension less than $B_x$. Such projection generally discards information relevant to decoding and may consequently reduce the GMI.
\end{remark}

The complexity of the Vec-GNNDR characterized in Theorem~\ref{thm:main_result} arises primarily from the truncation of the eigenvalues $\lambda_i(\by, \bv)$ whenever $\lambda_i(\by, \bv) \geq P$. In the special case where all conditional eigenvalues are strictly less than $P$, both the optimal Vec-GNNDR and the corresponding GMI admit simplified closed-form expressions.  
\begin{assumption}
    For any channel output $\by$ and receiver CSI $\bv$ obtained from the in-block memory channel~\eqref{eq:in-channel},  
    the conditional covariance matrix satisfies $0 \prec\operatorname{Cov} [\bX \mid \by, \bv] \prec P $ almost surely.
    \label{ass:cond_cov_P}
\end{assumption}
Under Assumption~\ref{ass:cond_cov_P}, the optimal Vec-GNNDR given in Theorem~\ref{thm:main_result} simplifies to the following result.
\begin{theorem}
     Consider the in-block memory channel described in Section~\ref{sec:set-up}, which satisfies Assumption~\ref{ass:cond_cov_P}. Let $\bm{\mu}(\by, \bv)$, $\bm{\Sigma}(\by, \bv)$, and $\lambda_i(\by, \bv)$ be defined as in Theorem~\ref{thm:main_result}. Then the processing and scaling functions of the optimal Vec-GNNDR are given by
    \begin{equation}
        \bg(\by, \bv) = \mathrm{diag}\left( \sqrt{\frac{P}{(P - \lambda_i(\by, \bv)) \lambda_i(\by, \bv) }}\right) \bW^H(\by, \bv) \bm{\mu}(\by, \bv),
        \label{eq:opt_g_leqP}
    \end{equation}
    \begin{equation}
        \bm{f}(\by, \bv) = \mathrm{diag}\left( \sqrt{\frac{P - \lambda_i(\by, \bv)}{P \lambda_i(\by, \bv)}}\right) \bW^H(\by, \bv),
        \label{eq:opt_f_leqP}
    \end{equation}
    and the corresponding optimal GMI admits the closed form
    \begin{equation*}
        I_{\mathrm{GMI,opt}} = \frac{1}{B_x} \mathbb{E}_{\bY, \bV}  \left[ \log \frac{P^{B_x}}{\det (\bm{\Sigma}(\bY, \bV))} \right].
    \end{equation*}
    \label{thm:main_result_2}
\end{theorem}
\begin{proof}
    Under Assumption~\ref{ass:cond_cov_P}, we have $\lambda_i(\by, \bv) < P$ for all indices $i$ and for any $(\by, \bv)$. Hence, no truncation is required, and $[\lambda_i(\by, \bv)]_\varepsilon = \lambda_i(\by, \bv)$. Substituting this into \eqref{eq:opt_g_eps} and \eqref{eq:opt_f_eps} given in Theorem~\ref{thm:main_result}, we obtain the optimal processing and scaling functions as given in \eqref{eq:opt_g_leqP} and \eqref{eq:opt_f_leqP}, respectively.  For the corresponding optimal GMI, we have
    \begin{equation}
        \label{eq:IGMI_ass3}
        \begin{aligned}
            I_{\mathrm{GMI,opt}} & = \mathbb{E}_{\bY,\bV} \left[ \frac{1}{B_x} \sum_{i=1}^{B_x} \left( \log \frac{P}{\lambda_i(\bY,\bV)} - 1 +\frac{\lambda_i(\bY, \bV)}{P}\right) + \frac{1}{B_x P} \|\bm{\mu}(\bY,\bV)\|_2^2\right] \\
            & = \mathbb{E}_{\bY, \bV} \left[\frac{1}{B_x}  \log \frac{P^{B_x}}{\det (\bm{\Sigma}(\bY, \bV))} + \frac{1}{B_x P} \mathrm{Tr}(\bm{\Sigma}(\bY, \bV)) + \frac{1}{B_x P}\|\bm{\mu}(\bY,\bV)\|_2^2 \right] - 1 \\
            & = \mathbb{E}_{\bY, \bV} \left[ \frac{1}{B_x}  \log \frac{P^{B_x}}{\det (\bm{\Sigma}(\bY, \bV))} + \frac{1}{B_x P} \mathrm{Tr}\left[ \bX \bX^H \mid \bY, \bV \right] \right] - 1 \\
            & = \frac{1}{B_x} \mathbb{E}_{\bY, \bV} \left[ \log \frac{P^{B_x}}{\det(\bm{\Sigma}(\bY, \bV))} \right],
        \end{aligned}
    \end{equation}
    where the last equality follows from the law of total expectation and the fact that  $\mathbb{E}[\bX \bX^H] = P \boldsymbol{I}_{B_x}$. This completes the proof.
\end{proof}



\begin{remark}
    The derivation of the optimal GNNDR in \cite{wang2022generalized}
    can be regarded as a special case with $B_x=1$ under
    Assumption~\ref{ass:cond_cov_P}, i.e., when the conditional
    covariance is strictly smaller than $P$. This condition, however,
    need not hold for general channels. When the conditional covariance
    reaches or exceeds $P$, the corresponding boundary case requires
    additional treatment, for which Theorem~\ref{thm:main_result}
    employs an $\varepsilon$-truncation construction. In general, this
    construction provides a sequence of Vec-GNNDR metrics whose GMIs
    approach the supremum as $\varepsilon \to 0$, while exact attainment
    is not guaranteed in all such boundary cases.
    Related boundary issues may also arise in extensions of GNNDR to
    more general input constellations, such as \cite{pang2025generalized},
    and may require additional treatment beyond the interior-case
    characterization.

    Indeed, similar to the derivation in
    Theorem~\ref{thm:main_result_2}, the optimal GMI
    in~\eqref{eq:IGMI} can be rewritten as
    \begin{equation}
        I_{\mathrm{GMI,opt}}
        =
        \frac{1}{B_x}
        \mathbb{E}_{\bY,\bV}
        \left[
        \log
        \frac{P^{B_x}}
        {\det(\bm{\Sigma}(\bY,\bV))}
        \right]
        -
        \frac{1}{B_x}
        \mathbb{E}_{\bY,\bV}
        \left[
        \sum_{i=1}^{B_x}
        \left(
        \log\frac{P}{\lambda_i(\bY,\bV)}
        +
        \frac{\lambda_i(\bY,\bV)}{P}
        -1
        \right)
        \mathbb{I}\{\lambda_i(\bY,\bV)\geq P\}
        \right],
        \label{eq:equival_gmi_opt}
    \end{equation}
    where the second term is a nonnegative correction relative to the
    expression in~\eqref{eq:IGMI_ass3}. The correction is strictly
    positive whenever $\lambda_i(\bY,\bV)>P$ for some $i$ on a set of
    positive probability.
\end{remark}

\subsection{Joint Codebook-Metric Design Perspective} \label{sec:cb_opt}

In this subsection, we develop a joint design perspective for the
Gaussian codebook ensemble and the decoding metric based on the optimal
Vec-GNNDR characterization. We first show that, for each fixed
$\bm{\Sigma}\succ0$, whitening extends the closed-form GMI
characterization and the associated metric construction to the colored
Gaussian input setting. This reduces the joint codebook--metric design
to an optimization over the input covariance. Since optimization over
the full positive semidefinite covariance domain is generally nonconvex
and analytically intractable, we focus on the diagonal covariance family
and derive first-order self-consistent optimality conditions.

\begin{corollary}
\label{cor:colored-codebook}
Consider the in-block memory channel in Section~II with independent
Gaussian codebook blocks
$\mathbf{x}_l(m)\sim\mathcal{CN}(\mathbf{0},P\boldsymbol{\Sigma})$,
where $P>0$, $\boldsymbol{\Sigma}\succ0$, and
$\operatorname{Tr}(\boldsymbol{\Sigma})=B_x$.
Suppose Assumption~\ref{ass:psd_cond_cov} holds under the induced joint distribution.
Let $\boldsymbol{\mu}_{\mathrm{color}}(\mathbf{y},\mathbf{v})$
and $\boldsymbol{\Sigma}_{\mathrm{color}}(\mathbf{y},\mathbf{v})$
denote the corresponding posterior mean and covariance, and define
\begin{equation}
\begin{aligned}
\hat{\boldsymbol{\mu}}(\mathbf{y},\mathbf{v})
=\boldsymbol{\Sigma}^{-1/2}
  \boldsymbol{\mu}_{\mathrm{color}}(\mathbf{y},\mathbf{v}), \qquad 
\hat{\boldsymbol{\Sigma}}(\mathbf{y},\mathbf{v})
=\boldsymbol{\Sigma}^{-1/2}
  \boldsymbol{\Sigma}_{\mathrm{color}}(\mathbf{y},\mathbf{v})
  \boldsymbol{\Sigma}^{-1/2}.
\end{aligned}
\label{eq:cb-whitened-moments}
\end{equation}
Then the supremum of the GMI over admissible Vec-GNNDR metrics is
\begin{equation}
I_{\mathrm{GMI,opt}}(\boldsymbol{\Sigma})
=\frac{1}{B_x}\mathbb{E}_{\mathbf{Y},\mathbf{V}}\!\left[
\operatorname{Tr}\psi\bigl(
\hat{\boldsymbol{\Sigma}}(\mathbf{Y},\mathbf{V})\bigr)
+\frac{\|\hat{\boldsymbol{\mu}}(\mathbf{Y},\mathbf{V})\|_2^2}{P}
\right],
\label{eq:cb-colored-gmi}
\end{equation}
where $\psi$ is defined in Theorem~\ref{thm:main_result}.
The optimal pair
$(\hat{\bg}^{\varepsilon},\hat{\bm{f}}^{\varepsilon})$
for the original colored codebook is obtained from the admissible
construction in Theorem~\ref{thm:main_result} by replacing the posterior
moments with $\hat{\boldsymbol{\mu}}$ and
$\hat{\boldsymbol{\Sigma}}$, respectively, and additionally
right-multiplying the resulting scaling matrix $\bm{f}(\by,\bv)$ by
$\boldsymbol{\Sigma}^{-1/2}$.
\end{corollary}


\begin{proof}
As discussed in Remark~\ref{remark:codebook}, the communication system is transformed into
an equivalent system with white Gaussian codebook blocks
$\hat{\mathbf{x}}_l(m)\sim
\mathcal{CN}(\mathbf{0},P\mathbf{I}_{B_x})$
via the whitening transformation
$\hat{\mathbf{x}}_l(m)
=\boldsymbol{\Sigma}^{-1/2}\mathbf{x}_l(m)$.
The induced channel transition law is
\[
p_{\mathbf{Y}\mid\hat{\mathbf{X}},S}
(\mathbf{y}\mid\hat{\mathbf{x}},s)
=
p_{\mathbf{Y}\mid\mathbf{X},S}
(\mathbf{y}\mid
\boldsymbol{\Sigma}^{1/2}\hat{\mathbf{x}},s).
\]
Let
\[
p_{\hat{\mathbf{X}}}(\hat{\mathbf{x}})
=\frac{1}{(\pi P)^{B_x}}
\exp\!\left(-\frac{
\hat{\mathbf{x}}^H\hat{\mathbf{x}}}{P}\right)
\]
denote the density of the whitened input, and define the effective
likelihood
\[
\mathcal{L}(\mathbf{x}\mid\mathbf{y},\mathbf{v})
=
\int p_{\mathbf{Y}\mid\mathbf{X},S}
(\mathbf{y}\mid\mathbf{x},s)
p_{S,\mathbf{V}}(s,\mathbf{v})\,ds.
\]
By Bayes' rule,
\[
\mathbb{E}[\hat{\mathbf{X}}\mid\mathbf{y},\mathbf{v}]
=
\frac{
\int \hat{\mathbf{x}}\,
p_{\hat{\mathbf{X}}}(\hat{\mathbf{x}})
\mathcal{L}(\boldsymbol{\Sigma}^{1/2}\hat{\mathbf{x}}
\mid\mathbf{y},\mathbf{v})\,d\hat{\mathbf{x}}
}{
\int p_{\hat{\mathbf{X}}}(\hat{\mathbf{x}})
\mathcal{L}(\boldsymbol{\Sigma}^{1/2}\hat{\mathbf{x}}
\mid\mathbf{y},\mathbf{v})\,d\hat{\mathbf{x}}
}.
\]
Applying the change of variables
$\mathbf{u}=\boldsymbol{\Sigma}^{1/2}\hat{\mathbf{x}}$,
the complex volume element satisfies
$d\hat{\mathbf{x}}
=\det(\boldsymbol{\Sigma})^{-1}d\mathbf{u}$.
Consequently,
\begin{align*}
\mathbb{E}[\hat{\mathbf{X}}\mid\mathbf{y},\mathbf{v}]
&=\boldsymbol{\Sigma}^{-1/2}
\frac{
\int \mathbf{u}\,p_{\mathrm{color}}(\mathbf{u})
\mathcal{L}(\mathbf{u}\mid\mathbf{y},\mathbf{v})\,d\mathbf{u}
}{
\int p_{\mathrm{color}}(\mathbf{u})
\mathcal{L}(\mathbf{u}\mid\mathbf{y},\mathbf{v})\,d\mathbf{u}
}=\boldsymbol{\Sigma}^{-1/2}
\boldsymbol{\mu}_{\mathrm{color}}(\mathbf{y},\mathbf{v})
=\hat{\boldsymbol{\mu}}(\mathbf{y},\mathbf{v}),
\end{align*}
where
\[
p_{\mathrm{color}}(\mathbf{u})
=
\frac{1}{(\pi P)^{B_x}\det(\boldsymbol{\Sigma})}
\exp\!\left(
-\frac{\mathbf{u}^H\boldsymbol{\Sigma}^{-1}\mathbf{u}}{P}
\right)
\]
is the density of $\mathcal{CN}(\mathbf{0},P\boldsymbol{\Sigma})$.
Similarly, applying the same change of variables to the conditional
second moment yields
\begin{align*}
\operatorname{Cov}
[\hat{\mathbf{X}}\mid\mathbf{y},\mathbf{v}]
&=\boldsymbol{\Sigma}^{-1/2}
\left(
\mathbb{E}[\mathbf{X}\mathbf{X}^H\mid\mathbf{y},\mathbf{v}]
-\boldsymbol{\mu}_{\mathrm{color}}(\mathbf{y},\mathbf{v})
 \boldsymbol{\mu}_{\mathrm{color}}^H(\mathbf{y},\mathbf{v})
\right)
\boldsymbol{\Sigma}^{-1/2}\\
&=\boldsymbol{\Sigma}^{-1/2}
\boldsymbol{\Sigma}_{\mathrm{color}}(\mathbf{y},\mathbf{v})
\boldsymbol{\Sigma}^{-1/2}
=\hat{\boldsymbol{\Sigma}}(\mathbf{y},\mathbf{v}).
\end{align*}
Since $\boldsymbol{\Sigma}\succ0$, Assumption~2 implies
$\hat{\boldsymbol{\Sigma}}(\mathbf{Y},\mathbf{V})\succ0$
almost surely. Theorem~\ref{thm:main_result} therefore applies to
the equivalent white-input channel. To relate the decoding metrics, let
$(\hat{\bg}_{\mathrm{w}},\hat{\bm{f}}_{\mathrm{w}})$
be any admissible processing--scaling pair for the equivalent channel,
where the subscript $\mathrm{w}$ distinguishes the white-input metric,
and set
\[
\bg(\mathbf{y},\mathbf{v})
=\hat{\bg}_{\mathrm{w}}(\mathbf{y},\mathbf{v}),
\qquad
\bm{f}(\mathbf{y},\mathbf{v})
=\hat{\bm{f}}_{\mathrm{w}}(\mathbf{y},\mathbf{v})
\boldsymbol{\Sigma}^{-1/2}.
\]
For every original codeword $\mathbf{x}$ and its whitened counterpart
$\hat{\mathbf{x}}
=\boldsymbol{\Sigma}^{-1/2}\mathbf{x}$, we have
\[
\|\bg(\mathbf{y},\mathbf{v})
-\bm{f}(\mathbf{y},\mathbf{v})\mathbf{x}\|_2^2
=
\|\hat{\bg}_{\mathrm{w}}(\mathbf{y},\mathbf{v})
-\hat{\bm{f}}_{\mathrm{w}}(\mathbf{y},\mathbf{v})
\hat{\mathbf{x}}\|_2^2.
\]
This identity holds for both the transmitted codeword and every
independent competing codeword, so the GMI is unchanged.
Moreover, the correspondence is bijective, with inverse
$\hat{\bg}_{\mathrm{w}}=\bg$ and
$\hat{\bm{f}}_{\mathrm{w}}=\bm{f}\boldsymbol{\Sigma}^{1/2}$,
and preserves admissibility because $\boldsymbol{\Sigma}$ is fixed
and nonsingular. Thus, the suprema over the two metric classes coincide.
Applying Theorem~\ref{thm:main_result} with posterior moments
$\hat{\boldsymbol{\mu}}$ and $\hat{\boldsymbol{\Sigma}}$
gives \eqref{eq:cb-colored-gmi}.

Finally, when this supremum is finite, let
$(\hat{\bg}_{\mathrm{w}}^{\varepsilon},
\hat{\bm{f}}_{\mathrm{w}}^{\varepsilon})$
be the admissible $\varepsilon$-near-optimal pair obtained from
Theorem~\ref{thm:main_result} using these posterior moments.
The corresponding pair for the original colored codebook is
\[
\hat{\bg}^{\varepsilon}(\mathbf{y},\mathbf{v})
=\hat{\bg}_{\mathrm{w}}^{\varepsilon}(\mathbf{y},\mathbf{v}),
\qquad
\hat{\bm{f}}^{\varepsilon}(\mathbf{y},\mathbf{v})
=\hat{\bm{f}}_{\mathrm{w}}^{\varepsilon}(\mathbf{y},\mathbf{v})
\boldsymbol{\Sigma}^{-1/2},
\]
and has the same GMI, at least
$I_{\mathrm{GMI,opt}}-\varepsilon$.
This establishes the stated construction and completes the proof.
\end{proof}

We now consider the input-covariance optimization induced by
Corollary~\ref{cor:colored-codebook}. Rather than optimizing over all
positive semidefinite matrices satisfying
$\operatorname{Tr}(\bm{\Sigma})=B_x$, we restrict attention to diagonal
covariance matrices
$\bm{\Sigma}=\operatorname{diag}(\lambda_1,\ldots,\lambda_{B_x})$,
with $\lambda_i\geq0$ and $\sum_{i=1}^{B_x}\lambda_i=B_x$.
This structural restriction corresponds to allocating power
$P\lambda_i$ across the input coordinates and is consistent with
subcarrier power allocation in multicarrier communication
systems~\cite{weinstein1971data,wang2000wireless}.
The first-order analysis below concerns interior allocations with
$\lambda_i>0$ for all $i$, for which the whitening transformation in
Corollary~\ref{cor:colored-codebook} is well defined.

\begin{theorem}
\label{thm:diagonal-codebook}
Consider a Gaussian random codebook with an input covariance matrix $\bm{\Sigma} = \operatorname{diag}(\lambda_1, \dots, \lambda_{B_x}) \succ 0$ subject to the power constraint $\sum_{i=1}^{B_x} \lambda_i = B_{x}$. 
For transmission over the in-block memory channel described in Section \ref{sec:set-up} satisfying Assumption \ref{ass:psd_cond_cov}, any interior stationary point of the diagonal covariance optimization problem associated with the optimal Vec-GNNDR must satisfy the following system of self-consistent equations:
\begin{subequations}
\label{eq:cb-self-consistency}
\begin{align}
\lambda_i
&=
\frac{
P+\mathbb{E}_{\mathbf{Y},\mathbf{V}}\!\left[
\left[
\bigl(\hat{\bm{\Sigma}}(\mathbf{Y},\mathbf{V})
-P\mathbf{I}_{B_x}\bigr)_+
\right]_{ii}
\right]
}{
P\bigl(\nu-\Delta_i(\bm{\Sigma})\bigr)
},
\quad i=1,\ldots,B_x,
\label{eq:cb-self-consistency-a}\\
\sum_{i=1}^{B_x}
&\frac{
P+\mathbb{E}_{\mathbf{Y},\mathbf{V}}\!\left[
\left[
\bigl(\hat{\bm{\Sigma}}(\mathbf{Y},\mathbf{V})
-P\mathbf{I}_{B_x}\bigr)_+
\right]_{ii}
\right]
}{
P\bigl(\nu-\Delta_i(\bm{\Sigma})\bigr)
}
=B_x,
\label{eq:cb-self-consistency-b}
\end{align}
\end{subequations}
where $\nu$ is a common real scalar satisfying
$\nu>\Delta_i(\bm{\Sigma})$ for all $i$, and
\begin{equation}
\begin{aligned}
\Delta_i(\bm{\Sigma})
=\mathbb{E}_{\mathbf{Y},\mathbf{V}}\Bigl[
\mathbb{E}_{\mathbf{X}\mid\mathbf{Y},\mathbf{V}}\!\left[
q(\hat{\mathbf{X}}\mid\mathbf{Y},\mathbf{V})
\hat{S}_i(\mathbf{X}\mid\mathbf{Y},\mathbf{V})
\right]+\operatorname{Tr}G\bigl(
\hat{\bm{\Sigma}}(\mathbf{Y},\mathbf{V})\bigr)
\mathbb{E}_{\mathbf{X}\mid\mathbf{Y},\mathbf{V}}
[S_i(\mathbf{X})]
\Bigr].
\end{aligned}
\label{eq:cb-delta}
\end{equation}
Here $\hat{\mathbf{X}}=\bm{\Sigma}^{-1/2}\mathbf{X}$,
and the score and centered score are respectively defined by
\begin{equation}
\begin{aligned}
S_i(\mathbf{x})
=
\frac{\partial}{\partial\lambda_i}
\log\mathcal{CN}(\mathbf{x};\mathbf{0},P\bm{\Sigma})
=
\frac{|x_i|^2}{P\lambda_i^2}-\frac{1}{\lambda_i}, \qquad
\hat{S}_i(\mathbf{x}\mid\mathbf{y},\mathbf{v})
=
S_i(\mathbf{x})
-\mathbb{E}_{\mathbf{X}\mid\mathbf{y},\mathbf{v}}
[S_i(\mathbf{X})].
\end{aligned}
\label{eq:cb-scores}
\end{equation}
The auxiliary functions are
\begin{equation}
G(t)=
\begin{cases}
-\log(t/P),&0<t<P,\\[1mm]
1-t/P,&t\geq P,
\end{cases}
\label{eq:cb-G}
\end{equation}
and
\begin{equation}
\begin{aligned}
q(\hat{\mathbf{x}}\mid\mathbf{y},\mathbf{v})
=
\bigl(\hat{\mathbf{x}}-\hat{\bm{\mu}}(\mathbf{y},\mathbf{v})\bigr)^H
G'\bigl(\hat{\bm{\Sigma}}(\mathbf{y},\mathbf{v})\bigr)
\bigl(\hat{\mathbf{x}}-\hat{\bm{\mu}}(\mathbf{y},\mathbf{v})\bigr),
\end{aligned}
\label{eq:cb-q}
\end{equation}
where $\hat{\bm{\mu}}$ and $\hat{\bm{\Sigma}}$ are given in
\eqref{eq:cb-whitened-moments}.
The matrix functions are defined spectrally.
All posterior moments and expectations are evaluated under the joint
law induced by the current covariance $\bm{\Sigma}$, with this
dependence suppressed in the notation.
\end{theorem}

\begin{proof}
Since $\hat{\mathbf{X}}\sim
\mathcal{CN}(\mathbf{0},P\mathbf{I}_{B_x})$,
the law of total expectation yields
\[
\mathbb{E}_{\mathbf{Y},\mathbf{V}}\!\left[
\operatorname{Tr}\hat{\bm{\Sigma}}(\mathbf{Y},\mathbf{V})
+\|\hat{\bm{\mu}}(\mathbf{Y},\mathbf{V})\|_2^2
\right]
=PB_x.
\]
Combining this identity with
Corollary~\ref{cor:colored-codebook} and
$G(t)=\psi(t)+1-t/P$, we obtain
\begin{equation}
\begin{aligned}
J(\bm{\Sigma})
:=B_x I_{\mathrm{GMI,opt}}=\mathbb{E}_{\mathbf{Y},\mathbf{V}}\!\left[
\operatorname{Tr}G\bigl(
\hat{\bm{\Sigma}}(\mathbf{Y},\mathbf{V})\bigr)
\right].
\end{aligned}
\label{eq:cb-objective}
\end{equation}
We regard the expression on the right-hand side as a function on the
positive diagonal family before imposing the trace constraint.
The covariance optimization therefore amounts to maximizing
$J(\bm{\Sigma})$ subject to $\lambda_i>0$ and
$\sum_{i=1}^{B_x}\lambda_i=B_x$.
Introduce the Lagrangian
\[
\mathcal{L}(\bm{\Sigma},\nu)
=
J(\bm{\Sigma})
-\nu\left(\sum_{i=1}^{B_x}\lambda_i-B_x\right),
\]
where $\nu$ is the multiplier associated with the power constraint.

The derivative calculation in
Appendix~\ref{app:cb_opt} gives
\begin{equation}
\begin{aligned}
\frac{\partial\mathcal{L}}{\partial\lambda_i}
=
\frac{
P+\mathbb{E}_{\mathbf{Y},\mathbf{V}}\!\left[
\left[
\bigl(\hat{\bm{\Sigma}}(\mathbf{Y},\mathbf{V})
-P\mathbf{I}_{B_x}\bigr)_+
\right]_{ii}
\right]
}{P\lambda_i}+\Delta_i(\bm{\Sigma})-\nu,
\end{aligned}
\label{eq:cb-final-gradient}
\end{equation}
where $\Delta_i$ is defined in \eqref{eq:cb-delta}.
At an interior stationary point, these derivatives vanish for all $i$.
Since $\lambda_i>0$ and the numerator of the first term is strictly
positive, we have $\nu>\Delta_i(\bm{\Sigma})$.
Rearranging \eqref{eq:cb-final-gradient} gives
\eqref{eq:cb-self-consistency-a}, and enforcing the power constraint
yields \eqref{eq:cb-self-consistency-b}.
This completes the proof.
\end{proof}

The first-order conditions above exhibit a structural affinity with the classical water-filling algorithm \cite{gallager1968information, cover1999elements}. Specifically, \eqref{eq:cb-self-consistency-a} reveals a thresholding mechanism: additional power is assigned according to the whitened covariance $\hat{\bm{\Sigma}}$ relative to the reference level $P$, together with the correction term $\Delta_i(\bm{\Sigma})$ induced by the dependence of the posterior distribution on the codebook covariance. This water-filling-like structure suggests iterative numerical schemes, although convergence and global optimality are not asserted here.

The main role of this subsection is to clarify how codebook design and metric design can be placed in a single GMI-based framework. The inner problem over Vec-GNNDR metrics is solved in closed form for each fixed Gaussian covariance, while the outer covariance problem remains difficult in general. The diagonal result above should therefore be interpreted as a first-order characterization for a tractable subfamily rather than a complete solution of the full joint optimization problem. Extending this framework to full covariance matrices and to more general codebook ensembles \cite{pang2025generalized} is left for future work.

\subsection{Case Study} \label{sec:Case}
In this subsection, we illustrate our results established in the previous subsection by recovering existing results for memoryless channels \cite{wang2022generalized} and colored Gaussian channels \cite{tse2005fundamentals} .
\subsubsection{Memoryless channel}
A memoryless channel is a special case of an in-block memory channel with no intra-block dependence; equivalently, it corresponds to block length $B_x=1$. Formally, the conditional joint probability density functions for such channels with codeword of length $N$ are given by:
\begin{equation} \label{eq:memoryless}
    \begin{aligned}
       P_{S, V}(s, v) &= \prod_{i=1}^N P_{S,V}(s_i, v_i),  \\
       P_{Y \mid X, S}\left(y \mid x, s\right) &= \prod_{i=1}^N P_{Y \mid X,S}(y_i \mid x_i, s_i).
    \end{aligned}
\end{equation}
Since $B_x = 1$, we use the element index $i$ instead of the block index $l$, and the processing and scaling functions $g$ and $f$ are scalar-valued.  
The following corollary characterizes the optimal $(g,f)$ pair and the corresponding $I_{\mathrm{GMI}, g, f}$.
\begin{corollary}[Memoryless channel]        \label{coral:memoryless}
    Consider the memoryless channel specified by the conditional distributions in~\eqref{eq:memoryless}.  
    Let the conditional mean and conditional variance be respectively defined as
    \begin{equation*}
        \mu(y,v) = \mathbb{E}[X \mid y, v], \qquad w(y,v) = \mathbb{E}[|X|^2|y, v] - |\mathbb{E}[X \mid y ,v]|^2.
    \end{equation*}
    Define the truncated variance as
    \begin{equation} \
        [w(y,v)]_{\varepsilon} = w(y, v) \mathbb{I}\{w(y, v)<P\} + \frac{P}{1+ P\delta_\varepsilon(y,v)} \mathbb{I} \{w(y, v) \geq P\},
        \label{eq:var_truncate}
    \end{equation}
    where $\delta_\varepsilon(y,v) = 0$ if $w(y,v) < P$, and otherwise $\delta_\varepsilon(y,v)$ is uniquely determined by the equation 
    \begin{equation*}
        \log(1+P \delta) - \delta w(y,v) = -\varepsilon.
    \end{equation*}
    Then the following statements hold:
    \begin{itemize}
        \item The optimal GMI is given by 
        \begin{equation*}
            I_{\mathrm{GMI,opt}} = \mathbb{E}_{Y, V} \left[ \psi(w(Y,V))+ \frac{1}{P}|\mu(Y,V)|^2 \right].
        \end{equation*}
        \item For the GNNDR employing the processing function $g^\varepsilon$ and scaling function $f^\varepsilon$ defined as 
        \begin{equation*}
            g^\varepsilon(y,v) = \sqrt{\frac{P}{(P-[w(y,v)]_{\varepsilon})[w(y,v)]_\varepsilon}} \mu(y,v) \qquad f^\varepsilon(y,v) = \sqrt{\frac{P - [w(y,v)]_\varepsilon}{P[w(y,v)]_\varepsilon}},
        \end{equation*}
        the corresponding GMI is at least $I_{\mathrm{GMI,opt}} - \varepsilon$.
    \end{itemize}
\end{corollary}
\begin{proof}
    See Appendix \ref{app:main}.
\end{proof}
If the conditional variance satisfies $0 < w(Y,V) < P$ almost surely, truncation as in~\eqref{eq:var_truncate} is unnecessary, and the optimal GMI reduces to
\begin{equation*}
    I_{\mathrm{GMI,opt}} = \mathbb{E}_{Y,V}\left[ \log\frac{P}{w(Y,V)}\right].
\end{equation*}
This expression is consistent with the known results for memoryless channels established in~\cite[Theorem~1]{wang2022generalized}.  
The only discrepancy lies in a common multiplicative factor appearing in ${f}(y,v)$ and $g(y,v)$, which does not affect the decoding rule.
\subsubsection{Additive colored Gaussian noise channel}
A widely encountered in-block memory channel is the additive colored Gaussian noise channel (ACGNC), defined as follows:
\begin{equation} \label{eq:colored_channel}
    \by_l = \bA \bx_l + \bm{n}_l,
\end{equation}
where $\bA \in \CC^{B_y \times B_x}$ is 
a block-constant channel matrix (assumed known at the receiver), and 
$\bm{n}_l \overset{\text{i.i.d.}}{\sim} \mathcal{C} \cN (0, \bm{\Sigma})$  denotes additive colored Gaussian noise. The following corollary characterizes the optimal processing function $\bg$, the scaling function $\bm{f}$, and the corresponding GMI.
\begin{corollary} \label{corollary:acgwc}
    Consider the ACGNC given in \eqref{eq:colored_channel} with $B_y = B_x$. The optimal processing function $\bg$ and scaling function $\bm{f}$ are respectively given by 
    \begin{equation*}
        \bm{f}(\by, \bv)  =  \bm{\Lambda} \bW^H, \quad 
       \bg(\by, \bv) = \bU^H \bm{\Sigma}^{-\frac{1}{2}} \by,
    \end{equation*}
    where $\bU, \bm{\Lambda}, \bW$ are given by the singular value decomposition of $\bm{M} = \bm{\Sigma}^{-\frac{1}{2}} \bA$ as $\bM = \bU \bm{\Lambda} \bW^H$ with $\bU \in \mathbb{U}(B_y), \bV \in \mathbb{U}(B_x)$, and $\bm{\Lambda} = \operatorname{diag}(\sigma_1, \cdots, \sigma_{B_x})$ satisfying $\sigma_1 \geq \sigma_2 \geq \cdots \geq \sigma_{B_x} \geq 0 $. The corresponding optimal GMI is given by 
    \begin{equation*}
    I_{\text{GMI}, \bg, \bm{f}} = B_x^{-1} \log \left[  \det (P\bA^H \bm{\Sigma}^{-1} \bA + \bI_{B_x})\right].
    \end{equation*}
\end{corollary}

\begin{proof}
    See Appendix \ref{app:main}.
\end{proof}
\begin{remark}
    With the optimal $\bg(\by, \bv)$ and $\bm{f}(\by, \bv)$ given in Corollary~\ref{corollary:acgwc}, the decoding rule recovers the ML decoder. Specifically, 
    \begin{equation*}
    \begin{aligned}
        \| \bg(\by, \bv) - \bm{f}(\by, \bv) \bx \|_2^2 &= \| \bU^H \bm{\Sigma}^{-\frac{1}{2}} \by - \bm{\Lambda} \bW^H \bx \|_2^2 \\
        &= \| \bm{\Sigma}^{-\frac{1}{2}} \by - \bU \bm{\Lambda} \bW^H \bx \|_2^2 \\
        &= \| \bm{\Sigma}^{-\frac{1}{2}} \by - \bm{\Sigma}^{-\frac{1}{2}} \bA \bx \|_2^2 \\ 
        &= (\by - \bA \bx)^H \bm{\Sigma}^{-1} (\by - \bA \bx).
    \end{aligned}
\end{equation*}
\end{remark}
\begin{remark}
    If $\bm{\Sigma} = N_0 \bI_{B_x}$, then 
    \begin{equation*}
        I_{\mathrm{GMI,opt}} = B_x^{-1} \log [\det (\bI_{B_x} + P N_0^{-1}\bA^H \bA)],
    \end{equation*}
    which coincides coincides with the Gaussian-input mutual information under $\bX \sim \mathcal{CN}(\mathbf{0}, P \mathbf{I}_{B_x})$ in \cite[Chapter 8]{tse2005fundamentals}.  In this case, $\boldsymbol{\Lambda}$ and $\bW$ are the singular value matrix and the right singular matrix of $\bA$, respectively, matching the structure of the optimal capacity-achieving decoder. 
\end{remark}

\section{Vec-GNNDR with Restricted Forms} \label{sec:restricted_vec_GNNDR}
In this section, we investigate several constrained variants of Vec-GNNDR in which the processing and scaling functions take restrictive forms. While such suboptimal Vec-GNNDRs may result in a reduced GMI, they simultaneously incur less computational complexity and provide valuable insights into the structural properties and theoretical understanding of Vec-GNNDR.
\subsection{Constant Scalar Scaling Function}
\label{subsec:cssf}
In this scenario, the Vec-GNNDR adopts the following formulation:
\begin{equation} \label{eq:Vec-GNNDR-cssf}
    \hat{m} = \arg \min_{m \in \mathcal{M}} \sum_{l=1}^L  \left\|\bg(\by_l, \bv_l) - \alpha \bx_l(m) \right\|^2_2,
\end{equation}
where $\alpha \in \mathbb{C}$ is an adjustable scalar. This implies that the scaling function is given by $\bm{f}(\by, \bv) = \alpha \bI_{B_x}$, which, together with $\bg(\by,\bv)$, aims to maximize the GMI. The following proposition establishes the optimal processing and scaling functions, along with the corresponding maximized GMI under the decoding rule \eqref{eq:Vec-GNNDR-cssf}.
\begin{proposition}[Constant scalar scaling function]
   Consider the in-block memory channel described in Section \ref{sec:set-up}, which satisfies Assumption \ref{ass:psd_cond_cov}.
   Let $\bm{\mu}(\by, \bv)$ and $\bm{\Sigma}(\by, \bv)$ be defined as in Theorem~\ref{thm:main_result}, and define the expected conditional covariance matrix as
   \begin{equation*}
       \bm{\Sigma} = \mathbb{E}_{\bY, \bV} \left[ \bm{\Sigma}(\bY, \bV)\right] = P\boldsymbol{I}_{B_x} - \mathbb{E}_{\bY, \bV} \left[ \bm{\mu}(\bY,\bV) \bm{\mu}^H(\bY,\bV)\right],
   \end{equation*}
   with normalized trace $T = B_x^{-1} \operatorname{Tr}(\bm{\Sigma})$. Under the decoding rule  \eqref{eq:Vec-GNNDR-cssf}, the optimal GMI is 
   \begin{equation}
        I_{\mathrm{GMI,cssf}} = \log \frac{P}{T},
        \label{eq:GMI_cssf}
    \end{equation}
    which is attained by the optimal processing function $\bg$ and scaling constant $\alpha$ given by
    \begin{equation*}
        \bg(\by, \bv) = \sqrt{\frac{P}{T(P-T)}} \bm{\mu}(\by, \bv), \qquad \alpha = \sqrt{\frac{P-T}{PT}},
    \end{equation*}
    whenever $0 < T < P$. In the boundary case $T = P$, the optimal solution degenerates to $\bg(\by, \bv) = \boldsymbol{0}_{B_x}$ and $\alpha=0$. 
    \label{prop:Vec-GNNDR-cssf}
\end{proposition}
\begin{proof}
    See Appendix \ref{app:restricted_vec_GNNDR}.
\end{proof}
\begin{remark}
   Truncation of $T$ is unnecessary since $T = B_x^{-1}\operatorname{Tr}(\bm{\Sigma}) \leq P$.  In the boundary case $T = P$, ~\eqref{eq:GMI_cssf} yields an optimal GMI of zero, indicating that no positive rate can be certified using the Vec-GNNDR with a constant-scalar scaling function in the in-block memory channel.  Consequently, we set $\bg(\by, \bv) = \boldsymbol{0}_{B_x}$ and $\alpha = 0$.
\end{remark}
Compared with the optimal processing and scaling functions in \eqref{eq:opt_g_eps} and \eqref{eq:opt_f_eps}, each $\lambda_i(\by, \bv)$ is replaced by the block-averaged posterior MSE $T$, given by
\begin{equation*}
    T = B_x^{-1} \mathrm{Tr}  \left( \mathbb{E}_{\bY, \bV} \left[\bm{\Sigma}(\bY, \bV)\right] \right) = B_x^{-1} \mathbb{E}_{\bY, \bV} \left[ \left\| \bX - \mathbb{E}[\bX \mid \bY, \bV]\right\|_2^2 \right],
\end{equation*}
which coincides with the result reported in~\cite{zhang2016remark}. In the special case $B_x = 1$, this reduces to $T = \mathbb{E}_{Y, V} [|X - E[X \mid Y, v]|^2]$, thereby recovering the result in \cite[Proposition~2]{wang2022generalized}.
\subsection{Constant Matrix Scaling Function}
In this case, we generalize the constant scalar scaling function in Section \ref{subsec:cssf} to a constant matrix scaling function. Specifically, we consider an adjustable matrix $\bm{\Pi} \in \CC^{B_x \times B_x}$, which is independent of both the channel output $\bY$ and the CSI $\bV$. The Vec-GNNDR is then given by
\begin{equation} \label{eq:Vec-GNNDR-cmsf}
    \hat{m} = \arg \min_{m \in \mathcal{M}} \sum_{l=1}^L  \left\|\bg(\by_l, \bv_l) - \bm{\Pi} \bx_l(m) \right\|^2_2.
\end{equation}
The corresponding optimal GMI under this decoding rule is characterized by the following proposition.

\begin{proposition}[Constant matrix scaling function]
\label{prop:Vec-GNNDR-cmsf}
Consider the in-block memory channel described in
Section~\ref{sec:set-up}, which satisfies
Assumption~\ref{ass:psd_cond_cov}.
Let $\bm{\mu}(\by,\bv)$ and $\bm{\Sigma}$ be defined as in
Proposition~\ref{prop:Vec-GNNDR-cssf}.
The matrix $\bm{\Sigma}$ admits the spectral decomposition
\begin{equation*}
    \tilde{\bm{W}}^H\bm{\Sigma}\tilde{\bm{W}}
    =
    \operatorname{diag}(\lambda_1,\ldots,\lambda_{B_x}),
\end{equation*}
where
$P \geq \lambda_1\geq\cdots\geq\lambda_{B_x}>0$ and
$\tilde{\bm{W}}
=[\tilde{\bw}_1,\ldots,\tilde{\bw}_{B_x}]$ is unitary.
Define $I^*=\{i:\lambda_i<P\}$ and $\tilde{\bm{W}}_*=[\tilde{\bw}_i]_{i\in I^*}$. Under the decoding rule \eqref{eq:Vec-GNNDR-cmsf}, the optimal GMI is
\begin{equation*}
    I_{\mathrm{GMI,cmsf}}
    =
    \frac{1}{B_x}
    \log\frac{P^{B_x}}{\det(\bm{\Sigma})}.
\end{equation*}
The optimal processing function
and scaling matrix are given by 
\begin{equation*}
    \begin{aligned}
        \bg(\by,\bv)
        =
        \operatorname{diag}
        \left(
            \sqrt{\frac{P}{(P-\lambda_i)\lambda_i}}
        \right)_{i\in I^*}
        \tilde{\bm{W}}_*^H\bm{\mu}(\by,\bv), \qquad
        \bm{\Pi}
        =
        \operatorname{diag}
        \left(
            \sqrt{\frac{P-\lambda_i}{P\lambda_i}}
        \right)_{i\in I^*}
        \tilde{\bm{W}}_*^H.
    \end{aligned}
\end{equation*}
The displayed pair omits the inactive output coordinates; zero entries
and zero rows may be appended to $\bg$ and $\bm{\Pi}$, respectively,
to retain the original dimensions.
If $I^*=\varnothing$, take
$\bg(\by,\bv)=\boldsymbol{0}_{B_x}$ and
$\bm{\Pi}=\boldsymbol{0}_{B_x\times B_x}$.
\end{proposition}
\begin{proof}
    See Appendix \ref{app:restricted_vec_GNNDR}.
\end{proof}

\begin{remark}
    By the maximum eigenvalue inequality (Lemma~\ref{lem:max_eigen_inequal}), we have $\lambda_{\max}(\bm{\Sigma}) \leq P$.  
    If $\lambda_i = P$, the $i$-th eigen-direction $\bw_i$ contributes zero to the optimized GMI, indicating that this direction provides no additional achievable GMI under the considered decoding metric. In this case, one may directly set $\sigma_i = 0$.  A detailed discussion of this scenario is provided in the proof in Appendix~\ref{app:restricted_vec_GNNDR}.
\end{remark}

In comparison with the optimal processing and scaling functions in \eqref{eq:opt_g_eps} and \eqref{eq:opt_f_eps}, each $\lambda_i(\by, \bv)$ and $\bw_i(\by, \bv)$ is replaced by the $i$-th eigenvalue and the $i$-th eigenvector of $\bm{\Sigma}$, respectively.
\subsection{CSI-dependent Scalar Scaling Function}
\label{subsec:csi-ssf}
In this subsection, we extend the constant scalar scaling function introduced in Section~\ref{subsec:cssf} by allowing it to depend on the receiver CSI $\bV$, while remaining independent of the channel output $\bY$.  
Under this refinement, the Vec-GNNDR is given by
\begin{equation}
    \hat{m} = \arg \min_{m \in \mathcal{M}} \sum_{l=1}^L \|\bg(\by_l, \bv_l) - \alpha(\bv_l) \bx_l(m)\|^2_2, 
    \label{eq:dec_csi-ssf}
\end{equation}
where $\alpha: \mathcal{V} \to \mathbb{C}$ denotes a CSI-dependent scalar function.  
Equivalently, the scaling function can be expressed as $\bm{f}(\by, \bv) = \alpha(\bv) \boldsymbol{I}_{B_x}$, which, together with $\bg(\by,\bv)$, is jointly optimized to maximize the GMI.  We now establish the following result concerning the optimal GMI under this constraint.
\begin{proposition}
\label{prop:Vec-GNNDR-csi-ssf}
Consider the in-block memory channel described in
Section~\ref{sec:set-up}, which satisfies
Assumption~\ref{ass:psd_cond_cov}.
Let $\bm{\mu}(\by,\bv)$ and $\bm{\Sigma}(\by,\bv)$ be defined as in
Theorem~\ref{thm:main_result}, and define
\begin{equation*}
    \begin{aligned}
        \bm{\Sigma}(\bv)
        =
        \mathbb{E}_{\bY\mid\bv}
        \left[\bm{\Sigma}(\bY,\bv)\right]=
        P\bI_{B_x}
        -
        \mathbb{E}_{\bY\mid\bv}
        \left[\bm{\mu}(\bY,\bv)\bm{\mu}^H(\bY,\bv)\right],
    \end{aligned}
\end{equation*}
with normalized trace
$T(\bv)=B_x^{-1}\operatorname{Tr}(\bm{\Sigma}(\bv))$.
The optimal GMI under \eqref{eq:dec_csi-ssf} is
\begin{equation*}
    I_{\mathrm{GMI,csi-ssf}}
    =
    \mathbb{E}_{\bV}
    \left[\log\frac{P}{T(\bV)}\right],
\end{equation*}
which is attained by the optimal processing function $\bg$ and CSI-dependent scalar scaling function  of the form
\begin{equation*}
    \bg(\by,\bv)
    =
    \sqrt{\frac{P}{T(\bv)(P-T(\bv))}}\,
    \bm{\mu}(\by,\bv),
    \qquad
    \alpha(\bv)
    =
    \sqrt{\frac{P-T(\bv)}{PT(\bv)}},
\end{equation*}
whenever $0<T(\bv)<P$.
In the boundary case $T(\bv)=P$, the optimal solution reduces to 
$\bg(\by,\bv)=\boldsymbol{0}_{B_x}$ and $\alpha(\bv)=0$.
\end{proposition}

\begin{proof}
See Appendix~\ref{app:restricted_vec_GNNDR}.
\end{proof}

For the case $B_x = 1$, we obtain 
\begin{equation*}
    T(v) = \mathbb{E}[|X|^2 \mid  v] - \mathbb{E}\left[ |\mu(Y, v)|^2 \mid v \right] = \mathbb{E} \left[ |X - \mu(Y, v)|^2 \mid v\right],
\end{equation*}
which recovers the GMI reported in \cite[Proposition 3]{wang2022generalized}. By the orthogonality principle (Lemma~\ref{lem:ortho}), we further derive
\begin{equation*}
    P -T(v) = \mathbb{E}\left[ | \mu(Y,v)|^2 \mid v \right] = \mathbb{E}[X^{H} \mu(Y,v) \mid v],
\end{equation*}
which implies that 
\begin{equation*}
    Q(v) := T(v)(P -T(v)) = [P - (P - T(v))](P - T(v)) = P \mathbb{E}\left[ | \mu(Y,v)|^2 \mid v \right] - (\mathbb{E}[X^{H} \mu(Y,v) \mid v])^2.
\end{equation*}
Consequently, the processing function and scaling function are respectively given by
\begin{equation*}
    g(y,v)  = \sqrt{\frac{P}{Q(v)}} \mu(y,v)
\end{equation*}
and 
\begin{equation*}
    f(v) = \alpha(v) =  \sqrt{\frac{P}{T(v)(P-T(v))}} \frac{P-T(v)}{P} = \sqrt{\frac{P}{Q(v)}} \frac{\mathbb{E}[X^{H} \mu(Y,v) \mid v]}{P}.
\end{equation*}
Thus, we recover the expressions of the processing and scaling functions in \cite[Proposition 3]{wang2022generalized}, up to a multiplicative factor of $\sqrt{P}$, which does not affect the decoding performance.

\subsection{CSI-dependent Matrix Scaling Function}
\label{subsec:Vec-GNNDR-csi-msf}
In this subsection, we generalize the CSI-dependent scalar scaling function introduced in Section~\ref{subsec:csi-ssf} to a CSI-dependent matrix scaling function. Specifically, we consider a mapping $\bm{\Pi}: \mathcal{V} \to \mathbb{C}^{B_x \times B_x}$ that depends on the CSI $\bV$ but remains independent of the channel output $\bY$.  
Under this formalism, the Vec-GNNDR takes the form 
\begin{equation}
    \hat{m} = \arg \min_{m \in \mathcal{M}} \sum_{l=1}^L \|\bg(\by_l, \bv_l) - \bm{\Pi}(\bv_l) \bx_l(m)\|^2_2.
    \label{eq:dec-csi-msf}
\end{equation}
The optimal GMI under this decoding rule is stated in the following proposition.
\begin{proposition}
\label{prop:Vec-GNNDR-csi-msf}
Consider the in-block memory channel described in
Section~\ref{sec:set-up}, which satisfies
Assumption~\ref{ass:psd_cond_cov}.
Let $\bm{\mu}(\by,\bv)$ and $\bm{\Sigma}(\bv)$ be defined as in
Proposition~\ref{prop:Vec-GNNDR-csi-ssf}.
For almost every $\bv$, let
\begin{equation*}
    \tilde{\bm{W}}^H(\bv)
    \bm{\Sigma}(\bv)
    \tilde{\bm{W}}(\bv)
    =
    \operatorname{diag}
    \left(\lambda_1(\bv),\ldots,\lambda_{B_x}(\bv)\right),
\end{equation*}
where
$P \geq \lambda_1(\bv)\geq\cdots\geq\lambda_{B_x}(\bv)>0$
and $\tilde{\bm{W}}(\bv)$ is unitary.
Define $I^*(\bv)=\{i:\lambda_i(\bv)<P\}$ and $\tilde{\bm{W}}_*(\bv)=[\tilde{\bw}_i(\bv)]_{i\in I^*(\bv)}$. The optimal GMI under the decoding rule \eqref{eq:dec-csi-msf} is given by
\begin{equation*}
    I_{\mathrm{GMI,csi-msf}}
    =
    \frac{1}{B_x}
    \mathbb{E}_{\bV}
    \left[
        \log\frac{P^{B_x}}{\det(\bm{\Sigma}(\bV))}
    \right],
\end{equation*}
and is attained by the processing function $\bg$ and the CSI-dependent matrix scaling function $\bm{\Pi}(\bv)$ respectively given by
\begin{equation*}
    \begin{aligned}
        \bg(\by,\bv)
        &=
        \operatorname{diag}
        \left(
            \sqrt{\frac{P}
            {(P-\lambda_i(\bv))\lambda_i(\bv)}}
        \right)_{i\in I^*(\bv)}
        \tilde{\bm{W}}_*^H(\bv)\bm{\mu}(\by,\bv),\\
        \bm{\Pi}(\bv)
        &=
        \operatorname{diag}
        \left(
            \sqrt{\frac{P-\lambda_i(\bv)}
            {P\lambda_i(\bv)}}
        \right)_{i\in I^*(\bv)}
        \tilde{\bm{W}}_*^H(\bv).
    \end{aligned}
\end{equation*}
Inactive output coordinates may be restored by zero padding.
If $I^*(\bv)=\varnothing$, take
$\bg(\by,\bv)=\boldsymbol{0}_{B_x}$ and
$\bm{\Pi}(\bv)=\boldsymbol{0}_{B_x\times B_x}$. 
\end{proposition}

\begin{proof}
See Appendix~\ref{app:restricted_vec_GNNDR}.
\end{proof}

\subsection{Linear Processing Function}
Within this formalism, we simplify the processing function by restricting it to be linear in the channel output $\by$, i.e., $\bg(\by, \bv) = \bm{\Gamma}^H(\bv) \by$ with   $\bm{\Gamma}: \mathcal{V} \mapsto \mathbb{C}^{B_y \times B_x}$, and $\bm{\Gamma}(\bv) = [\bm{\gamma}_1(\bv), \cdots, \bm{\gamma}_{B_x}(\bv)]$. Furthermore, we restrict the scaling function to be a CSI-dependent matrix $\bm{\Pi}(\bv)$, as in Section~\ref{subsec:Vec-GNNDR-csi-msf}.  
Under these restrictions, the Vec-GNNDR takes the form
\begin{equation}
    \hat{m} = \arg \min_{m \in \mathcal{M}} \sum_{l=1}^{L} \| \bm{\Gamma}^H(\bv_l) \by - \bm{\Pi}(\bv_l) \bx_l(m) \|^2_2.
    \label{eq:dec-lin}
\end{equation}
The following proposition provides the expression of the optimal GMI achieved under this decoding rule, along with the corresponding optimal processing and scaling functions.

\begin{proposition}
\label{prop:Vec-GNNDR-lin}
Consider the in-block memory channel described in
Section~\ref{sec:set-up}, and suppose
Assumption~\ref{ass:psd_cond_cov} holds.
Assume that
$\mathbb{E}[\|\bY\|_2^2\mid\bV]<\infty$ almost surely.
For each CSI realization $\bv$, define
\begin{equation}
    \bm{Q}(\bv)
    =
    \mathbb{E}[\bX\bY^H\mid\bv]
    \left(\mathbb{E}[\bY\bY^H\mid\bv]\right)^\dagger
    \mathbb{E}[\bY\bX^H\mid\bv].
    \label{eq:lin-Q}
\end{equation}
Let
\begin{equation*}
    \tilde{\bm{W}}^H(\bv)\bm{Q}(\bv)\tilde{\bm{W}}(\bv)
    =
    \operatorname{diag}
    \left(\lambda_1(\bv),\ldots,\lambda_{B_x}(\bv)\right),
\end{equation*}
where
$0\leq\lambda_1(\bv)\leq\cdots\leq\lambda_{B_x}(\bv)<P$
almost surely and
$\tilde{\bm{W}}(\bv)
=[\tilde{\bw}_1(\bv),\ldots,\tilde{\bw}_{B_x}(\bv)]$
is unitary. Define $I^*(\bv)=\{i:\lambda_i(\bv)>0\}$ and  $\tilde{\bm{W}}_*(\bv) = [\tilde{\bw}_i(\bv)]_{i\in I^*(\bv)}$
The optimal GMI under the decoding rule \eqref{eq:dec-lin} is given by
\begin{equation}
    I_{\mathrm{GMI,lin}}
    =
    \frac{1}{B_x}
    \mathbb{E}_{\bV}
    \left[
        \log\frac{P^{B_x}}
        {\det(P\bI_{B_x}-\bm{Q}(\bV))}
    \right].
    \label{eq:GMI_lin}
\end{equation}
Moreover, the optimum is attained by a linear processing function $\bg(\by,\bv)=\bm{\Gamma}^H(\bv)\by$ and a CSI-dependent scaling matrix $\bm{\Pi}(\bv)$ given by
\begin{equation}
    \begin{aligned}
        \bm{\Gamma}(\bv)
        =
        \left(\mathbb{E}[\bY\bY^H\mid\bv]\right)^\dagger
        \mathbb{E}[\bY\bX^H\mid\bv]\,
        \tilde{\bm{W}}_*(\bv) 
        \operatorname{diag}
        \left(
            \sqrt{\frac{P}
            {(P-\lambda_i(\bv))\lambda_i(\bv)}}
        \right)_{i\in I^*(\bv)},
    \end{aligned}
    \label{eq:lin-process}
\end{equation}
and
\begin{equation}
    \bm{\Pi}(\bv)
    =
    \operatorname{diag}
    \left(
        \sqrt{\frac{\lambda_i(\bv)}
        {P(P-\lambda_i(\bv))}}
    \right)_{i\in I^*(\bv)}
    \tilde{\bm{W}}_*^H(\bv).
    \label{eq:lin-scaling}
\end{equation}
Zero columns and zero rows may be appended to
$\bm{\Gamma}(\bv)$ and $\bm{\Pi}(\bv)$, respectively,
to retain the original dimensions.
If $I^*(\bv)=\varnothing$, both matrices are set to zero.
\end{proposition}

\begin{proof}
See Appendix~\ref{app:restricted_vec_GNNDR}.
\end{proof}

Now we provide a heuristic justification for Vec-GNNDR with a linear processing function by employing the generalized Bussgang decomposition~\cite{bussgang1952crosscorrelation,zhang2019regression,demir2020bussgang}. For a given CSI realization $\bv$, following the Bussgang formalism, we decompose the channel output $\bY$ as 
\begin{equation}
    \bY = \underbrace{\frac{\mathbb{E}[\bY \bX^H \mid \bv]}{P} \bX}_{\text{Signal:\;}\bm{A}(\bv) \bX} + \underbrace{\bY - \frac{\mathbb{E}[\bY \bX^H \mid \bv]}{P} \bX}_{\mathrm{Noise:\;} \bm{n}(\bv)},
    \label{eq:Bussgang_channel}
\end{equation}
where $\bm{A}(\bv)$ denotes the effective linear operator associated with $\bv$, and $\bm{n}(\bv)$ represents the residual noise.
This decomposition ensures that the signal and residual noise components are conditionally uncorrelated given $\bv$, i.e.,
\begin{equation*}
    \mathbb{E}\left[ \bm{n}^H(\bv) \bm{A}(\bv) \bX \mid \bv \right] = 0.
\end{equation*}
Furthermore, the conditional correlation matrix of $\bm{n}(\bv)$ is expressed as
\begin{equation*}
    W(\bv) = \mathbb{E}[\bm{n}(\bv) \bm{n}^H(\bv) \mid \bv] = \mathbb{E}[\bY \bY^H \mid \bv] - \frac{1}{P} \mathbb{E}[\bY \bX^H \mid \bv] \mathbb{E}[\bX \bY^H \mid \bv].
\end{equation*}
By invoking the \emph{worst-case uncorrelated additive noise} theorem~\cite{hassibi2003much,shomorony2013worst}, the capacity for channel~\eqref{eq:Bussgang_channel} under the worst-case scenario is attained when the residual noise follows a circularly symmetric complex Gaussian distribution, i.e., $\bm{n}(\bv) \sim \mathcal{CN}(0, W(\bv))$. For the prescribed Gaussian input, the worst-case uncorrelated-noise argument~\cite{hassibi2003much,shomorony2013worst} yields the following Gaussian-noise rate for a given $\bv$
\begin{equation}
    \begin{aligned}
        R(\bv) &= B_x^{-1} \log \left[ \det(P\bm{A}^H(\bv) \bm{W}^{\dagger}(\bv) \bm{A}(\bv)+\boldsymbol{I}_{B_x})\right] \\
        & = B_x^{-1} \log \left[ \det \left( \frac{1}{P} \mathbb{E}[\bX \bY^H \mid \bv] \left( \mathbb{E}[\bY \bY^H \mid \bv] - \frac{1}{P} \mathbb{E}[\bY \bX^H \mid \bv] \mathbb{E}[\bX \bY^H \mid \bv]\right)^{\dagger} \mathbb{E}[\bY \bX^H \mid \bv] + \boldsymbol{I}_{B_x} \right) \right]\\
        & = B_x^{-1} \log \frac{P^{B_x}}{\det \left(P \boldsymbol{I}_{B_x} - \mathbb{E}[\bX \bY^H \mid \bv]   \left( \mathbb{E}[\bY \bY^H \mid \bv]\right)^{\dagger} \mathbb{E}[\bY \bX^H \mid \bv] \right)},
    \end{aligned}
    \label{eq:R_Bussgang}
\end{equation}
where the last step follows from the Sherman–Morrison–Woodbury identity (Lemma~\ref{lem:wood_inv}) to compute the inverse of $\bW(\bv)$. Comparing~\eqref{eq:R_Bussgang} with the GMI expression in~\eqref{eq:GMI_lin}, we obtain the relationship
\begin{equation*}
    I_{\mathrm{GMI,lin}}
    =
    \mathbb{E}_{\bV}[R(\bV)]
    =
    \mathbb{E}_{\bV}
    \left[
        B_x^{-1}
        \log\det\left(
            \bI_{B_x}
            +
            P\bm{A}^H(\bV)
            \bm{W}^{\dagger}(\bV)
            \bm{A}(\bV)
        \right)
    \right].
\end{equation*}
This connection enables the following heuristic interpretation: the Vec-GNNDR with a linear processing function first observes the CSI $\bv$ and applies the Bussgang decomposition, which effectively absorbs the effect of (possibly noisy) nonlinearities into a residual noise term that is uncorrelated with the channel input. The resulting residual noise then governs the worst-case  achievable rate, thereby guaranteeing the achievable rate of $I_{\mathrm{GMI,lin}}$.

For the special case $B_x=1$, the matrix $\bm{Q}(v)$ reduces to
the scalar $\lambda_1(v)$, and the diagonalizing matrix can be
chosen as $\tilde{\bm{W}}(v)=1$. Hence, \eqref{eq:GMI_lin} becomes
\begin{equation*}
    I_{\mathrm{GMI,lin}}
    =
    \mathbb{E}_{V}
    \left[
        \log\frac{P}{
            P-
            \mathbb{E}[X\bY^H\mid V]
            \left(\mathbb{E}[\bY\bY^H\mid V]\right)^\dagger
            \mathbb{E}[\bY X^H\mid V]
        }
    \right].
\end{equation*}
For $\lambda_1(v)>0$, substituting $B_x=1$ into
\eqref{eq:lin-process} and \eqref{eq:lin-scaling} gives
\begin{equation*}
    \begin{aligned}
        \bm{\Gamma}(v)
        =
        \sqrt{PS(v)}
        \left(\mathbb{E}[\bY\bY^H\mid v]\right)^\dagger
        \mathbb{E}[\bY X^H\mid v], \qquad
        \Pi(v)
        =
        \sqrt{PS(v)}\,\frac{1}{P}
        \mathbb{E}[X\bY^H\mid v]
        \left(\mathbb{E}[\bY\bY^H\mid v]\right)^\dagger
        \mathbb{E}[\bY X^H\mid v],
    \end{aligned}
\end{equation*}
where $S(v)$ is given by
\begin{equation*}
    \begin{aligned}
        S(v)
        =
        \bigg[
        \mathbb{E}[X\bY^H\mid v]
        \left(\mathbb{E}[\bY\bY^H\mid v]\right)^\dagger
        \mathbb{E}[\bY X^H\mid v] \times
        \left(
            P-
            \mathbb{E}[X\bY^H\mid v]
            \left(\mathbb{E}[\bY\bY^H\mid v]\right)^\dagger
            \mathbb{E}[\bY X^H\mid v]
        \right)
        \bigg]^{-1}.
    \end{aligned}
\end{equation*}
In particular,
$S(v)=[\lambda_1(v)(P-\lambda_1(v))]^{-1}$.
If $\lambda_1(v)=0$, both coefficients are set to zero.

When $\mathbb{E}[\bY\bY^H\mid v]$ is invertible, these expressions
recover the GMI and processing--scaling pair in
\cite[Proposition~4]{wang2022generalized}, up to a common
multiplicative factor of $\sqrt{P}$ in the processing and scaling
functions, which does not affect the decoding rule or its
optimized GMI.


\subsection{Comparison of GMI under Restricted Vec-GNNDR Schemes}
In this subsection, we present a comparative analysis of the GMI achieved by the five restricted Vec-GNNDR schemes introduced in this section. By examining the relative dimensions of the optimization domains associated with the processing functions $\boldsymbol{g}$ and the scaling functions $\boldsymbol{f}$, one readily recognizes  the following ordering:
\begin{equation*}
    I_{\mathrm{GMI,opt}} \geq I_{\mathrm{GMI, csi-msf}} \geq \max \{ I_{\mathrm{GMI, csi-ssf}}, I_{\mathrm{GMI,cmsf}}\} \geq I_{\mathrm{GMI, cssf}},
\end{equation*}
Moreover, for the Vec-GNNDR with a linear processing function,  we have
\begin{equation*}
    I_{\mathrm{GMI,lin}} \leq I_{\mathrm{GMI, csi-msf}}.
\end{equation*}

\section{Numerical Results} \label{sec:simulation}
In this section, we numerically evaluate the performance of Vec-GNNDR over representative synthetic IBM channels. We compare the GMIs achieved by the optimal Vec-GNNDR and its restricted variants with the Gaussian-input mutual information per channel use, denoted by $I_{\mathrm{Gauss}}$, between the input $\bX$ and the channel output $\bY$. This quantity provides an upper bound on these GMIs under the prescribed Gaussian codebook.
\subsection{Block
noncoherent AWGN channel}
\label{subsec:non-coherent-AWGNC}
The codeword $\bx^L = [\bx_1, \cdots, \bx_L]$  is transmitted over a block noncoherent AWGN channel \cite{lapidoth2002phase,nuriyev2005capacity} as follows:
\begin{equation} \label{eq:channel_rotation}
    \by_l = e^{i \theta_l} \bx_l + \bm{n}_l,
\end{equation}
where $\theta_1, \cdots, \theta_L \overset{i.i.d}{\sim} U([0, 2 \pi])$ is the block-wise rotation angle and $\bm{n}_1, \cdots, \bm{n}_L \overset{i.i.d}{\sim} \mathcal{CN}(0, \sigma^2 \bI_{B_x})$   is circularly symmetric complex Gaussian noise. The rotation angle may vary across blocks, but within each block, all entries are subjected to the same phase rotation. The Gaussian-input mutual information upper bound admits the representation
\begin{equation}
    I_{\mathrm{Gauss}}
    =
    \log\left(1+\frac{P}{\sigma^2}\right)
    -
    \frac{2P}{\sigma^2}
    +
    \frac{1}{B_x}
    \mathbb{E}_{\bX,\bY}
    \left[
        \log I_0\left(
            \frac{2|\bX^H\bY|}{\sigma^2}
        \right)
    \right],
    \label{eq:Gaussian_MI_block_nonherent}
\end{equation}
where $I_0$ denotes the modified Bessel function of the first kind
of order zero. The derivation and the one-dimensional integration representation of ~\eqref{eq:Gaussian_MI_block_nonherent} in numerical simulation are provided
in Appendix~\ref{app:non-coherent-AWGNC}.

The following proposition characterizes the optimal processing function and scaling function and the corresponding optimal GMI for the block noncoherent AWGN channel.
\begin{proposition}[Block Noncoherent AWGN Channel]
    \label{prop:rotation}
    Consider the block noncoherent AWGN channel defined in~\eqref{eq:channel_rotation}.  
    The optimal GMI is given by 
    \begin{equation}
        I_{\mathrm{GMI,opt}} = \frac{B_x - 1}{ B_x} \left[\log \left( 1 +\frac{P}{\sigma^2}\right) - \frac{P}{P+\sigma^2}\right] + \frac{1}{B_x} \mathbb{E}_{T \sim \chi^2(2B_x)} \left[ \left( \log \frac{P+\sigma^2}{\sigma^2+PT/2} + \frac{P(T-2)}{2(P+\sigma^2)} \right) \mathbb{I}\{T <2\}\right].
        \label{eq:GMI-block}
    \end{equation}
    Furthermore, the optimal GMI is achieved by the Vec-GNNDR with the processing function $\bg$ and scaling function $\bm{f}$ specified as follows:
    \begin{itemize}
        \item If $\by^H \by < P + \sigma^2$, then
        \begin{equation*}
            \bg(\by) = \boldsymbol{0}_{B_x}, \qquad 
            \bm{f}(\by) = \mathrm{diag} \!\left( 
                \sqrt{\tfrac{P-\lambda_1(\by)}{P \lambda_1(\by)}}, 
                \tfrac{1}{\sigma}, \ldots, \tfrac{1}{\sigma} \right) 
                \bm{W}^H(\by),
        \end{equation*}
        \item If $\by^H \by \geq P + \sigma^2$, then
        \begin{equation*}
            \bg(\by) = \boldsymbol{0}_{B_x-1}, \qquad 
            \bm{f}(\by) = \mathrm{diag} \!\left( 
                \tfrac{1}{\sigma}, \ldots, \tfrac{1}{\sigma} \right) 
                \bm{W}_{\perp}^H(\by),
        \end{equation*}
    \end{itemize}
    where $\lambda_1(\by)$ is defined as
    \begin{equation*}
        \lambda_1(\by) = \frac{P \sigma^2}{P + \sigma^2} 
        + \left( \frac{P}{P + \sigma^2} \right)^2 \by^H \by,
    \end{equation*}
    and 
    \begin{equation*}
        \bm{W}(\by) = \big(\bw_1(\by), \bW_\perp(\by)\big), 
        \qquad \bw_1(\by) = \frac{\by}{\|\by\|_2},
    \end{equation*}
    with $\bW_\perp(\by)$ denoting an orthonormal basis for the orthogonal complement of $\bw_1(\by)$, i.e., the subspace spanned by $\bI_{B_x} - \bw_1(\by) \bw_1^H(\by)$.

\end{proposition}
\begin{proof}
    See Appendix \ref{app:non-coherent-AWGNC}.
\end{proof}

In fact, one may apply GNNDR element-wise since each pair $(x_i, y_i)$ of the codeword and the corresponding channel output shares the same joint marginal distribution for a random-phase channel as follows  
\begin{equation} \label{eq:distri_br_ele}
    y = e^{i \theta} x +n,
\end{equation}
with $\theta \sim U([0, 2\pi])$ and $n \sim \mathcal{CN}(0, \sigma^2)$. The following proposition characterizes the optimal processing function $g$, scaling function $\bm{f}$ and the corresponding GMI when GNNDR is applied in an element-wise manner.
\begin{proposition}[Block Noncoherent AWGN Channel with element-wise GNNDR] \label{prop:rotation_ele}
     Consider the block noncoherent AWGN channel defined in \eqref{eq:channel_rotation}.
     The optimal GMI achieved by element-wise GNNDR is given by 
     \begin{equation*}
         I_{\mathrm{GMI,ele}} = \mathbb{E}_{T \sim \chi^2(2)}  \left[ \left( \log \frac{P+\sigma^2}{\sigma^2+PT/2} + \frac{P(T-2)}{2(P+\sigma^2)} \right) \mathbb{I}\{T <2\}\right].
     \end{equation*}
     Moreover, the optimal GMI is attained by the element-wise GNNDR with processing function $g$ and scaling function $f$ specified as follows:
     \begin{equation*}
         g(y) = 0, \qquad f(y) = \sqrt{\frac{P-w(y)}{Pw(y)}}, \qquad \mathrm{if} \;|y|^2 <P+\sigma^2,
     \end{equation*}
     while for $|y|^2 \geq P + \sigma^2$ we set $g(y) = f(y) = 0$, where the variance term $w(y)$ is defined as
    \begin{equation*}
        w(y) = \frac{P \sigma^2}{P + \sigma^2} 
        + \left( \frac{P}{P + \sigma^2} \right)^2 |y|^2.
    \end{equation*}
\end{proposition}
\begin{proof}
    See Appendix \ref{app:non-coherent-AWGNC}.
\end{proof}

This example also illustrates why the block-level formulation is not merely a notational vectorization of symbol-wise GNNDR. Although the IBM channel can be viewed as a vector-valued memoryless channel across blocks, an element-wise metric only adapts to each input--output pair separately and therefore does not exploit the joint structure induced by the common phase within a block. In contrast, the Vec-GNNDR metric adapts its eigendirections to the posterior covariance and separates the dominant phase-uncertainty direction from the remaining orthogonal directions. By comparing the optimal GMI achieved by Vec-GNNDR and element-wise GNNDR, we observe that Vec-GNNDR can exploit the structural property that all elements within a block are rotated by the same unknown phase $\theta$. Element-wise GNNDR, on the other hand, treats each element independently: its scalar scaling function is adapted separately based on each individual observation, without utilizing the joint information across the block. 
Consequently, the element-wise metric cannot fully exploit the common rotational structure within the block, resulting in the lower achievable GMI \(I_{\mathrm{GMI,ele}}\).

Figure~\ref{fig:block_rotational} contrasts the achievable GMIs of the optimal Vec-GNNDR in Proposition~\ref{prop:rotation} with those of the element-wise GNNDR in Proposition~\ref{prop:rotation_ele} for block noncoherent AWGN Channel. The performance gap widens with the signal-to-noise ratio $\mathrm{SNR}:=10 \log_{10}(P/\sigma^2)$ and becomes pronounced in the high-SNR regime. Indeed, 
\begin{equation*}
    \lim_{\mathrm{SNR} \to \infty} I_{\mathrm{GMI, opt}} \geq \lim_{\mathrm{SNR} \to \infty} \frac{B_x -1}{B_x} \left[ \log \left( 1+ \frac{P}{\sigma^2}\right) - \frac{P}{P+\sigma^2}\right] = \lim_{\mathrm{SNR} \to \infty} \frac{B_x - 1}{ B_x} \left[ \log \left( 1 + \frac{P}{\sigma^2}\right) - 1\right] \to \infty,
\end{equation*}
so $I_{\mathrm{GMI,opt}}$ grows unboundedly at a rate of $I_{\mathrm{GMI,opt}}=\frac{(B_x-1)\log 10}{ 10 B_x}\mathrm{SNR}+O(1)$. In sharp contrast, the GMI attainable by element-wise GNNDR saturates to a finite constant:
\begin{equation*}
     \lim_{\text{SNR} \to \infty} I_{\text{GMI, ele}} =  \mathbb{E}_{T \sim \chi^2(2)} \left[ \left( \log \frac{2}{T} + \frac{T}{2} -1\right) \mathbb{I} \{T<2\}\right] = \gamma - e^{-1} +E_1(1) \approx 0.42872,
\end{equation*}
where $\gamma\approx 0.577$ is the Euler–Mascheroni constant and $E_{1}(x)=\int_{x}^{\infty}\tfrac{e^{-t}}{t}\,\mathrm{d}t$ denotes the exponential integral. This finite saturation level appears as the gray reference curve in Figure~\ref{fig:block_rotational}. Although $I_{\mathrm{GMI,opt}}$ remains below $I_{\mathrm{Gaussian}}$, both increase with SNR and block length over the range shown.
\begin{figure}[tbp]
    \centering
    \includegraphics[width=12cm]{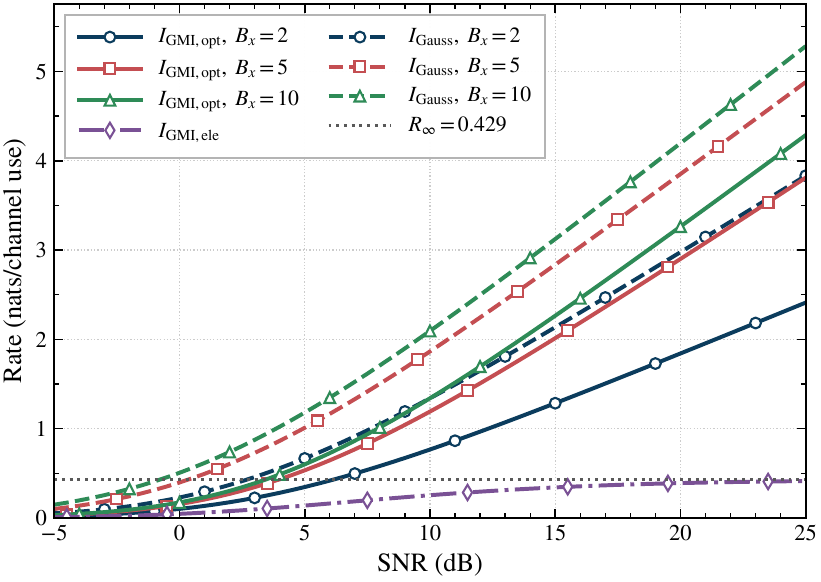}
    \caption{GMI comparison for  element-wise GNNDR and optimal Vec-GNNDR with block length for block noncoherent AWGN channel of $B_x = 2, 5, 10$.}
    \label{fig:block_rotational}
\end{figure}

We further observe that the gap widens as the block size increases.
For Vec-GNNDR, a larger block contains more jointly rotated symbols,
allowing the block-level metric to exploit the common phase structure
more effectively. As $B_x \to \infty$, the asymptotic GMI is given by
\begin{equation} \label{eq:block_limit}
    \lim_{B_x \to \infty} I_{\text{GMI, opt}} = \log \left( 1 +  \frac{P}{\sigma^2}\right) - \frac{P}{P + \sigma^2},
\end{equation}
Here, $\log(1+P/\sigma^2)$ corresponds to the coherent complex AWGN
benchmark, while the second term characterizes the remaining rate gap
under the unknown common phase and the considered Vec-GNNDR metric.
The calculation of $I_{\mathrm{GMI,opt}}$ in the limit
$B_x\to\infty$ is provided in Appendix~\ref{app:non-coherent-AWGNC}.
In contrast, element-wise GNNDR processes each symbol independently
and therefore cannot exploit the additional joint structure available
as the block size increases. Consequently, its GMI remains independent
of $B_x$.

\begin{remark}
    For the block noncoherent AWGN channel, the expected conditional covariance matrix is 
    \begin{equation*}
        \mathbb{E}_{\bY}[\bm{\Sigma}(\bY)] = \mathbb{E}_{\bY}[\mathrm{Cov}[\bX \mid \bY]] = P \boldsymbol{I}_{B_x},
    \end{equation*}
   where the detailed derivation is provided in Appendix~\ref{app:restricted_vec_GNNDR}. Consequently, we obtain
    \begin{equation*}
        \max \{I_{\mathrm{GMI,lin}}, I_{\mathrm{GMI,cssf}}\} \leq I_{\mathrm{GMI,cmsf}} = 0.
    \end{equation*}
    This result implies that the restricted Vec-GNNDR scheme is not applicable to the block noncoherent AWGN channel.
\end{remark}


\subsection{Phase noise channel}
\label{subsec:pnc}
We consider the following synthetic model for the phase noise channel, as proposed in \cite{zou2007compensation}:
\begin{equation} \label{eq:channel_phase}
    \by_l = \operatorname{diag}(e^{i \phi_{l1}}, \cdots, e^{i\phi_{l B_x}}) \bx_l  + \bm{n}_l, 
\end{equation}
where $\bm{n}_1, \cdots, \bm{n}_{l} \overset{i.i.d}{\sim} \mathcal{CN}(0 ,\sigma^2 \bI_{B_x})$ denotes the additive white Gaussian noise. The phase vector $\boldsymbol{\phi}_l = (\phi_{l1}, \cdots, \phi_{l B_x})$ evolves according to a discrete-time Brownian motion defined as
\begin{equation*}
    \phi_{l1} \sim \mathcal{N}(0, c^2), \quad \phi_{l(i+1)} = \phi_{li} + c Z_i,  \quad Z_i \overset{i.i.d}{\sim} \cN(0,1) \quad \text{for} \; i =1, \cdots, B_x-1,
\end{equation*}
where $c>0$ denotes the diffusion intensity of the Brownian process. While the phase rotation  applied to each block $\bx_l$ is independent across blocks, the phases within each block are correlated due to the Brownian motion. Specifically, the joint distribution of $(\phi_{l1}, \cdots, \phi_{lB_x}) $ is multivariate Gaussian with zero mean and covariance matrix $\Sigma_{\phi} \in \RR^{B_x \times B_x}$, where the $(i,j)$-th entry of $\Sigma_{\phi}$ is given by $(\Sigma_{\phi})_{ij} = c^2 \min \{i,j\}$ for $1 \leq i,j \leq B_x$. The Gaussian mutual information $I_{\mathrm{Gauss}}$ is given by 
\begin{equation}
    I_{\mathrm{Gauss}}
    =
    \log \left(1+ \frac{P}{\sigma^2} \right)
    +
    \frac{1}{B_x}
    \mathbb{E}_{\bX,\boldsymbol{\phi},\bN}
    \left[
        \frac{\|\bN\|_2^2}{\sigma^2}
        +
        \log\left(
            (\pi\sigma^2)^{B_x}p(\bY\mid\bX)
        \right)
    \right],
    \label{eq:Gaussian_MI_phase_noise}
\end{equation}
where $\bN\sim\mathcal{CN}(\boldsymbol{0},\sigma^2\bI_{B_x})$ independent of $\bX, \boldsymbol{\phi}$. The expectation is taken under their joint distribution, with
$\bY$ generated according to \eqref{eq:channel_phase}, and
$p(\by\mid\bx)$ denotes the channel transition density after
averaging over the unknown phase vector.
The derivation and numerical evaluation of~\eqref{eq:Gaussian_MI_phase_noise} are provided in Appendix~\ref{app:pnc}.

The optimal GMI $I_{\mathrm{GMI,opt}}$ is evaluated numerically by computing
the expectation in Theorem~\ref{thm:main_result}, as detailed shown in Appendix~\ref{app:pnc}. In contrast, the constant
scalar scaling, constant matrix scaling, and linear processing schemes
admit closed-form optimal GMIs and processing--scaling pairs, as
characterized by the following proposition.

\begin{proposition}[Restricted Vec-GNNDR for the phase noise channel]
\label{prop:lin-GNNDR-pnc}
Consider the phase noise channel defined in \eqref{eq:channel_phase}, with
$\bX \sim \mathcal{CN}(\boldsymbol{0},P\bI_{B_x})$, $P>0$, and $\sigma^2>0$.
Define
\begin{equation*}
    \eta = \frac{1}{B_x}\sum_{i=1}^{B_x}e^{-c^2 i},
    \qquad
    \nu_i = P(1-e^{-c^2 i})+\sigma^2,
    \qquad
    \overline{\nu}
    = \frac{1}{B_x}\sum_{i=1}^{B_x}\nu_i
    = P(1-\eta)+\sigma^2.
\end{equation*}
The optimal processing--scaling pairs are given as follows.

\begin{enumerate}
    \item[(i)] \emph{Constant scalar scaling function.}
    Under the decoding rule \eqref{eq:Vec-GNNDR-cssf}, the optimal GMI is
    \begin{equation}
        I_{\mathrm{GMI,cssf}}
        = \log\frac{P+\sigma^2}{\overline{\nu}},
        \label{eq:GMI-cssf-pnc}
    \end{equation}
    attained by
    \begin{equation}
        \begin{aligned}
            \bg_{\mathrm{cssf}}(\by)
            =
            \frac{1}{\sqrt{\eta\overline{\nu}}}
            \operatorname{diag}
            \left(e^{-c^2/2},\ldots,e^{-c^2B_x/2}\right)\by, \qquad
            \bm{f}_{\mathrm{cssf}}(\by)
            =
            \sqrt{\frac{\eta}{\overline{\nu}}}\,\bI_{B_x}.
        \end{aligned}
        \label{eq:cssf-pnc-functions}
    \end{equation}

    \item[(ii)] \emph{Constant matrix scaling function.}
    Under the decoding rule \eqref{eq:Vec-GNNDR-cmsf}, the optimal GMI is
    \begin{equation}
        I_{\mathrm{GMI,cmsf}}
        =
        \frac{1}{B_x}\sum_{i=1}^{B_x}
        \log\frac{P+\sigma^2}{\nu_i},
        \label{eq:GMI-cmsf-pnc}
    \end{equation}
    attained by
    \begin{equation}
        \begin{aligned}
            \bg_{\mathrm{cmsf}}(\by)
            =
            \operatorname{diag}
            \left(\nu_1^{-1/2},\ldots,\nu_{B_x}^{-1/2}\right)\by, \qquad
            \bm{f}_{\mathrm{cmsf}}(\by)
            =
            \operatorname{diag}
            \left(
                \frac{e^{-c^2/2}}{\sqrt{\nu_1}},
                \ldots,
                \frac{e^{-c^2B_x/2}}{\sqrt{\nu_{B_x}}}
            \right).
        \end{aligned}
        \label{eq:cmsf-pnc-functions}
    \end{equation}

    \item[(iii)] \emph{Linear processing function.}
    Under the decoding rule \eqref{eq:dec-lin}, the optimal GMI is
    \begin{equation}
        I_{\mathrm{GMI,lin}}
        =
        \frac{1}{B_x}\sum_{i=1}^{B_x}
        \log\left(
            \frac{P+\sigma^2}
            {P(1-e^{-c^2 i})+\sigma^2}
        \right).
        \label{eq:GMI-lin-pnc}
    \end{equation}
    An optimal pair is
    $\bg_{\mathrm{lin}}(\by)=\bm{\Gamma}^H\by$ and
    $\bm{f}_{\mathrm{lin}}(\by)=\bm{\Pi}$, where
    \begin{equation*}
        \bm{\Gamma}
        =\operatorname{diag}(\gamma_1,\ldots,\gamma_{B_x}),
        \qquad
        \bm{\Pi}
        =\operatorname{diag}(\pi_1,\ldots,\pi_{B_x}),
    \end{equation*}
    with
    \begin{equation*}
        \gamma_i=\frac{1}{\sqrt{\nu_i}},
        \qquad
        \pi_i=\frac{e^{-c^2 i/2}}{\sqrt{\nu_i}},
        \qquad i=1,\ldots,B_x.
    \end{equation*}
\end{enumerate}
\end{proposition}
\begin{proof}
    See Appendix~\ref{app:pnc}.
\end{proof}

As a simple baseline without phase-noise compensation, we consider the identity decoding rule
\begin{equation} \label{eq:dec_iden}
    d(\bx, \by) = \|\by- \bx\|_2^2,
\end{equation}
which corresponds to setting $\bg(\by) = \by$ and $\bm{f}(\by) = \bI_{B_x}$, without considering the phase rotation for $\bx_l$. The following proposition characterizes the GMI achieved under the identity decoding rule for the phase noise channel.

\begin{figure}[tbp]
    \centering
    \includegraphics[width=12cm]{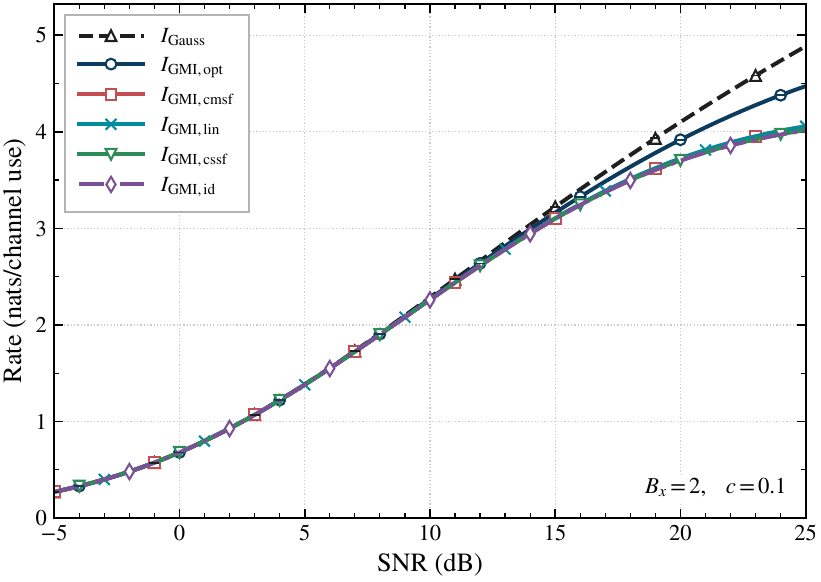}
    \caption{Performance comparison of the optimal (unrestricted) Vec-GNNDR, restricted optimal Vec-GNNDR variants, and the identity decoding rule.  The achievable GMI is evaluated for a phase-noise channel with block length $B_x = 2$ and diffusion intensity $c = 0.1$.}
    \label{fig:GMI_PNC1}
\end{figure}

\begin{proposition}[Phase noise channel with identity decoding rule]
    Consider the phase noise channel defined in \eqref{eq:channel_phase}. Under the identity decoding rule specified in \eqref{eq:dec_iden}, the corresponding GMI is given by
    \begin{equation*}
        I_{\mathrm{GMI,id}} = \theta^* \left[\sigma^2 + \frac{2P}{B_x} \sum_{i=1}^{B_x}   (1 - e^{-c^2 i / 2})\right] - \frac{(P+\sigma^2)\theta^*}{1 - \theta^* P} + \log(1-P\theta^*),
    \end{equation*}
    where $\theta^*$ is given by 
    \begin{equation*}
        \theta^{*}=\frac{2 A-P-\sqrt{P^{2}+4 A\left(P+\sigma^{2}\right)}}{2 A P},
    \end{equation*}
    where $A= \sigma^2 + \frac{2P}{B_x} \sum_{i=1}^{B_x}(1- e^{-c^2 i / 2})$.
    \label{prop:pnc_id}
\end{proposition}
\begin{proof}
    See Appendix~\ref{app:pnc}.
\end{proof}
Figure~\ref{fig:GMI_PNC1} compares the GMIs achieved by the optimal
Vec-GNNDR, its restricted variants, and the identity decoding rule.
The optimal Vec-GNNDR provides a substantial GMI gain over the identity
decoder, whereas the constant-scalar scaling, constant-matrix scaling,
and linear-processing variants yield more modest improvements.
The advantage of the unrestricted scheme becomes more pronounced
as the SNR increases over the range shown. This trend validates the efficacy of the Vec-GNNDR scheme.
In fact, for $B_x\ge2$, $c>0$, and $P,\sigma^2>0$,
Proposition~\ref{prop:lin-GNNDR-pnc} establishes
$I_{\mathrm{GMI,cmsf}}=I_{\mathrm{GMI,lin}}$.
Combining this identity with the nesting relations among the admissible
metric classes yields
\begin{equation}
\label{eq:PNC_GMI_ordering}
I_{\mathrm{GMI,opt}}
>
I_{\mathrm{GMI,cmsf}}
=
I_{\mathrm{GMI,lin}}
>
I_{\mathrm{GMI,cssf}}
>
I_{\mathrm{GMI,id}}.
\end{equation}
Despite a visible gap at higher SNR, $I_{\mathrm{GMI,opt}}$ broadly follows the increase of $I_{\mathrm{Gauss}}$ and remains closer to it than the restricted baselines.


The GMIs achieved by all decoding schemes increase with SNR.
For fixed $B_x$ and $c>0$, however, the identity decoder and the
restricted Vec-GNNDR variants approach finite high-SNR limits.
Specifically, the identity decoding rule satisfies
\begin{equation}
\label{eq:PNC_limit}
\lim_{\mathrm{SNR}\to\infty}I_{\mathrm{GMI,id}}
=
s+1-\sqrt{1+4s}
+\log\left(\frac{2}{\sqrt{1+4s}-1}\right),
\end{equation}
where
$s=\frac{2}{B_x}\sum_{j=1}^{B_x}(1-e^{-c^2j/2})$.
The constant-scalar-scaling scheme has the saturation level
\begin{equation*}
\lim_{\mathrm{SNR}\to\infty}I_{\mathrm{GMI,cssf}}
=
\log\left(\frac{1}{1-\eta}\right),
\end{equation*}
where $\eta=\frac{1}{B_x}\sum_{j=1}^{B_x}e^{-c^2j}$.
The constant-matrix-scaling and linear-processing schemes share
the same high-SNR limit:
\begin{equation*}
\begin{aligned}
\lim_{\mathrm{SNR}\to\infty}I_{\mathrm{GMI,cmsf}}
=\lim_{\mathrm{SNR}\to\infty}I_{\mathrm{GMI,lin}}=\frac{1}{B_x}\sum_{j=1}^{B_x}
\log\left(\frac{1}{1-e^{-c^2j}}\right).
\end{aligned}
\end{equation*}
Detailed calculations are provided in Appendix~\ref{app:pnc}.

\begin{figure}[tbp]
  \centering
  \subfloat{
    \includegraphics[width=0.48\linewidth]{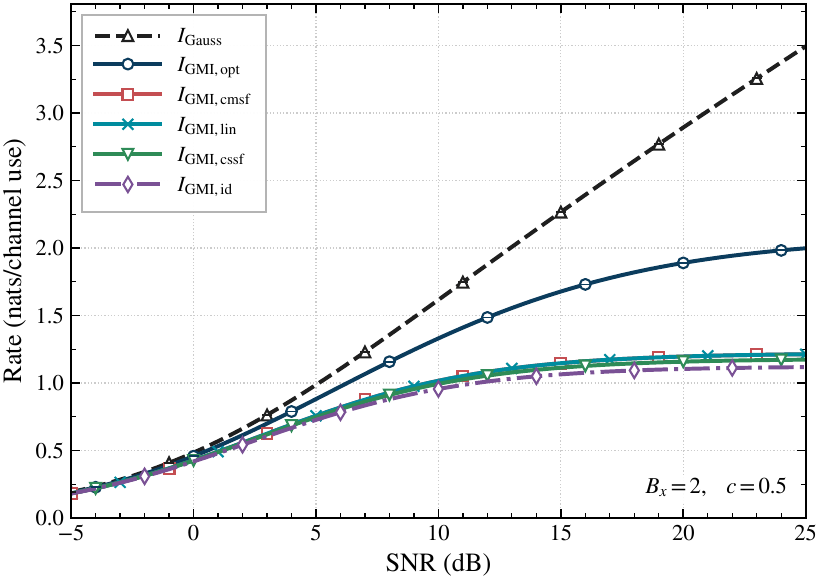}
  }
  \hfill
  \subfloat{
    \includegraphics[width=0.48\linewidth]{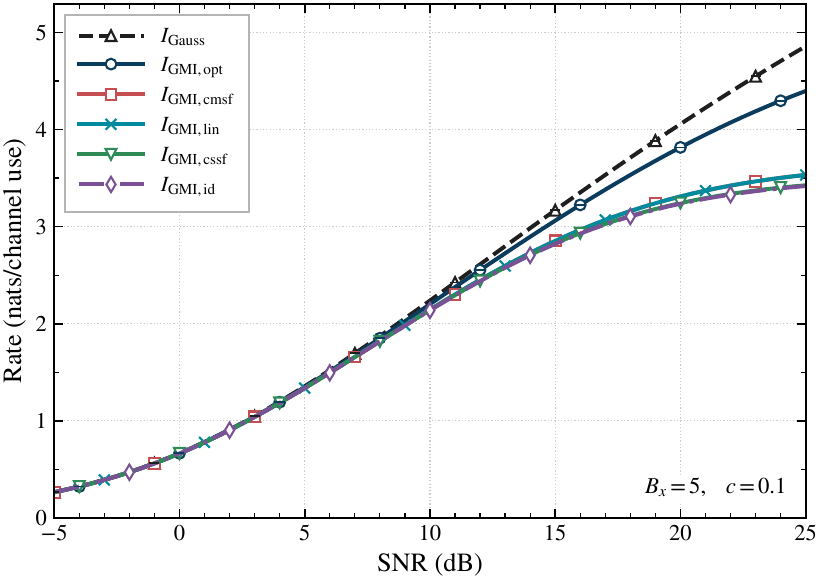}
  }%
  \caption{ Comparison of achievable GMI for the optimal (unrestricted) Vec-GNNDR, its restricted variants, and the identity decoding rule. The results are presented for two configurations: $B_x = 2, c = 0.5$ (left) and $B_x = 5, c = 0.1$ (right). }
  \label{fig:GMI_PNC2}
\end{figure}
Furthermore, Figure~\ref{fig:GMI_PNC2} illustrates a comparison of achievable GMIs across various block lengths $B_x$ and diffusion intensities $c$. It is consistently observed that the optimal Vec-GNNDR and its restricted variants outperform the identity decoding rule. Notably, the performance gap widens as both the block length and diffusion intensity increase compared with Figure~\ref{fig:GMI_PNC1}. This trend arises because larger values of $B_x$ and $c$ induce stronger phase noise, resulting in more severe codeword rotations. Consequently, the identity decoding rule suffers from a greater decoding mismatch, whereas the optimal Vec-GNNDR exhibits superior robustness.

\section{Conclusion and further work} \label{sec:conclusion}
The GNNDR offers a simple yet effective approach to enhancing achievable rates under mismatched decoding while retaining the structural simplicity of NN decoding, making it attractive from an implementation perspective. Existing formulations and analyses of GNNDR, however, have largely focused on symbol-wise memoryless channels. We propose the Vec-GNNDR formalism to extend the GNNDR framework to IBM channels, where statistical dependence among symbols within each block can be explicitly exploited. The contribution is therefore not the IBM abstraction alone, but the identification of a GMI-optimal block-level metric that exploits intra-block statistical structure unavailable to element-wise GNNDR. The resulting optimal Vec-GNNDR admits a clear signal-processing interpretation: it performs GNNDR along the principal directions identified by the posterior covariance, thereby exploiting the available block-level information as effectively as possible within the prescribed metric class. In addition, the closed-form metric optimization for each fixed Gaussian covariance naturally leads to a GMI-based codebook-metric design framework; within the diagonal covariance family, we provide first-order self-consistent conditions for stationary points. Numerical results demonstrate significant rate gains over conventional receiver designs, highlighting the strong potential of Vec-GNNDR for practical deployment in realistic communication scenarios.

There are several promising directions for future research. First, it would be of interest to extend the Vec-GNNDR framework to general memory channels, where memory may span across entire codewords rather than being confined within individual blocks. We conjecture that, under suitable dependence and information-stability conditions, an analogous GMI characterization to \eqref{eq:max_GMI_opt} may continue to hold by treating the entire codeword sequence as a single block, with the resulting metric optimization retaining a form similar to that in Theorem~\ref{thm:main_result}. However, it remains an open question to identify the classes of channels with specific memory structures for which the probability convergence analysis in Proposition~\ref{prop:GMI} continues to hold. Second, in the present analysis, the dimension of the output after the processing function is required to be greater than that of the original received signal (i.e., $B_x\geq B$). However, in practical implementations, it may be desirable to compress the transformed output dimension (i.e., $B_x<B$) to facilitate faster decoding, albeit at the cost of potential information loss. An interesting direction for future work is to characterize the optimal scaling and processing functions, as well as the corresponding GMI, under this reduced-dimension setting. Third, the codebook-metric framework developed here leaves the full covariance optimization problem open; establishing global optimality conditions, convergent algorithms, or useful sufficient conditions beyond the diagonal covariance family would substantially strengthen the framework. Finally, while the present analysis is predicated on Gaussian codebook ensembles, practical communication systems typically utilize discrete constellations.  Accordingly, extending the Vec-GNNDR formalism to accommodate finite constellations—and investigating the joint optimization of the constellation geometry and decoding metrics therein—constitutes an important direction for subsequent research.

From a practical standpoint, several implementation challenges remain for the Vec-GNNDR. The computation of the optimal scaling and processing functions requires evaluating the conditional mean and the conditional covariance matrix. For general channel models, obtaining closed-form expressions for these quantities is often intractable. Although Monte Carlo approximations can be employed, they typically require a large number of samples drawn from the posterior distribution, thereby motivating the development of more efficient sampling algorithms \cite{liu2001monte}. For scenarios with lower precision requirements, an alternative practical approach is to employ machine learning techniques to train a neural network that learns the mapping from the received signals and CSIs to the conditional mean and conditional covariance; see, for example, \cite{pang2025generalized}.

\section*{Acknowledgement}
This work is supported by the National Natural Science Foundation of China (NSFC12271289 and NSFC62231022). The authors thank Lingyi Chen, Weiye Ren, and Shitong Wu for helpful discussions.
\clearpage
\appendices
\section{Auxiliary Lemmas } \label{app:lemma_main}
In this section, we present several auxiliary lemmas that play a crucial role in the derivation of the GMI and in the optimization of the Vec-GNNDR, both in its unrestricted and restricted forms.
\begin{lemma}[Strong Law of Large Numbers, Theorem 2.4.1 in \cite{durrett2019probability}]
    Let $X_1, X_2, \cdots$ be independent
    identically distributed random variables with $\mathbb{E}|X_i| < +\infty$. Let $\EE X_i = \mu$ and $S_n = X_1 + X_2 + \cdots + X_n$. Then 
    \begin{equation*}
        \frac{S_n}{n} \to \mu ,\quad \mathrm{a.s.}
    \end{equation*}
    as $n \to \infty$.
    \label{lem:SLLN}
\end{lemma}


\begin{lemma}[Rayleigh quotient, Theorem 4.2.2 in \cite{horn2012matrix}]
    Let $\bm{A} \in \mathbb{C}^{n \times n}$ be a Hermitian matrix. Then
    \begin{equation*}
        \lambda_{\max}(\bm{A}) 
        = \max_{\bw \neq 0} \frac{\bw^H \bm{A} \bw}{\bw^H \bw},
    \end{equation*}
    where $\lambda_{\max}(\bm{A})$ denotes the largest eigenvalue of $\bm{A}$.
    \label{lem:ray_quoti}
\end{lemma}

\begin{lemma}[Karamata's Inequality, Theorem 12.4 in \cite{cvetkovski2012inequalities}]
    Let $f : I \to \mathbb{R}$ be a convex function on an interval $I \subseteq \mathbb{R}$. 
    Suppose that $\{x_i\}_{i=1}^n$ and $\{y_i\}_{i=1}^n$ are sequences in $I$ such that $\{x_i\}_{i=1}^n$ majorizes $\{y_i\}_{i=1}^n$, i.e.,
    \begin{itemize}
        \item $x_1 \geq x_2 \geq \cdots \geq x_n$ and $y_1 \geq y_2 \geq \cdots \geq y_n$;
        \item $x_1+\cdots +x_i \geq y_1 + \cdots +y_i$ for all $i \in \{1, \cdots ,n-1\}$;
        \item $x_1 +\cdots+x_n = y_1 + \cdots+y_n$.
    \end{itemize}
    Then,
    \begin{equation*}
        f(x_1) + f(x_2) + \cdots +f(x_n) \geq f(y_1) + f(y_2) +\cdots f(y_n).
    \end{equation*}
    If $f$ is a strictly convex function, the inequality holds  if and only if we have $x_i =y_i$ for all $i \in \{1, \cdots, n\}$. 
    \label{lem:karamata}
\end{lemma}

\begin{lemma}[Schur–Horn Theorem, Theorem 5 in \cite{horn1954doubly}]
   Let $\{d_i\}_{i=1}^n$ and $\{\lambda_i\}_{i=1}^n$ be two sequences of real numbers arranged in non-increasing order. There is a Hermitian matrix with diagonal values $d_1, \cdots, d_n $ and eigenvalues $\lambda_1, \cdots, \lambda_n$ if and only if 
   \begin{equation*}
       d_1+\cdots + d_i  \leq \lambda_1 + \cdots \lambda_i, \quad i \in \{1, \cdots, n-1\},
   \end{equation*}
   and 
   \begin{equation*}
        d_1+\cdots + d_n  = \lambda_1 + \cdots +\lambda_n.
   \end{equation*}
   \label{lem:schur-horn}
\end{lemma}

\begin{lemma}[Convexity of $\psi$ in \eqref{eq:psi_thm1}]
\label{lem:psi_convex}
For $P>0$, the function
\begin{equation*}
\psi(q)=
\begin{cases}
\displaystyle \log\frac{P}{q}-1+\frac{q}{P}, & 0<q<P,\\[6pt]
0, & q\ge P,
\end{cases}
\end{equation*}
is continuously differentiable and convex on $(0,\infty)$.
\end{lemma}

\begin{proof}
The two branches agree at $q=P$, and the left and right derivatives
at this point both equal zero. Hence, $\psi$ is continuously
differentiable on $(0,\infty)$, with
\begin{equation*}
\psi'(q)=
\begin{cases}
\displaystyle -\frac{1}{q}+\frac{1}{P}, & 0<q<P,\\[6pt]
0, & q\ge P.
\end{cases}
\end{equation*}
On $(0,P)$, the derivative $\psi'$ is strictly increasing and
converges to zero as $q\uparrow P$, whereas it remains zero on
$[P,\infty)$. Therefore, $\psi'$ is nondecreasing on $(0,\infty)$,
which establishes the convexity of $\psi$.
\end{proof}

\begin{lemma}[Spectral Majorization Inequality]
    Let $\bm{A} \in \mathbb{C}^{n \times n}$ be a positive definite matrix, and let $f : (0,\infty) \to \mathbb{R}$ be a convex function. 
    Then, for any unitary matrix $\bm{W} = [\bw_1, \dots, \bw_n] \in \mathbb{C}^{n \times n}$, we have 
    \begin{equation*}
        \sum_{i=1}^n f(\bw_i^H \bm{A} \bw_i) \leq \sum_{i=1}^n f(\lambda_i(\bm{A})),
    \end{equation*}
    where $\lambda_i(\bm{A})$ denotes the $i$-th largest eigenvalue of $\bm{A}$. Equality holds if 
    \begin{equation*}
        \bm{W}^H \bm{A} \bm{W} = \mathrm{diag}\{\lambda_{\pi(1)}(\bm{A}), \dots, \lambda_{\pi(n)}(\bm{A})\}
    \end{equation*}
    for some permutation $\pi \in S_n$.
    \label{lem:spec_maj_inequal}
\end{lemma}
\begin{proof}
    Without loss of generality,  assume that  
    \begin{equation*}
        \bw_1^H \bm{A} \bw_1 \geq \cdots \geq \bw_n^H \bm{A} \bw_n.
    \end{equation*}
    Since $\bm{W}^H \bm{A} \bm{W}$ is unitarily similar to $\bm{A}$, it has the same multiset of eigenvalues as $\bm{A}$. By the Schur-Horn theorem (Lemma~\ref{lem:schur-horn}), we obtain
    \begin{equation*}
        \bw_1^H \bm{A} \bw_1 +\cdots + \bw_i^H \bm{A} \bw_i \leq  \lambda_1(\bm{A}) + \cdots \lambda_i(\bm{A}), \quad i\in \{1, \cdots, n-1\}.
    \end{equation*}
    Applying Karamata’s inequality (Lemma~\ref{lem:karamata}) then gives 
    \begin{equation}
        \sum_{i=1}^n f(\bw_i^H \bm{A} \bw_i) \leq \sum_{i=1}^n f(\lambda_i(\bm{A})).
    \end{equation}
     Finally, equality is achieved when $\bm{W}$ diagonalizes $\bm{A}$, i.e.,
    \[
        \bm{W}^H \bm{A} \bm{W} = \mathrm{diag}\big(\lambda_1(\bm{A}), \dots, \lambda_n(\bm{A})\big).
    \]
    Since the ordering of the diagonal entries is immaterial due to the permutation invariance of the summation, the equality condition holds up to any permutation $\pi \in S_n$ 
\end{proof}

\begin{lemma}[Maximum eigenvalue inequality]
   Let $\bm{A} \in \CC^{n \times n}$ be a Hermitian matrix and $\bm{B} \in \CC^{n \times n}$ a positive semi-definite matrix. Then
    \begin{equation*}
        \lambda_{\max}(\bm{A} - {\bm{B}}) \leq \lambda_{\max}(\bm{A}).
    \end{equation*}
    \label{lem:max_eigen_inequal}
\end{lemma}
\begin{proof}
    By the Rayleigh quotient characterization (Lemma~\ref{lem:ray_quoti}), there exists a unit-norm vector $\bw$ such that
    \begin{equation*}
        \lambda_{\max}(\bm{A} - \bm{B}) = \bw^H (\bm{A} - \bm{B}) \bw.
    \end{equation*}
    Applying the Rayleigh quotient to $\bm{A}$ yields
    \begin{equation*}
        \lambda_{\max}(\bm{A}) 
        \geq \bw^H \bm{A} \bw
        = \bw^H (\bm{A} - \bm{B}) \bw + \bw^H \bm{B} \bw
        = \lambda_{\max}(\bm{A} - \bm{B}) + \bw^H \bm{B} \bw.
    \end{equation*}
    Since $\bm{B}$ is positive semi-definite, we have $\bw^H \bm{B} \bw \geq 0$, and thus
    \[
        \lambda_{\max}(\bm{A}) \geq \lambda_{\max}(\bm{A} - \bm{B}),
    \]
    which completes the proof.
\end{proof}

\begin{lemma}
    Let $\bX \in \mathbb{C}^{B_x}, \bY \in \mathbb{C}^{B_y}$ and $\bV$ be complex random vectors, where $\mathbb{E}[\bX \bX^H] = P \bI_{B_x}$ and $\bX$ is independent of $\bV$.  Then, for any realization $\bv$ of $\bV$, we have 
    \begin{equation*}
        \lambda_{\max} \left( \mathbb{E}[\bX \bY^H \mid \bv]   \left( \mathbb{E}[\bY \bY^H \mid \bv]\right)^{\dagger} \mathbb{E}[\bY \bX^H \mid \bv] \right) \leq P.
    \end{equation*}
    Moreover, if $\mathrm{Cov}[\bX|\by, \bv] \succ 0$ for any realization $(\by, \bv) $, then the inequality is strict.
    \label{lem:max_eig_cond_cov}
\end{lemma}
\begin{proof}
    Define the conditional covariance matrices
    \begin{equation*}
        \bm{\Sigma}_{\bX \bX \mid \bv} = \mathbb{E}[\bX \bX^H \mid \bv], \quad \bm{\Sigma}_{\bX \bY \mid \bv} = \mathbb{E}[\bX \bY^H \mid \bv], \quad \bm{\Sigma}_{\bY \bY \mid \bv} = \mathbb{E}[\bY \bY^H \mid \bv].
    \end{equation*}
    Since $\bX \perp \bV$, we have
    \begin{equation*}
        \bm{\Sigma}_{\bX \bX \mid \bv} = \mathbb{E}[\bX \bX^H \mid \bv] = \mathbb{E}[\bX \bX^H] = P \boldsymbol{I}_{B_x}. 
    \end{equation*}

    For any $\bu \in \mathbb{C}^{B_x}$ and $\bc \in \mathbb{C}^{B_y}$, it holds that
    \begin{equation}
        \mathbb{E} \left[ |\bu^H \bX - \bm{c}^H \bY|^2 \mid \bv \right] \geq 0,
        \label{eq:aux_proof}
    \end{equation}
    with a fixed $\bv$.
    Expanding the quadratic form and taking conditional expectation yields
    \begin{equation}
        \bu^H \bm{\Sigma}_{\bX \bX \mid \bv} \bu -  2 \Re \left\{ \bu^H \bm{\Sigma}_{\bX \bY \mid \bv} \bm{c} \right\} + \bm{c}^H \bm{\Sigma}_{\bY \bY \mid \bv} \bm{c} \geq 0,
        \label{eq:aux_proof_2}
    \end{equation}

    Minimizing the left-hand side with respect to $\bc$, we obtain
    \begin{equation*}
        \bu^H  \left( \bm{\Sigma}_{\bX \bX \mid \bv} -  \bm{\Sigma}_{\bX \bY \mid \bv} \bm{\Sigma}_{\bY \bY \mid \bv}^{\dagger } \bm{\Sigma}_{\bX \bY \mid \bv}^H\right) \bm{u} \geq 0,
    \end{equation*}
    for all $\bu \in \mathbb{C}^{B_x}$.  Hence,
    \begin{equation*}
         \bm{\Sigma}_{\bX \bY \mid \bv} \bm{\Sigma}_{\bY \bY \mid \bv}^{\dagger } \bm{\Sigma}_{\bX \bY \mid \bv}^H \preccurlyeq  \bm{\Sigma}_{\bX \bX \mid \bv} = P \boldsymbol{I}_{B_x}, 
    \end{equation*}
    which immediately implies
    \begin{equation*}
        \lambda_{\max} \left(  \bm{\Sigma}_{\bX \bY \mid \bv} \bm{\Sigma}_{\bY \bY \mid \bv}^{\dagger } \bm{\Sigma}_{\bX \bY \mid \bv}^H\right) \leq P.
    \end{equation*}

   Next, we establish that if $\mathrm{Cov}[\bX \mid \by, \bv] \succ 0$, then the inequality in \eqref{eq:aux_proof} holds strictly for any $\bu \in \mathbb{C}^{B_x} \backslash {\boldsymbol{0}}$ and $\bm{c} \in \mathbb{C}^{B_y}$. We proceed by contradiction. Suppose that there exist $\bu \in \mathbb{C}^{B_x} \backslash {\boldsymbol{0}}$ and $\bm{c} \in \mathbb{C}^{B_y}$ such that
    \begin{equation*}
        \mathbb{E} \left[ |\bu^H \bX - \bm{c}^H \bY|^2 \mid \bv \right] = 0,
    \end{equation*}
    Then it must follow that
    \begin{equation*}
        \mathrm{Var} [\bu^H \bX -\bm{c}^H \bY \mid \bv] = 0.
    \end{equation*}
    Consequently, there exists $\by$ such that 
    \begin{equation*}
       \mathrm{Var} [\bu^H \bX -\bm{c}^H \bY \mid \by, \bv] = 0,
    \end{equation*}
    which implies 
    \begin{equation*}
        \bu^H \mathrm{Cov} [\bX \mid \by, \bv] \bu =  \mathrm{Var} [\bu^H \bX  \mid \by, \bv] = \mathrm{Var} [\bu^H  \bX - \bm{c}^H \bY  \mid \by, \bv] = 0,
    \end{equation*}
    This contradicts the assumption $\mathrm{Cov}[\bX \mid \by, \bv] \succ 0$. Hence, the inequality in \eqref{eq:aux_proof} is strict for all $\bu \in \mathbb{C}^{B_x} \backslash {\boldsymbol{0}}$. It follows that the inequality in \eqref{eq:aux_proof_2} is also strict for all $\bu \in \mathbb{C}^{B_x} \backslash {\boldsymbol{0}}$, and therefore
    \begin{equation*}
        \lambda_{\max} \left(  \bm{\Sigma}_{\bX \bY \mid \bv} \bm{\Sigma}_{\bY \bY \mid \bv}^{\dagger } \bm{\Sigma}_{\bX \bY \mid \bv}^H\right) < P,
    \end{equation*}
    which completes the proof.
\end{proof}

\begin{lemma}[Sherman-Morrison-Woodbury formula, Section 0.7.4 in~\cite{horn2012matrix}]
    Let $\bm{A}\in\mathbb{C}^{n\times n}$ and $\bm{R}\in\mathbb{C}^{r\times r}$ be invertible, and let $\bm{X}\in\mathbb{C}^{n\times r}$ and $\bm{Y}\in\mathbb{C}^{r\times n}$. Define $\bm{B}:=\bm{A}+\bm{X}\bm{R}\bm{Y}$. If $\bm{R}^{-1}+\bm{Y}\bm{A}^{-1}\bm{X}$ is invertible (equivalently, $\bm{B}$ is invertible), then 
    \begin{equation*} 
    \bm{B}^{-1} = \bm{A}^{-1} - \bm{A}^{-1}\bm{X}\big(\bm{R}^{-1}+\bm{Y}\bm{A}^{-1}\bm{X}\big)^{-1}\bm{Y}\bm{A}^{-1} = \bm{A}^{-1} - \bm{A}^{-1}\bm{X}\bm{R}\big(\bm{I}_r+\bm{Y}\bm{A}^{-1}\bm{X}\bm{R}\big)^{-1}\bm{Y}\bm{A}^{-1}. \end{equation*}
    \label{lem:wood_inv}
\end{lemma}

\begin{lemma}[Orthogonality
principle, Result 5.2 in \cite{schreier2010statistical}]
     Let $\bX$ and $\bY$ be two random vectors. For any measurable function $\bg(\by)$, the following orthogonality relation holds:
    \begin{equation*}
        \mathbb{E} \!\left[ \bg^H(\bY)\bigl(\mathbb{E}[\bX \mid \bY] - \bX \bigr) \right] = 0.
    \end{equation*}
    In particular, by choosing $\bg(\by) = \mathbb{E}[\bX \mid \by]$, we obtain
    \begin{equation*}
        \mathbb{E} \!\left[ \bigl|\mathbb{E}[\bX \mid \bY]\bigr|^2 \right] 
        = \mathbb{E} \!\left[ \mathbb{E}[\bX^H \mid \bY]\, \bX \right].
    \end{equation*}
    \label{lem:ortho}
\end{lemma}

\section{Codebook Optimization}
\label{app:cb_opt}

In this appendix, we derive the gradient of
$J(\bm{\Sigma})$ in \eqref{eq:cb-objective}, thereby establishing
\eqref{eq:cb-final-gradient}.
Throughout, $\bm{\Sigma}
=\operatorname{diag}(\lambda_1,\ldots,\lambda_{B_x})\succ0$,
and all posterior moments and expectations are evaluated under the
joint law induced by
$\mathbf{X}\sim\mathcal{CN}(\mathbf{0},P\bm{\Sigma})$.
Derivatives are taken with respect to the real parameter $\lambda_i$,
holding the physical input and observation variables fixed.
Assume that all displayed expectations are finite and that differentiation may be interchanged with integration and expectation locally within the positive diagonal covariance family.

\subsection{Derivatives of the Posterior Moments}

Let $p(\mathbf{x};\bm{\Sigma})$ denote the input density,
\[
p(\mathbf{x};\bm{\Sigma})
=
\prod_{j=1}^{B_x}\frac{1}{\pi P\lambda_j}
\exp\!\left(-\frac{|x_j|^2}{P\lambda_j}\right).
\]
Recalling the score function in \eqref{eq:cb-scores}, we have
\begin{equation}
\frac{\partial p(\mathbf{x};\bm{\Sigma})}{\partial\lambda_i}
=
p(\mathbf{x};\bm{\Sigma})S_i(\mathbf{x}),
\qquad
S_i(\mathbf{x})
=
\frac{|x_i|^2}{P\lambda_i^2}-\frac{1}{\lambda_i}.
\label{eq:cb-app-prior-score}
\end{equation}
The effective channel likelihood
$p(\mathbf{y},\mathbf{v}\mid\mathbf{x})$
does not depend on $\bm{\Sigma}$ when $\mathbf{x}$ is fixed.
Thus, differentiating the marginal density
\[
p(\mathbf{y},\mathbf{v};\bm{\Sigma})
=
\int p(\mathbf{y},\mathbf{v}\mid\mathbf{x})
p(\mathbf{x};\bm{\Sigma})\,d\mathbf{x}
\]
gives
\begin{align*}
\frac{\partial p(\mathbf{y},\mathbf{v};\bm{\Sigma})}
{\partial\lambda_i}
=
\int p(\mathbf{y},\mathbf{v}\mid\mathbf{x})
p(\mathbf{x};\bm{\Sigma})S_i(\mathbf{x})\,d\mathbf{x}=
p(\mathbf{y},\mathbf{v};\bm{\Sigma})
\mathbb{E}_{\mathbf{X}\mid\mathbf{y},\mathbf{v}}
[S_i(\mathbf{X})].
\end{align*}
Consequently,
\begin{equation}
\frac{\partial\log p(\mathbf{y},\mathbf{v};\bm{\Sigma})}
{\partial\lambda_i}
=
\mathbb{E}_{\mathbf{X}\mid\mathbf{y},\mathbf{v}}
[S_i(\mathbf{X})].
\label{eq:cb-marginal-score}
\end{equation}
By Bayes' rule, the posterior density is
\[
p(\mathbf{x}\mid\mathbf{y},\mathbf{v};\bm{\Sigma})
=
\frac{
p(\mathbf{y},\mathbf{v}\mid\mathbf{x})
p(\mathbf{x};\bm{\Sigma})
}{
p(\mathbf{y},\mathbf{v};\bm{\Sigma})
}.
\]
Applying the quotient rule and substituting
\eqref{eq:cb-app-prior-score} and
\eqref{eq:cb-marginal-score}, we obtain
\begin{equation}
\begin{aligned}
\frac{\partial
p(\mathbf{x}\mid\mathbf{y},\mathbf{v};\bm{\Sigma})}
{\partial\lambda_i}=
p(\mathbf{x}\mid\mathbf{y},\mathbf{v};\bm{\Sigma})
\left(
S_i(\mathbf{x})
-\mathbb{E}_{\mathbf{X}\mid\mathbf{y},\mathbf{v}}
[S_i(\mathbf{X})]
\right)=
p(\mathbf{x}\mid\mathbf{y},\mathbf{v};\bm{\Sigma})
\hat{S}_i(\mathbf{x}\mid\mathbf{y},\mathbf{v}),
\end{aligned}
\label{eq:cb-posterior-score}
\end{equation}
where the centered score satisfies
\begin{equation}
\mathbb{E}_{\mathbf{X}\mid\mathbf{y},\mathbf{v}}\!
[\hat{S}_i(\mathbf{X}\mid\mathbf{y},\mathbf{v})]=0.
\label{eq:cb-app-centered-score}
\end{equation}
Next, differentiating the conditional mean yields
\begin{equation}
\begin{aligned}
\frac{\partial\bm{\mu}_{\mathrm{color}}(\mathbf{y},\mathbf{v})}
{\partial\lambda_i}
&=
\int \mathbf{x}
\frac{\partial
p(\mathbf{x}\mid\mathbf{y},\mathbf{v};\bm{\Sigma})}
{\partial\lambda_i}\,d\mathbf{x}=
\mathbb{E}_{\mathbf{X}\mid\mathbf{y},\mathbf{v}}\!\left[
\mathbf{X}\hat{S}_i(\mathbf{X}\mid\mathbf{y},\mathbf{v})
\right].
\end{aligned}
\label{eq:cb-app-mean-derivative}
\end{equation}
We now differentiate
\[
\bm{\Sigma}_{\mathrm{color}}(\mathbf{y},\mathbf{v})
=
\mathbb{E}_{\mathbf{X}\mid\mathbf{y},\mathbf{v}}
[\mathbf{X}\mathbf{X}^H]
-
\bm{\mu}_{\mathrm{color}}(\mathbf{y},\mathbf{v})
\bm{\mu}_{\mathrm{color}}^H(\mathbf{y},\mathbf{v}).
\]
Suppressing the arguments $(\mathbf{y},\mathbf{v})$ of the posterior
moments, the product rule gives
\begin{align*}
\frac{\partial\bm{\Sigma}_{\mathrm{color}}}{\partial\lambda_i}
=
\int \mathbf{x}\mathbf{x}^H
\frac{\partial
p(\mathbf{x}\mid\mathbf{y},\mathbf{v};\bm{\Sigma})}
{\partial\lambda_i}\,d\mathbf{x}-
\frac{\partial\bm{\mu}_{\mathrm{color}}}{\partial\lambda_i}
\bm{\mu}_{\mathrm{color}}^H
-
\bm{\mu}_{\mathrm{color}}
\frac{\partial\bm{\mu}_{\mathrm{color}}^H}{\partial\lambda_i}.
\end{align*}
Since $\lambda_i$ and the scores are real, differentiation commutes
with conjugate transpose.
Substituting \eqref{eq:cb-posterior-score} and
\eqref{eq:cb-app-mean-derivative} yields
\begin{align*}
\frac{\partial\bm{\Sigma}_{\mathrm{color}}}{\partial\lambda_i}
=
\mathbb{E}_{\mathbf{X}\mid\mathbf{y},\mathbf{v}}
[\mathbf{X}\mathbf{X}^H\hat{S}_i]-
\mathbb{E}_{\mathbf{X}\mid\mathbf{y},\mathbf{v}}
[\mathbf{X}\hat{S}_i]\bm{\mu}_{\mathrm{color}}^H-
\bm{\mu}_{\mathrm{color}}
\mathbb{E}_{\mathbf{X}\mid\mathbf{y},\mathbf{v}}
[\mathbf{X}^H\hat{S}_i],
\end{align*}
where $\hat{S}_i=\hat{S}_i(\mathbf{X}\mid\mathbf{y},\mathbf{v})$.
Using \eqref{eq:cb-app-centered-score}, we may add the zero matrix
\[
\bm{\mu}_{\mathrm{color}}\bm{\mu}_{\mathrm{color}}^H
\mathbb{E}_{\mathbf{X}\mid\mathbf{y},\mathbf{v}}[\hat{S}_i]
=\mathbf{0}
\]
to obtain
\begin{equation}
\begin{aligned}
\frac{\partial\bm{\Sigma}_{\mathrm{color}}(\mathbf{y},\mathbf{v})}
{\partial\lambda_i}
=
\mathbb{E}_{\mathbf{X}\mid\mathbf{y},\mathbf{v}}\Bigl[
(\mathbf{X}-\bm{\mu}_{\mathrm{color}})
(\mathbf{X}-\bm{\mu}_{\mathrm{color}})^H \hat{S}_i(\mathbf{X}\mid\mathbf{y},\mathbf{v})
\Bigr].
\end{aligned}
\label{eq:cb-colored-covariance-derivative}
\end{equation}

\subsection{Derivative of the Objective}

We first account for the explicit dependence of the whitening
transformation on $\lambda_i$.
Since $\bm{\Sigma}$ is diagonal,
\[
\frac{\partial\bm{\Sigma}^{-1/2}}{\partial\lambda_i}
=
-\frac{1}{2\lambda_i}\mathbf{E}_{ii}\bm{\Sigma}^{-1/2}
=
-\frac{1}{2\lambda_i}\bm{\Sigma}^{-1/2}\mathbf{E}_{ii},
\]
where $\mathbf{E}_{ii}$ has a single nonzero entry, equal to one, at
position $(i,i)$.
Applying the product rule to
$\hat{\bm{\Sigma}}
=\bm{\Sigma}^{-1/2}\bm{\Sigma}_{\mathrm{color}}
 \bm{\Sigma}^{-1/2}$ gives
\begin{equation}
\begin{aligned}
\frac{\partial\hat{\bm{\Sigma}}}{\partial\lambda_i}
=
-\frac{1}{2\lambda_i}
\left(
\mathbf{E}_{ii}\hat{\bm{\Sigma}}
+\hat{\bm{\Sigma}}\mathbf{E}_{ii}
\right)+
\bm{\Sigma}^{-1/2}
\frac{\partial\bm{\Sigma}_{\mathrm{color}}}{\partial\lambda_i}
\bm{\Sigma}^{-1/2}.
\end{aligned}
\label{eq:cb-whitened-covariance-derivative}
\end{equation}
The function $G$ in \eqref{eq:cb-G} is continuously differentiable
on $(0,\infty)$, including at $t=P$.
Differentiating \eqref{eq:cb-objective}, while accounting for the
dependence of the output--CSI distribution on $\bm{\Sigma}$, yields
\begin{equation}
\begin{aligned}
\frac{\partial J}{\partial\lambda_i}
=
\mathbb{E}_{\mathbf{Y},\mathbf{V}}\Biggl[
\operatorname{Tr}\!\left(
G'(\hat{\bm{\Sigma}})
\frac{\partial\hat{\bm{\Sigma}}}{\partial\lambda_i}
\right)+
\operatorname{Tr}G(\hat{\bm{\Sigma}})
\frac{\partial\log p(\mathbf{Y},\mathbf{V};\bm{\Sigma})}
{\partial\lambda_i}
\Biggr],
\end{aligned}
\label{eq:cb-app-objective-derivative}
\end{equation}
where the arguments $(\mathbf{Y},\mathbf{V})$ of
$\hat{\bm{\Sigma}}$ are suppressed. To simplify the first trace term, observe that
\[
G'(t)=
\begin{cases}
-1/t,&0<t<P,\\
-1/P,&t\geq P,
\end{cases}
\qquad
tG'(t)=-\frac{P+(t-P)_+}{P}.
\]
By spectral calculus,
\[
\hat{\bm{\Sigma}}G'(\hat{\bm{\Sigma}})
=
-\frac{1}{P}\left(
P\mathbf{I}_{B_x}
+\bigl(\hat{\bm{\Sigma}}-P\mathbf{I}_{B_x}\bigr)_+
\right).
\]
Consequently, cyclicity of the trace gives
\begin{equation}
\begin{aligned}
\operatorname{Tr}\!\left[
G'(\hat{\bm{\Sigma}})
\left(
-\frac{
\mathbf{E}_{ii}\hat{\bm{\Sigma}}
+\hat{\bm{\Sigma}}\mathbf{E}_{ii}
}{2\lambda_i}
\right)
\right]=
-\frac{1}{\lambda_i}
\left[
\hat{\bm{\Sigma}}G'(\hat{\bm{\Sigma}})
\right]_{ii}=
\frac{
P+\left[
\bigl(\hat{\bm{\Sigma}}-P\mathbf{I}_{B_x}\bigr)_+
\right]_{ii}
}{P\lambda_i}.
\end{aligned}
\label{eq:cb-whitening-trace}
\end{equation}
For the remaining covariance term, use
\[
\hat{\mathbf{X}}-\hat{\bm{\mu}}
=
\bm{\Sigma}^{-1/2}
(\mathbf{X}-\bm{\mu}_{\mathrm{color}})
\]
and substitute \eqref{eq:cb-colored-covariance-derivative}.
Applying cyclicity of the trace and the definition of $q$ in
\eqref{eq:cb-q}, we obtain
\begin{equation}
\begin{aligned}
&\operatorname{Tr}\!\left[
G'(\hat{\bm{\Sigma}})
\bm{\Sigma}^{-1/2}
\frac{\partial\bm{\Sigma}_{\mathrm{color}}}{\partial\lambda_i}
\bm{\Sigma}^{-1/2}
\right]=
\mathbb{E}_{\mathbf{X}\mid\mathbf{y},\mathbf{v}}\Bigl[
(\hat{\mathbf{X}}-\hat{\bm{\mu}})^H
G'(\hat{\bm{\Sigma}})
(\hat{\mathbf{X}}-\hat{\bm{\mu}})
\hat{S}_i(\mathbf{X}\mid\mathbf{y},\mathbf{v})
\Bigr]=
\mathbb{E}_{\mathbf{X}\mid\mathbf{y},\mathbf{v}}\!\left[
q(\hat{\mathbf{X}}\mid\mathbf{y},\mathbf{v})
\hat{S}_i(\mathbf{X}\mid\mathbf{y},\mathbf{v})
\right].
\end{aligned}
\label{eq:cb-posterior-trace}
\end{equation}
Finally, substituting
\eqref{eq:cb-whitened-covariance-derivative},
\eqref{eq:cb-marginal-score},
\eqref{eq:cb-whitening-trace}, and
\eqref{eq:cb-posterior-trace}
into \eqref{eq:cb-app-objective-derivative} yields
\begin{equation}
\begin{aligned}
\frac{\partial J}{\partial\lambda_i}
=
\frac{
P+\mathbb{E}_{\mathbf{Y},\mathbf{V}}\!\left[
\left[
\bigl(\hat{\bm{\Sigma}}(\mathbf{Y},\mathbf{V})
-P\mathbf{I}_{B_x}\bigr)_+
\right]_{ii}
\right]
}{P\lambda_i}+\Delta_i(\bm{\Sigma}),
\end{aligned}
\label{eq:cb-app-objective-gradient}
\end{equation}
where the posterior and marginal score contributions are exactly
$\Delta_i(\bm{\Sigma})$ in \eqref{eq:cb-delta}.
Since
$\partial\mathcal{L}/\partial\lambda_i
=\partial J/\partial\lambda_i-\nu$,
this establishes \eqref{eq:cb-final-gradient}.

\section{Case Study} \label{app:main}
In this section, we provide a detailed derivation of the optimal processing function $\bg$, scaling function $\bm{f}$ and the corresponding optimal GMI for both memory channels and ACGNC. 
\subsection{Proof for Corollary~\ref{coral:memoryless} (Memoryless channel)}
We first demonstrate how the optimal processing function $g$ and the scaling function $f$ for GNNDR \cite{wang2022generalized} can be recovered within our Vec-GNNDR formalism by proving Corollary~\ref{coral:memoryless}.
\begin{proof}
     Under the memoryless model in~\eqref{eq:memoryless} with block length $B_x = 1$, the conditional covariance matrix in~\eqref{eq:cond_mean_cov} reduces to the scalar conditional variance
    \begin{equation*}
        \omega(y, v) := \operatorname{Var}[X \mid y,v] 
        = \mathbb{E}[|X|^2 \mid y, v] - \big|\mathbb{E}[X \mid y, v]\big|^2,
    \end{equation*}
    so that $\lambda_1(y,v) = \omega(y,v)$, and the $\bW(y, v)$ can be set to $1$.

    Substituting $\omega(y,v)$ and $\bW(y, v) = 1$ into~\eqref{eq:opt_g_eps} and~\eqref{eq:opt_f_eps} yields the optimal scalar processing and scaling functions:
    \begin{equation*}
        f^\varepsilon(y, v) 
        = \sqrt{\frac{P - [w(y,v)]_\varepsilon}{P \,[w(y,v)]_\varepsilon}}, 
        \qquad 
        g^\varepsilon(y, v) 
        = \sqrt{\frac{P}{(P - [w(y,v)]_\varepsilon)\,[w(y,v)]_\varepsilon}}\, \mu(y, v),
    \end{equation*}
    where $\delta_\varepsilon(y,v)$ is defined in Theorem~\ref{thm:main_result}, obtained by substituting $\lambda_1(\by, \bv)$ in~\eqref{eq:del_eps_main} with $w(y,v)$.

    Finally, inserting $w(y,v)$ into~\eqref{eq:IGMI} gives the optimal GMI as
    \begin{equation*}
        I_{\mathrm{GMI,cssf}} = \mathbb{E}_{Y, V} \left[ \psi(w(Y,V)) + \frac{1}{P} |\mu(Y,V)|^2\right],
    \end{equation*}
    which completes the proof.
\end{proof}
\subsection{Proof for Corollary~\ref{corollary:acgwc} (ACGNC)}
We next demonstrate that the optimal Vec-GNNDR coincides with maximum-likelihood decoding for the additive colored Gaussian noise channel by proving Corollary~\ref{corollary:acgwc}.
\begin{proof}
    By the properties of complex Gaussian distributions, with
    $\bm{\Sigma}\succ0$, the conditional distribution of $\bX$ given
    $\by$ is
    \begin{equation*}
        \bX \mid \by
        \sim \mathcal{CN}(\hat{\bm{\mu}},\hat{\bm{\Sigma}}),
    \end{equation*}
    where the conditional mean $\hat{\bm{\mu}}$ and covariance
    $\hat{\bm{\Sigma}}$ are given by
    \begin{equation*}
        \begin{aligned}
            \hat{\bm{\mu}}
            &= \bA^H
            (\bA\bA^H+P^{-1}\bm{\Sigma})^{-1}\by
            = \hat{\bm{\Sigma}}\bA^H\bm{\Sigma}^{-1}\by,\\
            \hat{\bm{\Sigma}}
            &= P\bI_{B_x}
            -P\bA^H(\bA\bA^H+P^{-1}\bm{\Sigma})^{-1}\bA= (\bA^H\bm{\Sigma}^{-1}\bA
            +P^{-1}\bI_{B_x})^{-1}.
        \end{aligned}
    \end{equation*}
    We define $\bM=\bm{\Sigma}^{-1/2}\bA$ and write its singular
    value decomposition as $\bM=\bU\bm{\Lambda}\bW^H$, where
    $\bm{\Lambda}=\operatorname{diag}\{\sigma_1,\ldots,\sigma_{B_x}\}$
    with $\sigma_1\geq\cdots\geq\sigma_{B_x}\geq0$.
    Then we can rewrite $\hat{\bm{\mu}}$ and
    $\hat{\bm{\Sigma}}$ as
    \begin{equation*}
        \begin{aligned}
            \hat{\bm{\mu}}
            &= (\bM^H\bM+P^{-1}\bI_{B_x})^{-1}
            \bA^H\bm{\Sigma}^{-1}\by= \bW\operatorname{diag}\left\{
            \frac{P}{1+P\sigma_i^2}\right\}
            \bm{\Lambda}\bU^H\bm{\Sigma}^{-1/2}\by,\\
            \hat{\bm{\Sigma}}
            &= (P^{-1}\bI_{B_x}+\bM^H\bM)^{-1}= \bW\operatorname{diag}\left\{
            \frac{P}{1+P\sigma_i^2}\right\}\bW^H.
        \end{aligned}
    \end{equation*}
    Hence, in the SVD-induced ordering of the eigendirections, we have
    $\EE[\bX\mid\by,\bv]=\hat{\bm{\mu}}$,
    $\bm{\Sigma}(\by,\bv)=\hat{\bm{\Sigma}}$,
    $\bW(\by,\bv)=\bW$, and
    \begin{equation*}
        \lambda_i(\by,\bv)=\frac{P}{1+P\sigma_i^2},
        \qquad
        \sigma_i(\by,\bv)=\sigma_i.
    \end{equation*}
    A common permutation of the eigendirections does not affect
    the decoding rule or its GMI.

    Observe that all posterior eigenvalues lie in $(0,P]$.
    For $\sigma_i>0$, we have $\lambda_i(\by,\bv)<P$, and
    substitution into Theorem~\ref{thm:main_result} gives
    \begin{equation*}
        g_i(\by,\bv)
        =\frac{1+P\sigma_i^2}{P\sigma_i}
        (\bW^H\hat{\bm{\mu}})_i
        =(\bU^H\bm{\Sigma}^{-1/2}\by)_i.
    \end{equation*}
    If $\sigma_i=0$, then $\lambda_i(\by,\bv)=P$ and
    $(\bW^H\hat{\bm{\mu}})_i=0$, so setting the corresponding
    scaling row to zero attains the optimal zero GMI contribution
    in that direction. Its processing component may also be chosen
    as $(\bU^H\bm{\Sigma}^{-1/2}\by)_i$, since this contributes only
    a codeword-independent term to the metric and leaves the GMI
    unchanged. Thus, no eigenvalue truncation is required, and
    an optimal processing--scaling pair is
    \begin{equation*}
        \begin{aligned}
            \bm{f}(\by,\bv)
            =\operatorname{diag}(\sigma_i)\bW^H
            =\bm{\Lambda}\bW^H, \qquad
            \bg(\by,\bv)
            =\bU^H\bm{\Sigma}^{-1/2}\by.
        \end{aligned}
    \end{equation*}
    Finally, using
    $\EE[\|\hat{\bm{\mu}}\|_2^2]
    =PB_x-\operatorname{Tr}(\hat{\bm{\Sigma}})$
    in Theorem~\ref{thm:main_result}, the corresponding optimal
    GMI is
    \begin{equation*}
        \begin{aligned}
            I_{\mathrm{GMI},\bg,\bm{f}}
            &=\frac{1}{B_x}\sum_{i=1}^{B_x}
            \log(1+P\sigma_i^2)=\frac{1}{B_x}\log\det\left(
            P\bA^H\bm{\Sigma}^{-1}\bA+\bI_{B_x}\right),
        \end{aligned}
    \end{equation*}
    which completes the proof.
\end{proof}

\section{Constrained Variants of Vec-GNNDR }
\label{app:restricted_vec_GNNDR}
In this section, we show the derivation for the various constrained forms of Vec-GNNDR in Section \ref{sec:restricted_vec_GNNDR} in detail. 
\subsection{Proof for Proposition~\ref{prop:Vec-GNNDR-cssf} (Constant scalar scaling function)}
\begin{proof}
     Under the decoding rule in~\eqref{eq:Vec-GNNDR-cssf}, we have $\bm{f}(\by, \bv) = \alpha \bI_{B_x}$, which implies $\sigma_i(\by, \bv) = |\alpha|$ and $\bW(\by, \bv) = e^{-i \, \operatorname{Arg}(\alpha)} \bI_{B_x}$.  
    The optimization of $\bg(\by, \bv)$ proceeds analogously to the proof of Theorem~\ref{thm:main_result}, except that the optimization problem in~\eqref{eq:max_2} simplifies to
    \begin{equation} 
        I_{\mathrm{GMI, cssf}} = \max_{|\alpha|} \left[ \log(1+P|\alpha|^2) + B_x^{-1} \frac{1+P|\alpha|^2}{P} \mathbb{E} \left[ \left\| \bm{\mu}(\bY, \bV) \right\|_2^2\right]- B_x^{-1} |\alpha|^2 \mathbb{E}_{\bY, \bV} \left[ \mathbb{E}[\|\bX\|_2^2 \mid \bY, \bV]\right]\right]. 
        \label{eq:cssf_GMI_max}
    \end{equation}
     Differentiating the objective in~\eqref{eq:cssf_GMI_max} with respect to $|\alpha|$ and setting the derivative to zero yields
    \begin{equation*}
        \frac{P}{1 + P|\alpha|^2} -  T = 0,
    \end{equation*}
    whose solution is
    \begin{equation*}
        |\alpha|^2 = \frac{P - T}{P T},
    \end{equation*}
    where we denote $T = B_x^{-1}\mathrm{Tr}[\bm{\Sigma}]$. Moreover, 
    \begin{equation*}
        T = B_{x}^{-1}\mathrm{Tr} \left[ \bm{\Sigma}\right] = P - B_x^{-1} \mathbb{E}_{\bY, \bV} \left[ \|\bm{\mu}(\bY, \bV)\|_2^2 \right] \leq P.
    \end{equation*}
    In the boundary case where $T = P$, we find that $\mathbb{E}_{\bY, \bV}[ \|\bm{\mu}(\bY, \bV)\|_2^2 ] = 0$. Consequently, the optimization objective in \eqref{eq:cssf_GMI_max} becomes monotonically decreasing with respect to $|\alpha|^2$ on the interval $(0, +\infty)$. The maximum is therefore attained at $|\alpha| = 0$. Substituting this result into~\eqref{eq:cssf_GMI_max} yields $I_{\mathrm{GMI,cssf}} = 0$, which implies that no positive rate is certified by the optimized GMI within this restricted metric class.  Therefore, in this case we set $\bm{g}(\by, \bv) = \boldsymbol{0}_{B_x}$ and $\alpha = 0$.

    If $0 < T < P$, substituting the optimal value of $|\alpha|$ into the expressions for the optimal processing function in~\eqref{eq:argument} and~\eqref{eq:opt_normal_g} yields 
    \begin{equation*}
        \bg(\by, \bv) = \sqrt{\frac{P}{T(P-T)}} \bm{\mu}(\by, \bv), \qquad \alpha = \sqrt{\frac{P-T}{PT}}
    \end{equation*}
    Finally, substituting the optimal value of $|\alpha|^2$ into~\eqref{eq:cssf_GMI_max} gives
    \begin{equation*}
        I_{\mathrm{GMI,cssf}} = \log \frac{P}{T},
    \end{equation*}
    thereby completing the proof. 
\end{proof}
\subsection{Proof for Proposition~\ref{prop:Vec-GNNDR-cmsf}
(Constant matrix scaling function)}

\begin{proof}
Under \eqref{eq:Vec-GNNDR-cmsf}, the scaling matrix is independent
of both $\by$ and $\bv$. After removing an irrelevant left unitary
factor, write
\begin{equation*}
    \bm{\Pi}
    =
    \operatorname{diag}(\sigma_1,\ldots,\sigma_{B_x})\bm{W}^H,
    \qquad
    \bm{W}=[\bw_1,\ldots,\bw_{B_x}]\in\mathcal{U}(B_x).
\end{equation*}
Initially take $\sigma_i>0$; zero singular values will be treated below.
For fixed $\bm{\Pi}$, optimizing the processing function as in the
proof of Theorem~\ref{thm:main_result} gives
\begin{equation*}
    g_i(\by,\bv)
    =
    \frac{1+P\sigma_i^2}{P\sigma_i}
    \bw_i^H\bm{\mu}(\by,\bv).
\end{equation*}
Define
\begin{equation*}
    q_i
    =
    \bw_i^H\bm{\Sigma}\bw_i
    =
    P-
    \mathbb{E}_{\bY,\bV}
    \left[
        |\bw_i^H\bm{\mu}(\bY,\bV)|^2
    \right].
\end{equation*}
Since $0\prec\bm{\Sigma}\preceq P\bI_{B_x}$, we have
$0<q_i\leq P$.
Substituting the optimal processing function into
\eqref{eq:max_2} gives the objective
\begin{equation}
    \frac{1}{B_x}
    \sum_{i=1}^{B_x}
    \left\{
        \log(1+P\sigma_i^2)
        -\sigma_i^2q_i
        +\frac{P-q_i}{P}
    \right\}.
    \label{eq:cmsf_GMI_max}
\end{equation}

For fixed $\bm{W}$, each summand is strictly concave in
$\sigma_i^2$. Its derivative is
\begin{equation*}
    \frac{P}{1+P\sigma_i^2}-q_i,
\end{equation*}
so the optimal singular value satisfies
\begin{equation}
    \sigma_i^2
    =
    \frac{1}{q_i}-\frac{1}{P}
    =
    \frac{1}{
        \mathbb{E}
        [\operatorname{Var}(\bw_i^H\bX\mid\bY,\bV)]
    }
    -\frac{1}{P}.
    \label{eq:cmsf_sigma}
\end{equation}
If $q_i<P$, this value is strictly positive and the resulting
contribution is $\log(P/q_i)$.
If $q_i=P$, then
\begin{equation*}
    \mathbb{E}_{\bY,\bV}
    \left[
        |\bw_i^H\bm{\mu}(\bY,\bV)|^2
    \right]=0,
\end{equation*}
and hence $\bw_i^H\bm{\mu}(\bY,\bV)=0$ almost surely.
Taking $\sigma_i=0$ and $g_i=0$ attains the zero contribution
$\log(P/q_i)=0$.

Consequently,
\begin{equation*}
    I_{\mathrm{GMI,cmsf}}
    =
    \max_{\bm{W}\in\mathcal{U}(B_x)}
    \frac{1}{B_x}
    \sum_{i=1}^{B_x}
    \log\frac{P}{\bw_i^H\bm{\Sigma}\bw_i}.
\end{equation*}
The convexity of $q\mapsto\log(P/q)$ and
Lemma~\ref{lem:spec_maj_inequal} imply
\begin{equation*}
    I_{\mathrm{GMI,cmsf}}
    =
    \frac{1}{B_x}\sum_{i=1}^{B_x}\log\frac{P}{\lambda_i}
    =
    \frac{1}{B_x}\log\frac{P^{B_x}}{\det(\bm{\Sigma})},
\end{equation*}
with equality when $\bm{W}=\tilde{\bm{W}}$ diagonalizes
$\bm{\Sigma}$.

For each $i\in I^*$, \eqref{eq:cmsf_sigma} gives
\begin{equation*}
    \sigma_i
    =
    \sqrt{\frac{P-\lambda_i}{P\lambda_i}},
    \qquad
    \frac{1+P\sigma_i^2}{P\sigma_i}
    =
    \sqrt{\frac{P}{(P-\lambda_i)\lambda_i}}.
\end{equation*}
Therefore,
\begin{equation*}
    \begin{aligned}
        \bg(\by,\bv)
        =
        \operatorname{diag}
        \left(
            \sqrt{\frac{P}{(P-\lambda_i)\lambda_i}}
        \right)_{i\in I^*}
        \tilde{\bm{W}}_*^H\bm{\mu}(\by,\bv), \qquad 
        \bm{\Pi}
        =
        \operatorname{diag}
        \left(
            \sqrt{\frac{P-\lambda_i}{P\lambda_i}}
        \right)_{i\in I^*}
        \tilde{\bm{W}}_*^H.
    \end{aligned}
\end{equation*}
Inactive coordinates may be restored by zero padding.
This completes the proof.
\end{proof}
\subsection{Proof for Proposition~\ref{prop:Vec-GNNDR-csi-ssf}
(CSI-dependent scalar scaling function)}
\label{subsec:Vec-GNNDR-csi-ssf}

\begin{proof}
Under the decoding rule in \eqref{eq:dec_csi-ssf}, we have
$\bm{f}(\by,\bv)=\alpha(\bv)\bI_{B_x}$, which implies that
$\sigma_i(\by,\bv)=|\alpha(\bv)|$ and
$\bW(\by,\bv)=e^{-i\operatorname{Arg}(\alpha(\bv))}\bI_{B_x}$
whenever $\alpha(\bv)\neq0$.
The optimization of $\bg(\by,\bv)$ proceeds exactly as in the
proof of Theorem~\ref{thm:main_result}; the only modification is that
the optimization problem in \eqref{eq:max_2} reduces to
\begin{equation}
    \begin{aligned}
        I_{\mathrm{GMI,csi-ssf}}
        =
        \mathbb{E}_{\bV}\bigg[
        \max_{|\alpha(\bV)|}\bigg\{
        \log(1+P|\alpha(\bV)|^2) 
        +
        \frac{1+P|\alpha(\bV)|^2}{B_xP}
        \mathbb{E}_{\bY\mid\bV}
        \left[\|\bm{\mu}(\bY,\bV)\|_2^2\right]
        -
        \frac{|\alpha(\bV)|^2}{B_x}
        \mathbb{E}
        \left[\|\bX\|_2^2\mid\bV\right]
        \bigg\}\bigg],
    \end{aligned}
    \label{eq:max_csi-ssf}
\end{equation}
where zero scaling coefficients are handled by the boundary case below.

By an argument analogous to that in the proof of
Proposition~\ref{prop:Vec-GNNDR-cssf}, the optimal scaling factor is
given by
\begin{equation}
    |\alpha(\bv)|^2
    =
    \frac{P-T(\bv)}{PT(\bv)},
    \label{eq:alpha_v}
\end{equation}
where $T(\bv)=B_x^{-1}\operatorname{Tr}[\bm{\Sigma}(\bv)]$.
Moreover, we have
\begin{equation*}
    \begin{aligned}
        T(\bv)
        =
        B_x^{-1}\operatorname{Tr}[\bm{\Sigma}(\bv)] \leq
        B_x^{-1}\operatorname{Tr}
        \left(\mathbb{E}[\bX\bX^H\mid\bv]\right)
        =
        B_x^{-1}\mathbb{E}[\|\bX\|_2^2\mid\bv].
    \end{aligned}
\end{equation*}
Since the channel assumption implies that
$\bV\leftrightarrow\bS\leftrightarrow(\bX,\bY)$ forms a Markov chain
and that $\bX$ is independent of $\bS$, it follows that
$\bX$ is independent of $\bV$. Consequently,
\begin{equation*}
    T(\bv)
    \leq
    B_x^{-1}\mathbb{E}[\|\bX\|_2^2\mid\bv]
    =
    B_x^{-1}\mathbb{E}[\|\bX\|_2^2]
    =
    P.
\end{equation*}
Furthermore, Assumption~\ref{ass:psd_cond_cov} ensures that
$T(\bv)>0$ almost surely.

In the boundary case $T(\bv)=P$, we obtain
$\mathbb{E}_{\bY\mid\bv}[\|\bm{\mu}(\bY,\bv)\|_2^2]=0$.
In this case, with the same analysis as in the proof of
Proposition~\ref{prop:Vec-GNNDR-cssf}, we may directly set
$\alpha(\bv)=0$ and $\bg(\by,\bv)=\boldsymbol{0}_{B_x}$
without incurring any GMI loss.

When $0<T(\bv)<P$, the optimal processing function $\bg$ and
CSI-dependent scalar scaling function $\alpha$ are given by
\begin{equation*}
    \bg(\by,\bv)
    =
    \sqrt{\frac{P}{T(\bv)(P-T(\bv))}}\,
    \bm{\mu}(\by,\bv),
    \qquad
    \alpha(\bv)
    =
    \sqrt{\frac{P-T(\bv)}{PT(\bv)}},
\end{equation*}
by the same analysis as in the proof of
Proposition~\ref{prop:Vec-GNNDR-cssf}.

Substituting \eqref{eq:alpha_v} into \eqref{eq:max_csi-ssf},
and following the same calculation as in the proof of
Proposition~\ref{prop:Vec-GNNDR-cssf}, we conclude that the optimal
GMI under the decoding rule \eqref{eq:dec_csi-ssf} is
\begin{equation*}
    I_{\mathrm{GMI,csi-ssf}}
    =
    \mathbb{E}_{\bV}
    \left[\log\frac{P}{T(\bV)}\right],
\end{equation*}
which completes the proof.
\end{proof}

\subsection{Proof for Proposition~\ref{prop:Vec-GNNDR-csi-msf}
(CSI-dependent matrix scaling function)}
\label{app:Vec-GNNDR-csi-msf}

\begin{proof}
Under the decoding rule \eqref{eq:dec-csi-msf}, the scaling matrix
$\bm{\Pi}(\bv)$ depends on the CSI $\bv$ but is independent of the
channel output $\by$. Consequently, we may write
$\sigma_i(\by,\bv)=\sigma_i(\bv)$ and
$\bw_i(\by,\bv)=\bw_i(\bv)$.
The optimization of $\bg(\by,\bv)$ proceeds exactly as in the proof
of Theorem~\ref{thm:main_result}.
In this case, for each fixed $\bv$, the maximization problem in
\eqref{eq:max_2} simplifies to
\begin{equation}
    \begin{aligned}
        \max_{\{\sigma_i,\bw_i\}_{i=1}^{B_x}}
        \frac{1}{B_x}\sum_{i=1}^{B_x}
        \bigg\{
        \log(1+P\sigma_i^2(\bv))
        +
        \frac{1+P\sigma_i^2(\bv)}{P}
        \mathbb{E}_{\bY\mid\bv}
        \left[
            |\bw_i^H(\bv)\bm{\mu}(\bY,\bv)|^2
        \right]
        -
        \sigma_i^2(\bv)\bw_i^H(\bv)
        \mathbb{E}[\bX\bX^H\mid\bv]\bw_i(\bv)
        \bigg\},
    \end{aligned}
    \label{eq:max_csi-msf}
\end{equation}
where $\{\bw_i(\bv)\}_{i=1}^{B_x}$ forms an orthonormal basis.

By an argument analogous to that in the proof of
Proposition~\ref{prop:Vec-GNNDR-cmsf}, the optimal $\sigma_i(\bv)$
is given by
\begin{equation}
    \sigma_i^2(\bv)
    =
    \frac{1}{
        \bw_i^H(\bv)\bm{\Sigma}(\bv)\bw_i(\bv)
    }
    -\frac{1}{P}.
    \label{eq:opt_sigma_v}
\end{equation}
Applying the Rayleigh quotient (Lemma~\ref{lem:ray_quoti}) and the
maximum eigenvalue inequality (Lemma~\ref{lem:max_eigen_inequal})
yields
\begin{equation*}
    \begin{aligned}
        0
        <
        \bw_i^H(\bv)\bm{\Sigma}(\bv)\bw_i(\bv)
        \leq
        \lambda_{\max}(\bm{\Sigma}(\bv))
        \leq
        \lambda_{\max}
        \left(\mathbb{E}[\bX\bX^H\mid\bv]\right)
        =
        P,
    \end{aligned}
\end{equation*}
where the strict positivity follows from
Assumption~\ref{ass:psd_cond_cov}, and the last equality follows from
$\mathbb{E}[\bX\bX^H\mid\bv]=P\bI_{B_x}$.

In the boundary case
$\bw_i^H(\bv)\bm{\Sigma}(\bv)\bw_i(\bv)=P$, we must have
\begin{equation*}
    \mathbb{E}_{\bY\mid\bv}
    \left[
        |\bw_i^H(\bv)\bm{\mu}(\bY,\bv)|^2
    \right]=0.
\end{equation*}
With the same analysis as in the proof of
Proposition~\ref{prop:Vec-GNNDR-cmsf}, in this situation we may set
$\sigma_i(\bv)=0$ and $g_i(\by,\bv)=0$ without incurring any GMI
loss in the $i$-th direction.

Substituting \eqref{eq:opt_sigma_v} into \eqref{eq:max_csi-msf},
the optimization reduces to
\begin{equation*}
    I_{\mathrm{GMI,csi-msf}}
    =
    \mathbb{E}_{\bV}
    \left[
        \max_{\{\bw_i\}_{i=1}^{B_x}}
        \frac{1}{B_x}\sum_{i=1}^{B_x}
        \log\frac{P}{
            \bw_i^H(\bV)\bm{\Sigma}(\bV)\bw_i(\bV)
        }
    \right].
\end{equation*}
By the same reasoning as in the proof of
Proposition~\ref{prop:Vec-GNNDR-cmsf}, applying
Lemma~\ref{lem:spec_maj_inequal} gives the optimal GMI under
\eqref{eq:dec-csi-msf} as
\begin{equation*}
    I_{\mathrm{GMI,csi-msf}}
    =
    \frac{1}{B_x}
    \mathbb{E}_{\bV}
    \left[
        \log\frac{P^{B_x}}{\det(\bm{\Sigma}(\bV))}
    \right],
\end{equation*}
which is attained by choosing
$\bW(\bv)=\tilde{\bm{W}}(\bv)$ such that
\begin{equation}
    \tilde{\bm{W}}^H(\bv)\bm{\Sigma}(\bv)\tilde{\bm{W}}(\bv)
    =
    \operatorname{diag}
    \left(\lambda_1(\bv),\ldots,\lambda_{B_x}(\bv)\right).
    \label{eq:opt_w_csi-msf}
\end{equation}
Here the eigenvalues $\{\lambda_i(\bv)\}_{i=1}^{B_x}$ are arranged
in non-increasing order.

Finally, combining \eqref{eq:opt_sigma_v} and
\eqref{eq:opt_w_csi-msf}, and discarding the directions with zero
GMI contribution, i.e., those with $\lambda_i(\bv)=P$, the optimal
processing function and CSI-dependent scaling matrix are
\begin{equation*}
    \bg(\by,\bv)
    =
    \operatorname{diag}
    \left(
        \sqrt{\frac{P}{
            \lambda_i(\bv)(P-\lambda_i(\bv))
        }}
    \right)_{i\in I^*(\bv)}
    \tilde{\bm{W}}_*^H(\bv)\bm{\mu}(\by,\bv),
\end{equation*}
and
\begin{equation*}
    \bm{\Pi}(\bv)
    =
    \operatorname{diag}
    \left(
        \sqrt{\frac{P-\lambda_i(\bv)}{P\lambda_i(\bv)}}
    \right)_{i\in I^*(\bv)}
    \tilde{\bm{W}}_*^H(\bv),
\end{equation*}
where $I^*(\bv)$ and $\tilde{\bm{W}}_*(\bv)$ are defined in
Proposition~\ref{prop:Vec-GNNDR-csi-msf}.
Inactive coordinates may be restored by zero padding.
This completes the proof.
\end{proof}
\subsection{Proof for Proposition~\ref{prop:Vec-GNNDR-lin}
(Linear processing function)}
\label{subsec:Vec-GNNDR-lin}

\begin{proof}
Following the proof of Proposition~\ref{prop:Vec-GNNDR-csi-msf},
we simplify the notation by writing
$\sigma_i(\by,\bv)=\sigma_i(\bv)$ and
$\bw_i(\by,\bv)=\bw_i(\bv)$, while letting
$\bg(\by,\bv)=\bm{\Gamma}^H(\bv)\by$.
Under this notation, the optimization problem in
\eqref{eq:max_GMI_gf} becomes
\begin{equation}
    \begin{aligned}
        I_{\mathrm{GMI,lin}}
        =
        \frac{1}{B_x}\mathbb{E}_{\bV}\bigg[
        &\max_{\{\bm{\gamma}_i,\sigma_i,\bw_i\}_{i=1}^{B_x}}
        \sum_{i=1}^{B_x}\bigg\{
        -
        \bm{\gamma}_i^H(\bV)
        \mathbb{E}[\bY\bY^H\mid\bV]\bm{\gamma}_i(\bV)-
        \sigma_i^2(\bV)\bw_i^H(\bV)
        \mathbb{E}[\bX\bX^H\mid\bV]\bw_i(\bV)\\
        &+
        2\sigma_i(\bV)\Re\left(
            \bw_i^H(\bV)
            \mathbb{E}_{\bY\mid\bV}
            [\bm{\mu}(\bY,\bV)\bY^H]
            \bm{\gamma}_i(\bV)
        \right)+
        \frac{
            \bm{\gamma}_i^H(\bV)
            \mathbb{E}[\bY\bY^H\mid\bV]\bm{\gamma}_i(\bV)
        }{1+P\sigma_i^2(\bV)}
        +
        \log(1+P\sigma_i^2(\bV))
        \bigg\}\bigg],
    \end{aligned}
    \label{eq:GMI-max-lin}
\end{equation}
where the maximization is subject to
$\{\bw_i(\bv)\}_{i=1}^{B_x}$ forming an orthonormal basis.

We first optimize over $\bm{\gamma}_i(\bv)$.
For $\sigma_i(\bv)>0$, differentiating the quadratic objective
and setting its derivative to zero yields
\begin{equation*}
    -
    \frac{P\sigma_i^2(\bv)}{1+P\sigma_i^2(\bv)}
    \mathbb{E}_{\bY\mid\bv}[\bY\bY^H]\bm{\gamma}_i(\bv)
    +
    \sigma_i(\bv)
    \mathbb{E}_{\bY\mid\bv}
    \left[
        \bm{\mu}^H(\bY,\bv)\bw_i(\bv)\bY
    \right]
    =
    \boldsymbol{0}.
\end{equation*}
The corresponding solution is
\begin{equation}
    \bm{\gamma}_i(\bv)
    =
    \frac{1+P\sigma_i^2(\bv)}{P\sigma_i(\bv)}
    \left(
        \mathbb{E}_{\bY\mid\bv}[\bY\bY^H]
    \right)^\dagger
    \mathbb{E}_{\bY\mid\bv}
    \left[
        \bm{\mu}^H(\bY,\bv)\bw_i(\bv)\bY
    \right].
    \label{eq:opt_gamma}
\end{equation}
Substituting \eqref{eq:opt_gamma} into \eqref{eq:GMI-max-lin},
we obtain
\begin{equation}
    \begin{aligned}
        I_{\mathrm{GMI,lin}}
        =
        \frac{1}{B_x}\mathbb{E}_{\bV}\bigg[
        \max_{\{\sigma_i,\bw_i\}_{i=1}^{B_x}}
        \sum_{i=1}^{B_x}\bigg\{
        \frac{1+P\sigma_i^2(\bV)}{P}
        \bw_i^H(\bV)\bm{Q}(\bV)\bw_i(\bV)
        +
        \log(1+P\sigma_i^2(\bV))
        -
        P\sigma_i^2(\bV)
        \bigg\}\bigg],
    \end{aligned}
    \label{eq:max_GMI_lin_proof}
\end{equation}
where zero singular values are treated below.
Here we used
$\mathbb{E}[\bX\bX^H\mid\bv]=P\bI_{B_x}$ and
\begin{equation*}
    \begin{aligned}
        \mathbb{E}_{\bY\mid\bv}
        \left[\bm{\mu}(\bY,\bv)\bY^H\right]
        =
        \mathbb{E}_{\bY\mid\bv}
        \left[
            \mathbb{E}[\bX\mid\bY,\bv]\bY^H
        \right] =
        \mathbb{E}[\bX\bY^H\mid\bv].
    \end{aligned}
\end{equation*}
The matrix $\bm{Q}(\bv)$ is defined in \eqref{eq:lin-Q} as
\begin{equation*}
    \bm{Q}(\bv)
    =
    \mathbb{E}[\bX\bY^H\mid\bv]
    \left(
        \mathbb{E}[\bY\bY^H\mid\bv]
    \right)^\dagger
    \mathbb{E}[\bY\bX^H\mid\bv].
\end{equation*}

Next, optimizing over $\sigma_i^2(\bv)$, we differentiate with
respect to $\sigma_i^2(\bv)$, set the derivative equal to zero,
and obtain
\begin{equation}
    \frac{P}{1+P\sigma_i^2(\bv)}
    -
    P
    +
    \bw_i^H(\bv)\bm{Q}(\bv)\bw_i(\bv)
    =
    0.
    \label{eq:psigma_lin}
\end{equation}
Since $\bV\leftrightarrow\bS\leftrightarrow(\bX,\bY)$ forms a
Markov chain and $\bX$ is independent of $\bS$, it follows that
$\bX\perp\bV$.
Hence, applying the Rayleigh quotient
(Lemma~\ref{lem:ray_quoti}) together with
Lemma~\ref{lem:max_eig_cond_cov}, we conclude that
\begin{equation*}
    \begin{aligned}
        0
        \leq
        \bw_i^H(\bV)\bm{Q}(\bV)\bw_i(\bV)
        =
        \bw_i^H(\bV)
        \mathbb{E}[\bX\bY^H\mid\bV]
        \left(
            \mathbb{E}[\bY\bY^H\mid\bV]
        \right)^\dagger
        \mathbb{E}[\bY\bX^H\mid\bV]\bw_i(\bV)
        <
        P
    \end{aligned}
\end{equation*}
almost surely under Assumption~\ref{ass:psd_cond_cov}.
This ensures the existence of a valid solution to
\eqref{eq:psigma_lin} as
\begin{equation}
    \sigma_i^2(\bv)
    =
    \frac{1}{
        P-\bw_i^H(\bv)\bm{Q}(\bv)\bw_i(\bv)
    }
    -
    \frac{1}{P}.
    \label{eq:sigma_lin_opt}
\end{equation}

In the boundary case where
$\bw_i^H(\bv)\bm{Q}(\bv)\bw_i(\bv)=0$, we may set
$\sigma_i(\bv)=0$ and $\bm{\gamma}_i(\bv)=\boldsymbol{0}$
without loss of GMI, by the same reasoning as in
Proposition~\ref{prop:Vec-GNNDR-csi-msf}.
Substituting \eqref{eq:sigma_lin_opt} into
\eqref{eq:max_GMI_lin_proof}, the GMI maximization problem reduces to
\begin{equation*}
    I_{\mathrm{GMI,lin}}
    =
    \frac{1}{B_x}\mathbb{E}_{\bV}
    \left[
        \max_{\{\bw_i\}_{i=1}^{B_x}}
        \sum_{i=1}^{B_x}
        \log\frac{P}{
            P-\bw_i^H(\bV)\bm{Q}(\bV)\bw_i(\bV)
        }
    \right].
\end{equation*}
Following the same line of reasoning as in
Proposition~\ref{prop:Vec-GNNDR-csi-msf}, applying
Lemma~\ref{lem:spec_maj_inequal} to
$P\bI_{B_x}-\bm{Q}(\bv)\succ0$ with the convex function
$q\mapsto\log(P/q)$ gives
\begin{equation*}
    I_{\mathrm{GMI,lin}}
    =
    \frac{1}{B_x}\mathbb{E}_{\bV}
    \left[
        \log\frac{P^{B_x}}{
            \det(P\bI_{B_x}-\bm{Q}(\bV))
        }
    \right],
\end{equation*}
which is achieved by choosing
$\bW(\bv)=\tilde{\bm{W}}(\bv)$ satisfying
\begin{equation}
    \tilde{\bm{W}}^H(\bv)\bm{Q}(\bv)\tilde{\bm{W}}(\bv)
    =
    \operatorname{diag}
    \left(\lambda_1(\bv),\ldots,\lambda_{B_x}(\bv)\right).
    \label{eq:opt_W_lin}
\end{equation}
Here $\{\lambda_i(\bv)\}_{i=1}^{B_x}$ is arranged in
non-decreasing order.

Finally, combining \eqref{eq:opt_gamma},
\eqref{eq:sigma_lin_opt}, and \eqref{eq:opt_W_lin}, the optimal
linear processing function and CSI-dependent scaling matrix are
given by
\begin{equation*}
    \begin{aligned}
        \bm{\Gamma}(\bv)
        &=
        \left(
            \mathbb{E}[\bY\bY^H\mid\bv]
        \right)^\dagger
        \mathbb{E}[\bY\bX^H\mid\bv]\,
        \tilde{\bm{W}}_*(\bv)\\
        &\qquad{}\times
        \operatorname{diag}
        \left(
            \sqrt{\frac{P}{
                (P-\lambda_i(\bv))\lambda_i(\bv)
            }}
        \right)_{i\in I^*(\bv)},
    \end{aligned}
\end{equation*}
and
\begin{equation*}
    \bm{\Pi}(\bv)
    =
    \operatorname{diag}
    \left(
        \sqrt{\frac{\lambda_i(\bv)}{
            P(P-\lambda_i(\bv))
        }}
    \right)_{i\in I^*(\bv)}
    \tilde{\bm{W}}_*^H(\bv),
\end{equation*}
where $I^*(\bv)$ and $\tilde{\bm{W}}_*(\bv)$ are defined in
Proposition~\ref{prop:Vec-GNNDR-lin}.
Inactive coordinates may be restored by zero padding.
This completes the proof.
\end{proof}

\section{Block Noncoherent AWGN Channel} \label{app:non-coherent-AWGNC}
This appendix derives the Gaussian input mutual information $I_{\mathrm{Gauss}}$, the optimal processing and scaling functions
and the corresponding GMIs for the block noncoherent AWGN channel
in Section~\ref{subsec:non-coherent-AWGNC}.

\subsection{Derivation  of $I_{\mathrm{Gauss}}$}
\label{app:gaussian-mi-block-noncoherent}

Throughout this subsection, let
$\text{snr}=P/\sigma^2>0$ denote the signal-to-noise ratio in linear
scale. All information quantities are measured in nats, and
$I_{\mathrm{Gauss}}=B_x^{-1}I(\bX;\bY)$ is normalized per channel use.

\paragraph{Mutual-information expression}
For one block of the channel in \eqref{eq:channel_rotation}, write
$\bY=e^{i\Theta}\bX+\bN$, where
$\bX\sim\mathcal{CN}(\boldsymbol{0},P\bI_{B_x})$,
$\bN\sim\mathcal{CN}(\boldsymbol{0},\sigma^2\bI_{B_x})$, and
$\Theta\sim\operatorname{Unif}[0,2\pi)$ are mutually independent.
Integrating out the unknown phase gives
\begin{equation}
    \begin{aligned}
        p(\by\mid\bx)
        =
        \frac{1}{2\pi(\pi\sigma^2)^{B_x}}
        \int_0^{2\pi}
        \exp\left(
            -\frac{\|\by-e^{i\theta}\bx\|_2^2}{\sigma^2}
        \right)\,\mathrm{d}\theta =
        \frac{
            \exp\left(
                -\frac{\|\bx\|_2^2+\|\by\|_2^2}{\sigma^2}
            \right)
        }{(\pi\sigma^2)^{B_x}}
        I_0\left(
            \frac{2|\bx^H\by|}{\sigma^2}
        \right),
    \end{aligned}
    \label{eq:block-likelihood-gaussian-mi}
\end{equation}
where the second equality follows from
\begin{equation*}
    I_0(z)
    =
    \frac{1}{2\pi}
    \int_0^{2\pi}e^{z\cos t}\,\mathrm{d}t.
\end{equation*}
Here $I_0$ is the modified Bessel function of the first kind of
order zero.

By the circular symmetry of the Gaussian input,
\begin{equation*}
    \bY\mid\Theta
    \sim
    \mathcal{CN}
    \left(\boldsymbol{0},(P+\sigma^2)\bI_{B_x}\right),
\end{equation*}
independently of the value of $\Theta$. Thus $\bY$ is independent
of $\Theta$, and its marginal density is
\begin{equation*}
    p(\by)
    =
    \frac{1}{[\pi(P+\sigma^2)]^{B_x}}
    \exp\left(
        -\frac{\|\by\|_2^2}{P+\sigma^2}
    \right).
\end{equation*}
Substituting this density and
\eqref{eq:block-likelihood-gaussian-mi} into the definition of
mutual information, and using
$\mathbb{E}\|\bX\|_2^2=B_xP$ and
$\mathbb{E}\|\bY\|_2^2=B_x(P+\sigma^2)$, gives
\begin{equation*}
    \begin{aligned}
        I_{\mathrm{Gauss}}
        =
        \frac{1}{B_x}
        \mathbb{E}_{\bX,\bY}
        \left[
            \log\frac{p(\bY\mid\bX)}{p(\bY)}
        \right] =
        \log(1+\text{snr})
        -
        2\text{snr}
        +
        \frac{1}{B_x}
        \mathbb{E}_{\bX,\bY}
        \left[
            \log I_0\left(
                \frac{2|\bX^H\bY|}{\sigma^2}
            \right)
        \right],
    \end{aligned}
\end{equation*}
which establishes \eqref{eq:Gaussian_MI_block_nonherent}.

\paragraph{An equivalent phase-information representation}
Direct numerical evaluation of
\eqref{eq:Gaussian_MI_block_nonherent} is undesirable at high SNR:
$I_0$ grows exponentially, and the last two terms involve
cancellation between quantities of order $\text{snr}$.
To obtain a better-conditioned expression, use the independence
of $\bY$ and $\Theta$ and the chain rule for mutual information:
\begin{equation}
    \begin{aligned}
        I_{\mathrm{Gauss}}
        =
        \frac{1}{B_x}
        \left[
            I(\bX;\bY\mid\Theta)
            -
            I(\Theta;\bY\mid\bX)
        \right]=
        \log(1+\text{snr})
        -
        \frac{1}{B_x}I(\Theta;\bY\mid\bX).
    \end{aligned}
    \label{eq:block-mi-phase-penalty}
\end{equation}
For given $(\bx,\by)$, define
\begin{equation*}
    \kappa=\frac{2|\bx^H\by|}{\sigma^2},
    \qquad
    \varphi=\operatorname{Arg}(\bx^H\by).
\end{equation*}
Bayes' rule gives the von Mises posterior
\begin{equation*}
    p(\theta\mid\bx,\by)
    =
    \frac{\exp(\kappa\cos(\theta-\varphi))}
    {2\pi I_0(\kappa)}.
\end{equation*}
Its relative entropy with respect to the uniform phase prior is
\begin{equation}
    \begin{aligned}
        \mathcal{J}(\kappa)
        :=
        \int_0^{2\pi}
        p(\theta\mid\bx,\by)
        \log\left(2\pi p(\theta\mid\bx,\by)\right)
        \,\mathrm{d}\theta =
        \kappa\frac{I_1(\kappa)}{I_0(\kappa)}
        -
        \log I_0(\kappa),
    \end{aligned}
    \label{eq:block-phase-information-function}
\end{equation}
where $I_1$ is the modified Bessel function of the first kind
of order one, and $\mathcal{J}(0)=0$ by continuity.
Consequently,
\begin{equation*}
    I(\Theta;\bY\mid\bX)
    =
    \mathbb{E}_{\bX,\bY}
    \left[
        \mathcal{J}\left(
            \frac{2|\bX^H\bY|}{\sigma^2}
        \right)
    \right].
\end{equation*}

\paragraph{Reduction to one-dimensional integration}
Define
\begin{equation*}
    U=\frac{|\bX^H\bY|}{P},
    \qquad
    S=\frac{\|\bX\|_2^2}{P},
    \qquad
    \beta=\sqrt{1+\text{snr}^{-1}}.
\end{equation*}
Then $S$ has a gamma distribution with shape parameter $B_x$
and unit scale, with density
\begin{equation*}
    q_S(s)
    =
    \frac{s^{B_x-1}e^{-s}}{\Gamma(B_x)},
    \qquad s>0.
\end{equation*}
Here $\Gamma(\cdot)$ denotes the gamma function.
Conditionally on $S=s$, the phase-corrected inner product satisfies
\begin{equation*}
    \frac{e^{-i\Theta}\bX^H\bY}{P}\,\bigg|\,S=s
    \sim
    \mathcal{CN}\left(s,\frac{s}{\text{snr}}\right).
\end{equation*}
Integrating its complex Gaussian density over the polar angle
therefore yields
\begin{equation*}
    q_{U\mid S}(u\mid s)
    =
    \frac{2\text{snr}\,u}{s}
    \exp\left(
        -\text{snr}\left(s+\frac{u^2}{s}\right)
    \right)
    I_0(2\text{snr}\,u),
    \qquad u>0.
\end{equation*}
Averaging over $S$ gives
\begin{equation}
    \begin{aligned}
        q_U(u)
        &=
        \frac{2\text{snr}\,u\,I_0(2\text{snr}\,u)}
        {\Gamma(B_x)}
        \int_0^\infty
        s^{B_x-2}
        \exp\left(
            -(1+\text{snr})s
            -
            \frac{\text{snr}\,u^2}{s}
        \right)\,\mathrm{d}s\\
        &=
        \frac{4\text{snr}\,u^{B_x}}
        {\Gamma(B_x)\beta^{B_x-1}}
        I_0(2\text{snr}\,u)
        K_{B_x-1}(2\text{snr}\,\beta u),
        \qquad u>0.
    \end{aligned}
    \label{eq:block-inner-product-density}
\end{equation}
The second equality follows from the integral identity
\begin{equation*}
    \int_0^\infty
    s^{\nu-1}\exp\left(-as-\frac{b}{s}\right)\,\mathrm{d}s
    =
    2\left(\frac{b}{a}\right)^{\nu/2}
    K_\nu(2\sqrt{ab}),
    \qquad a,b>0,
\end{equation*}
with $\nu=B_x-1$, $a=1+\text{snr}$, and
$b=\text{snr}\,u^2$.
The function $K_\nu$ is the modified Bessel function of the second
kind of order $\nu$; in
\eqref{eq:block-inner-product-density}, its order is $B_x-1$.

Since $\kappa=2\text{snr}\,U$, combining
\eqref{eq:block-mi-phase-penalty} and
\eqref{eq:block-phase-information-function} yields the exact
one-dimensional representation
\begin{equation}
    I_{\mathrm{Gauss}}
    =
    \log(1+\text{snr})
    -
    \frac{1}{B_x}
    \int_0^\infty
    \mathcal{J}(2\text{snr}\,u)\,q_U(u)\,\mathrm{d}u.
    \label{eq:block-mi-one-dimensional}
\end{equation}

\subsection{Proof for Proposition~\ref{prop:rotation} (Optimal Vec-GNNDR for block noncoherent AWGN channel)}
\begin{proof}
    Given the rotation phase $\theta$ and observation $\by$, the log-posterior density of $\bx$ is given by 
    \begin{equation*}
        \log p (\bx \mid \theta, \by)  \propto  -\frac{1}{\sigma^2}  \|\by - e^{i \theta} \bx\|_2^2 - \frac{1}{P} \|\bx\|_2^2.
    \end{equation*}
    Thus $\bx \mid \theta, \by$ is circularly symmetric complex Gaussian with
    \begin{equation*}
        \mathbb{E}[\bx \mid \theta, \by] = \frac{P}{P+\sigma^2} e^{-i \theta} \by, \quad \operatorname{Cov}[\bx \mid \theta, \by] = \frac{P \sigma^2}{P + \sigma^2} \bI_{B_x}.
    \end{equation*}
    To compute $\mathbb{E}[\bx \mid \by]$, we integrate over $\theta$ with respect to the conditional distribution $p(\theta \mid \by)$. Since 
    \begin{equation*}
        p(\theta \mid \by) \propto p(\by \mid \theta) p(\theta) = p(\by \mid \theta) \frac{1}{2 \pi},
    \end{equation*}
    and the distribution of $\by \mid \theta$ is given by 
    \begin{equation*}
        \by \mid \theta \sim \mathcal{CN}(0, e^{i \theta} P \bI_{B_x} e^{-i \theta} + \sigma^2 \bI_{B_x}),
    \end{equation*}
    which is independent of $\theta$, we have $p(\theta \mid \by) =\frac{1}{2 \pi}$. Therefore, 
    \begin{equation} 
        \bm{\mu}(\by) = \mathbb{E}[\bx \mid \by] = \int_0^{2\pi} \frac{P}{P+\sigma^2} e^{-i \theta} \by \frac{1}{ 2 \pi} \mathrm{d} \theta = \bm{0}_{B_x},
        \label{eq:mu_block_rotation}
    \end{equation}
    so  the optimal processing function is $\bg(\by) = \bm{0}_{B_x}$.
    For the conditional covariance $\Sigma(\by) = \operatorname{Cov}[\bx \mid \by]$, we apply the law of total covariance  as 
    \begin{equation} \label{eq:decom_cond_dov}
        \bm{\Sigma}(\by) = \mathbb{E}_{\theta \mid \by} [\operatorname{Cov}(\bx \mid \theta, \by) ] + \operatorname{Cov}_{\theta \mid \by} (\mathbb{E}[\bx \mid \theta, \by] ).
    \end{equation}
   The first term equals $\frac{P \sigma^2}{P + \sigma^2} \bI_{B_x}$. For the second term, 
    \begin{equation*}
        \operatorname{Cov}_{\theta \mid \by} (\mathbb{E}[\bx \mid \theta, \by]) = \mathbb{E}_{\theta \mid \by} \left[ \left( \frac{P}{P+\sigma^2} \right)^2 e^{i \theta } \by e^{-i \theta} \by^{H}\right] - \mathbb{E}_{\theta \mid \by}\left[  \frac{P}{P+\sigma^2} e^{i \theta} \by\right] \mathbb{E}_{\theta \mid \by} \left[  \frac{P}{P+\sigma^2} e^{-i \theta} \by^H\right] = \left( \frac{P}{P+\sigma^2}\right)^2 \by \by^H,
    \end{equation*}
    since $\mathbb{E}_{\theta \mid \by } [e^{i \theta}] = 0$. 
    Hence by \eqref{eq:decom_cond_dov}, we obtain 
    \begin{equation}
        \bm{\Sigma}(\by) = \frac{P \sigma^2}{P + \sigma^2} \bI_{B_x} + \left( \frac{P}{P+\sigma^2}\right)^2 \by \by^H.
        \label{eq:sigma_block_rotation}
    \end{equation}
    The eigenvalues of $\bm{\Sigma}(\by)$ are therefore
    \begin{equation}
        \lambda_1(\by) = \frac{P \sigma^2}{P + \sigma^2} +  \left( \frac{P}{P+\sigma^2}\right)^2 \by^H\by , \quad  \lambda_2(\by) = \cdots = \lambda_{B_x }(\by) = \frac{P \sigma^2}{P + \sigma^2}.
        \label{eq:lambda_block_rotation}
    \end{equation}
    Clearly, $\lambda_i(\by) < P$ for $i=2,\ldots,B_x$, and $\lambda_1(\by)<P$ if $\by^H\by < P+\sigma^2$, whereas $\lambda_1(\by)\ge P$ if $\by^H\by \ge P+\sigma^2$.

    Substituting~\eqref{eq:mu_block_rotation}and~\eqref{eq:sigma_block_rotation} into~\eqref{eq:IGMI} yields
    \begin{equation*}
        \begin{aligned}
            I_{\mathrm{GMI,opt}} &= \frac{B_x - 1}{B_x} \left[ \log \left( 1+\frac{P}{\sigma^2}\right) - \frac{P}{P+\sigma^2}\right] + \frac{1}{B_x} \mathbb{E}_{\bY} \left[ \left(\log \left( \frac{P+\sigma^2}{\sigma^2 + \frac{P}{P+\sigma^2}\bY^H \bY} \right) + \frac{P}{P+\sigma^2} \left( \frac{\bY^H \bY}{P+\sigma^2} -1 \right) \right) \mathbb{I} \{\bY^H \bY < P+\sigma^2\}\right] \\
            & =  \frac{B_x - 1}{ B_x} \left[\log \left( 1 +\frac{P}{\sigma^2}\right) - \frac{P}{P+\sigma^2}\right] + \frac{1}{B_x} \mathbb{E}_{T \sim \chi^2(2B_x)} \left[ \left( \log \frac{P+\sigma^2}{\sigma^2+PT/2} + \frac{P(T-2)}{2(P+\sigma^2)} \right) \mathbb{I}\{T <2\}\right],
        \end{aligned}
    \end{equation*}
    where we use the fact $\frac{1}{\sqrt{P+\sigma^2}} \bY \sim \mathcal{CN}(0,\boldsymbol{I}_{B_x})$.

     We further notice that $\bm{\mu}(\by) = \boldsymbol{0}_{B_x}$ for all observations $\by$. Hence, choosing $\sigma_1(\by)$ sufficiently small does not yield any additional contribution in~\eqref{eq:max_2} compared with directly setting $\sigma_1(\by) = 0$. Consequently, there is no need to truncate $\lambda_1(\by)$ when $\by^H \by \geq P + \sigma^2$; instead, it suffices to set $\sigma_1(\by) = 0$.

     The principal eigenvector corresponding to $\lambda_1(\by)$ is given by $\bw_1(\by) = \by / \|\by\|_2$. Define
    \begin{equation*}
        \bW(\by) = \big( \bw_1(\by), \bW_{\perp}(\by) \big),
    \end{equation*}
    where $\bW_{\perp}(\by)$ denotes any orthonormal basis of the subspace orthogonal to $\bw_1(\by)$.

    Substituting $\{[\lambda_i(\by)]_\varepsilon\}_{i=1}^{B_x}$ and $\bW(\by)$ into~\eqref{eq:opt_g_eps} and~\eqref{eq:opt_f_eps} yields the following optimal processing and scaling functions:
    \begin{itemize}
        \item If $\by^H \by < P + \sigma^2$, then
        \begin{equation*}
            \bg(\by) = \boldsymbol{0}_{B_x}, \qquad 
            \bm{f}(\by) = \mathrm{diag} \!\left( 
            \sqrt{\tfrac{P-\lambda_1(\by)}{P \lambda_1(\by)}}, 
            \tfrac{1}{\sigma}, \ldots, \tfrac{1}{\sigma} \right) 
            \bm{W}^H(\by).
        \end{equation*}
        \item If $\by^H \by \geq P + \sigma^2$, then
        \begin{equation*}
            \bg(\by) = \boldsymbol{0}_{B_x-1}, \qquad 
            \bm{f}(\by) = \mathrm{diag} \!\left( 
            \tfrac{1}{\sigma}, \ldots, \tfrac{1}{\sigma} \right) 
            \bm{W}_{\perp}^H(\by).
        \end{equation*}
    \end{itemize}
    
    This completes the proof.
\end{proof}

\subsection{Proof for Proposition~\ref{prop:rotation_ele} (block noncoherent AWGN channel with element-wise GNNDR)}
\begin{proof}
    Since the same processing function $g$ and scaling function $f$ are applied to each element within a block, the original vector-valued processing and scaling functions in \eqref{def:g} and \eqref{def:f} reduce to
    \begin{equation*}
        \bg(\by) = [g(y_1), \cdots, g(y_{B_x})], \quad  \bm{f}(\by) = \operatorname{diag} \{f(y_1), \cdots, f(y_{B_x})\},
    \end{equation*}
    where $\by = [y_1, \cdots, y_{B_x}]$. In this case, the right singular matrix becomes $\bW(\by) = \boldsymbol{I}_{B_x}$, and the singular values are $\sigma_i(\by) = g(y_i)$. Since each pair $(x_i, y_i)$ shares the same marginal distribution as described in \eqref{eq:distri_br_ele}, the optimization problem in \eqref{eq:max_GMI_opt} becomes
    \begin{equation} \label{eq:sim_opt}
        \max_{\theta<0} \left(\theta \mathbb{E}_{(X,Y)} [ \left| g(Y) - f(Y) x \right|^2] - \mathbb{E}_{(X,Y)} \left[ \frac{\theta |g(Y)|^2 }{1 - \theta P f^2(Y)} \right] + \mathbb{E}_{(X,Y)} \left[ \log \left(1- \theta P |f(Y)|^2 \right)\right]\right),
    \end{equation}
    which can be interpreted as an optimization problem for a memoryless channel characterized by the marginal distribution in \eqref{eq:distri_br_ele}.

    Following the same reasoning as in the proof of Proposition~\ref{prop:rotation}, the conditional mean and conditional variance are given by
    \begin{equation*}
        \mu(y) = 0, 
        \qquad 
        w(y) = \frac{P\sigma^2}{P+\sigma^2} 
        + \left(\frac{P}{P+\sigma^2}\right)^2 |y|^2.
    \end{equation*}
    Clearly, $w(y) < P$ whenever $|y|^2 < P+\sigma^2$, and $w(y) \geq P$ otherwise.  Hence, by invoking Corollary~\ref{coral:memoryless}, the optimal GMI under the element-wise GNNDR rule is given by
    \begin{equation*}
        \begin{aligned}
            I_{\mathrm{GMI,ele}} 
            &= \mathbb{E}_{Y} \!\left[ 
            \left( \log \!\left( \frac{P+\sigma^2}{\sigma^2 + \tfrac{P}{P+\sigma^2} Y^H Y}\right) 
            + \frac{P}{P+\sigma^2} \!\left( \frac{Y^HY}{P+\sigma^2} - 1 \right) 
            \right) 
            \mathbb{I}\{Y^H Y < P + \sigma^2\} 
            \right] \\
            &= \mathbb{E}_{T \sim \chi^2(2)} \!\left[ 
            \left( \log \frac{P+\sigma^2}{\sigma^2 + PT/2} 
            + \frac{P(T-2)}{2(P+\sigma^2)} 
            \right) 
            \mathbb{I}\{T < 2\} 
        \right],
        \end{aligned}
    \end{equation*}
    where we have used the fact that $\tfrac{1}{\sqrt{P+\sigma^2}}Y \sim \mathcal{CN}(0,1)$.

    Analogously to the analysis in the proof for  Proposition~\ref{prop:rotation}, truncation of $w(y)$ is not required; it suffices to set $f(y) = 0$ whenever $|y|^2 \geq P+\sigma^2$. Therefore, the optimal processing and scaling functions are
    \begin{equation*}
        g(y) = 0, 
        \qquad 
        f(y) = \sqrt{\frac{P-w(y)}{Pw(y)}}, 
        \qquad \text{if } |y|^2 < P+\sigma^2,
    \end{equation*}
    and $g(y) = f(y) = 0$ otherwise. This completes the proof.
    
\end{proof}
\begin{remark}
    The element-wise processing function $g$ and scaling function $f$ derived in the above proof coincide with those obtained by directly applying GNNDR as in \cite{wang2022generalized}. However, this equivalence holds only under the condition that all elements within the same block share the same marginal distribution. When the marginal distributions differ across elements, the optimization problem in \eqref{eq:max_GMI_opt} can no longer be reduced to the simplified form in \eqref{eq:sim_opt}. 
\end{remark}
\subsection{Calculation for limit in \eqref{eq:block_limit}}

\begin{proof} 
When $0<T<2$, we have 
\begin{equation*}
    0\leq-\log \left( 1 + \frac{P(T-2)}{2(P+\sigma^2)}\right) + \frac{P(T-2)}{2(P+\sigma^2)}=\log \left(  \frac{P + \sigma^2}{\sigma^2 + PT/2}\right) + \frac{P(T-2)} {2(P+\sigma^2)} \leq \log \left(1 + \frac{P}{\sigma^2} \right),
\end{equation*}
where we use the inequality $\log (1+x) \leq x$.
Then we can bound $I_{\mathrm{GMI,opt}}$ in \eqref{eq:GMI-block} as 
\begin{equation*}
    \frac{B_x -1}{B_x} \left[ \log \left( 1+ \frac{P}{\sigma^2}\right) - \frac{P}{P+\sigma^2}\right]\leq I_{\mathrm{GMI,opt}} \leq   \log \left( 1 + \frac{P}{\sigma^2}\right) - \frac{B_x -1}{B_x} \frac{P}{P+\sigma^2}
\end{equation*}
then taking the limit $B_x \to \infty$ and by the squeeze theorem, we have 
\begin{equation*}
    \lim_{B_x \to \infty} I_{\mathrm{GMI,opt}} = \log \left( 1 + \frac{P}{\sigma^2}\right) - \frac{P}{P+\sigma^2},
\end{equation*}
which completes the proof.
\end{proof}

\section{Phase Noise Channel}
\label{app:pnc}

In this section, we present the numerical evaluation of
$I_{\mathrm{GMI,opt}}$ and the proofs of the analytical results
for the phase noise channel considered in Section~\ref{subsec:pnc}.
\subsection{Derivation and recursive evaluation of $I_{\mathrm{Gauss}}$}
\label{app:gaussian-mi-phase-noise}

Consider a single block of the phase noise channel in
\eqref{eq:channel_phase}, written as
\begin{equation*}
    \bY=\bm{\Theta}(\boldsymbol{\phi})\bX+\bN,
    \qquad
    \bm{\Theta}(\boldsymbol{\varphi})
    =
    \operatorname{diag}
    \left(e^{i\varphi_1},\ldots,e^{i\varphi_{B_x}}\right).
\end{equation*}
Here
$\bX\sim\mathcal{CN}(\boldsymbol{0},P\bI_{B_x})$,
$\bN\sim\mathcal{CN}(\boldsymbol{0},\sigma^2\bI_{B_x})$,
and
$\boldsymbol{\phi}\sim
\mathcal{N}(\boldsymbol{0},\bm{\Sigma}_{\phi})$
are mutually independent, with
$(\bm{\Sigma}_{\phi})_{ij}=c^2\min\{i,j\}$ and $c>0$.
Let $\text{snr}=P/\sigma^2$.
All logarithms are natural, and
$I_{\mathrm{Gauss}}=B_x^{-1}I(\bX;\bY)$ is normalized
per channel use.

\paragraph{Channel transition density}
Conditioned on the input and the phase vector, the output has
the Gaussian density
\begin{equation}
    p(\by\mid\bx,\boldsymbol{\varphi})
    =
    \frac{1}{(\pi\sigma^2)^{B_x}}
    \exp\left(
        -\frac{
            \|\by-\bm{\Theta}(\boldsymbol{\varphi})\bx\|_2^2
        }{\sigma^2}
    \right).
    \label{eq:pnc-mi-conditional-likelihood}
\end{equation}
Since $\boldsymbol{\phi}$ is independent of $\bX$, averaging
over the phase vector gives
\begin{equation}
    p(\by\mid\bx)
    =
    \frac{1}{(\pi\sigma^2)^{B_x}}
    \int_{\mathbb{R}^{B_x}}
    \exp\left(
        -\frac{
            \|\by-\bm{\Theta}(\boldsymbol{\varphi})\bx\|_2^2
        }{\sigma^2}
    \right)
    p_{\boldsymbol{\phi}}(\boldsymbol{\varphi})
    \,\mathrm{d}\boldsymbol{\varphi},
    \label{eq:pnc-mi-marginal-likelihood}
\end{equation}
where $p_{\boldsymbol{\phi}}$ is the density of
$\mathcal{N}(\boldsymbol{0},\bm{\Sigma}_{\phi})$.

\paragraph{Mutual-information identity}
By the circular symmetry of the white Gaussian input and the
unitarity of $\bm{\Theta}(\boldsymbol{\phi})$,
\begin{equation*}
    \bY\mid\boldsymbol{\phi}
    \sim
    \mathcal{CN}
    \left(\boldsymbol{0},(P+\sigma^2)\bI_{B_x}\right).
\end{equation*}
This distribution does not depend on $\boldsymbol{\phi}$, so
$I(\boldsymbol{\phi};\bY)=0$.
Moreover, conditioned on $\boldsymbol{\phi}$, the channel is
a Gaussian channel with a unitary channel matrix, and hence
\begin{equation*}
    I(\bX;\bY\mid\boldsymbol{\phi})
    =
    B_x\log(1+\text{snr}).
\end{equation*}
Applying the chain rule in two orders yields
\begin{equation*}
    \begin{aligned}
        I(\bX;\bY)
        +
        I(\boldsymbol{\phi};\bY\mid\bX)
        =
        I(\bX,\boldsymbol{\phi};\bY)
        =
        I(\boldsymbol{\phi};\bY)
        +
        I(\bX;\bY\mid\boldsymbol{\phi})
        =
        B_x\log(1+\text{snr}).
    \end{aligned}
\end{equation*}
Therefore,
\begin{equation}
    I_{\mathrm{Gauss}}
    =
    \log(1+\text{snr})
    -
    \frac{1}{B_x}
    I(\boldsymbol{\phi};\bY\mid\bX).
    \label{eq:pnc-mi-phase-penalty}
\end{equation}

The additive noise in the corresponding channel realization is
\begin{equation*}
    \bN
    =
    \bY-\bm{\Theta}(\boldsymbol{\phi})\bX.
\end{equation*}
Thus, evaluating \eqref{eq:pnc-mi-conditional-likelihood} at
$(\bX,\boldsymbol{\phi},\bY)$ gives
\begin{equation*}
    \log p(\bY\mid\bX,\boldsymbol{\phi})
    =
    -B_x\log(\pi\sigma^2)
    -
    \frac{\|\bN\|_2^2}{\sigma^2}.
\end{equation*}
Consequently,
\begin{equation*}
    \begin{aligned}
        I(\boldsymbol{\phi};\bY\mid\bX)
        =
        \mathbb{E}_{\bX,\boldsymbol{\phi},\bN}
        \left[
            \log
            \frac{
                p(\bY\mid\bX,\boldsymbol{\phi})
            }{
                p(\bY\mid\bX)
            }
        \right]=
        -
        \mathbb{E}_{\bX,\boldsymbol{\phi},\bN}
        \left[
            \frac{\|\bN\|_2^2}{\sigma^2}
            +
            \log\left(
                (\pi\sigma^2)^{B_x}p(\bY\mid\bX)
            \right)
        \right].
    \end{aligned}
\end{equation*}
Substituting this identity into \eqref{eq:pnc-mi-phase-penalty}
establishes \eqref{eq:Gaussian_MI_phase_noise}.
In particular,
$B_x^{-1}I(\boldsymbol{\phi};\bY\mid\bX)$ characterizes the
rate loss relative to the coherent Gaussian-input rate
$\log(1+\text{snr})$.

\paragraph{Recursive evaluation of the channel likelihood}
The Markov structure of the phase process permits a recursive
evaluation of \eqref{eq:pnc-mi-marginal-likelihood}.
Since the likelihood depends on each phase only through its
complex exponential, define the wrapped phases
\begin{equation*}
    \vartheta_i=\phi_i \bmod 2\pi\in[0,2\pi),
    \qquad i=1,\ldots,B_x,
\end{equation*}
with $\vartheta_0=0$.
The independent Gaussian phase increments induce the periodic
transition density
\begin{equation}
    q_c(t)
    =
    \frac{1}{\sqrt{2\pi}\,c}
    \sum_{k\in\mathbb{Z}}
    \exp\left(
        -\frac{(t+2\pi k)^2}{2c^2}
    \right),
    \qquad t\in\mathbb{R}.
    \label{eq:pnc-mi-wrapped-transition}
\end{equation}
In particular, $\vartheta_1$ has density $q_c(\theta)$, and the
conditional density of $\vartheta_i$ given
$\vartheta_{i-1}=u$ is $q_c(\theta-u)$.

For fixed $(\bx,\by)$, define the likelihood factor
\begin{equation}
    \ell_i(\theta)
    =
    \exp\left(
        -\frac{|y_i-e^{i\theta}x_i|^2}{\sigma^2}
    \right),
    \qquad \theta\in[0,2\pi).
    \label{eq:pnc-mi-emission}
\end{equation}
The phase-averaged channel density can then be written as
\begin{equation}
    p(\by\mid\bx)
    =
    \frac{1}{(\pi\sigma^2)^{B_x}}
    \int_{[0,2\pi)^{B_x}}
    \prod_{i=1}^{B_x}
    \left[
        q_c(\theta_i-\theta_{i-1})\ell_i(\theta_i)
    \right]
    \,\mathrm{d}\theta_1\cdots\mathrm{d}\theta_{B_x},
    \label{eq:pnc-mi-wrapped-likelihood}
\end{equation}
where $\theta_0=0$.

To evaluate this integral sequentially, let
$\pi_{i\mid i-1}(\theta)$ and $\pi_{i\mid i}(\theta)$ denote
the predictive and filtering densities of $\vartheta_i$,
conditioned on $\bX=\bx$ and on the first $i-1$ and $i$ output
symbols, respectively. Their dependence on the fixed input and
observations is suppressed for readability.
Initialize
\begin{equation}
    \pi_{1\mid0}(\theta)=q_c(\theta).
    \label{eq:pnc-mi-forward-initialization}
\end{equation}
For $i=1,\ldots,B_x$, compute the normalization constant
\begin{equation}
    \zeta_i(\bx,\by)
    =
    \int_0^{2\pi}
    \ell_i(\theta)\pi_{i\mid i-1}(\theta)
    \,\mathrm{d}\theta
    \label{eq:pnc-mi-forward-normalizer}
\end{equation}
and update the filtering density by Bayes' rule:
\begin{equation}
    \pi_{i\mid i}(\theta)
    =
    \frac{
        \ell_i(\theta)\pi_{i\mid i-1}(\theta)
    }{
        \zeta_i(\bx,\by)
    }.
    \label{eq:pnc-mi-forward-update}
\end{equation}
For $i<B_x$, the Markov transition then gives the next predictive
density:
\begin{equation}
    \pi_{i+1\mid i}(\theta)
    =
    \int_0^{2\pi}
    q_c(\theta-u)\pi_{i\mid i}(u)
    \,\mathrm{d}u.
    \label{eq:pnc-mi-forward-prediction}
\end{equation}

By \eqref{eq:pnc-mi-emission} and
\eqref{eq:pnc-mi-forward-normalizer}, the conditional density of
the $i$-th output symbol is
\begin{equation*}
    p(y_i\mid\bx,y_1,\ldots,y_{i-1})
    =
    \frac{\zeta_i(\bx,\by)}{\pi\sigma^2}.
\end{equation*}
The chain rule for conditional densities therefore yields
\begin{equation}
    p(\by\mid\bx)
    =
    \frac{1}{(\pi\sigma^2)^{B_x}}
    \prod_{i=1}^{B_x}\zeta_i(\bx,\by),
    \qquad
    \log\left(
        (\pi\sigma^2)^{B_x}p(\by\mid\bx)
    \right)
    =
    \sum_{i=1}^{B_x}\log\zeta_i(\bx,\by).
    \label{eq:pnc-mi-recursive-log-likelihood}
\end{equation}
Thus, the multidimensional phase integral is evaluated through
a sequence of one-dimensional prediction and update integrals,
with its log-likelihood obtained by accumulating the logarithms
of the normalization constants.

Finally, substituting
\eqref{eq:pnc-mi-recursive-log-likelihood} into
\eqref{eq:Gaussian_MI_phase_noise} gives
\begin{equation}
    I_{\mathrm{Gauss}}
    =
    \log(1+\text{snr})
    +
    \frac{1}{B_x}
    \mathbb{E}_{\bX,\boldsymbol{\phi},\bN}
    \left[
        \frac{\|\bN\|_2^2}{\sigma^2}
        +
        \sum_{i=1}^{B_x}\log\zeta_i(\bX,\bY)
    \right],
    \label{eq:pnc-mi-recursive-representation}
\end{equation}
where $\bY=\bm{\Theta}(\boldsymbol{\phi})\bX+\bN$ and the
normalization constants are obtained from
\eqref{eq:pnc-mi-forward-initialization}--\eqref{eq:pnc-mi-forward-prediction}.

\subsection{Numerical evaluation of the optimal GMI}
\label{app:pnc-opt-numerical}

For the white Gaussian input, the posterior moments can be evaluated
analytically, so only the outer expectation in
Theorem~\ref{thm:main_result} requires numerical integration.
Let
\begin{equation*}
    a=\frac{P}{P+\sigma^2},
\end{equation*}
and define the real symmetric matrix
$\bm{C}\in\mathbb{R}^{B_x\times B_x}$ by
\begin{equation}
    \begin{aligned}
        C_{ij}
        &=
        \mathbb{E}
        \left[e^{-i\phi_i}e^{i\phi_j}\right]
        -
        \mathbb{E}[e^{-i\phi_i}]
        \mathbb{E}[e^{i\phi_j}]\\
        &=
        e^{-c^2|i-j|/2}
        -
        e^{-c^2(i+j)/2}\\
        &=
        e^{-c^2|i-j|/2}
        \left(1-e^{-c^2\min\{i,j\}}\right).
    \end{aligned}
    \label{eq:pnc-phase-covariance}
\end{equation}
The second equality follows from
$\operatorname{Var}(\phi_i)=c^2i$ and
$\operatorname{Var}(\phi_i-\phi_j)=c^2|i-j|$.

By Lemma~\ref{lem:pnc-posterior-moments},
$p(\boldsymbol{\phi}\mid\by)=p(\boldsymbol{\phi})$, and Gaussian
conditioning gives
\begin{equation*}
    \bX\mid\by,\boldsymbol{\phi}
    \sim
    \mathcal{CN}\left(
        a\operatorname{diag}
        (e^{-i\phi_1},\ldots,e^{-i\phi_{B_x}})\by,
        \frac{P\sigma^2}{P+\sigma^2}\bI_{B_x}
    \right).
\end{equation*}
Consequently, the law of total covariance yields
\begin{equation}
    \bm{\Sigma}(\by)
    =
    \frac{P\sigma^2}{P+\sigma^2}\bI_{B_x}
    +
    a^2\operatorname{diag}(\by)\bm{C}
    \operatorname{diag}(\by)^H.
    \label{eq:pnc-posterior-covariance-numerical}
\end{equation}
Thus neither the posterior mean nor the posterior covariance
needs to be estimated by sampling.

Since
$\bY\sim\mathcal{CN}(\boldsymbol{0},(P+\sigma^2)\bI_{B_x})$,
the normalized output energies
\begin{equation*}
    E_i=\frac{|Y_i|^2}{P+\sigma^2},
    \qquad i=1,\ldots,B_x,
\end{equation*}
are independent unit-mean exponential random variables.
Let $\xi_1(\bm{E}),\ldots,\xi_{B_x}(\bm{E})$ be the eigenvalues of
\begin{equation*}
    \operatorname{diag}(\sqrt{E_1},\ldots,\sqrt{E_{B_x}})
    \bm{C}
    \operatorname{diag}(\sqrt{E_1},\ldots,\sqrt{E_{B_x}}).
\end{equation*}
Removing the phases of the entries of $\bY$ by a diagonal unitary
transformation in \eqref{eq:pnc-posterior-covariance-numerical},
the eigenvalues of $\bm{\Sigma}(\bY)$ are
\begin{equation}
    \lambda_i(\bY)
    =
    P\left(1-a+a\xi_i(\bm{E})\right),
    \qquad i=1,\ldots,B_x.
    \label{eq:pnc-posterior-eigenvalues-numerical}
\end{equation}
Moreover, Lemma~\ref{lem:pnc-posterior-moments} gives
\begin{equation*}
    \frac{1}{B_xP}
    \mathbb{E}_{\bY}
    \left[\|\bm{\mu}(\bY)\|_2^2\right]
    =
    \frac{a}{B_x}\sum_{i=1}^{B_x}e^{-c^2i}
    =
    a\eta,
\end{equation*}
where $\eta$ is defined in Proposition~\ref{prop:lin-GNNDR-pnc}.
Substituting these identities into
Theorem~\ref{thm:main_result} yields
\begin{equation}
    I_{\mathrm{GMI,opt}}
    =
    \frac{1}{B_x}
    \mathbb{E}_{\bm{E}}
    \left[
        \sum_{i=1}^{B_x}
        \psi\left(P(1-a+a\xi_i(\bm{E}))\right)
    \right]
    +
    a\eta,
    \label{eq:pnc-opt-gmi-numerical}
\end{equation}
where $\psi$ retains its piecewise definition in
\eqref{eq:psi_thm1}.

We evaluate the expectation in \eqref{eq:pnc-opt-gmi-numerical}
using randomized quasi-Monte Carlo integration.
For each independent scrambling $r=1,\ldots,J$, generate $M$
scrambled Sobol points
$\bm{U}^{(r,n)}\in(0,1)^{B_x}$, $n=1,\ldots,M$, and apply
the transformation
\begin{equation*}
    E_i^{(r,n)}=-\log U_i^{(r,n)},
    \qquad i=1,\ldots,B_x.
\end{equation*}
After computing the eigenvalues
$\xi_i(\bm{E}^{(r,n)})$, the estimate from the $r$-th scrambling is
\begin{equation}
    \hat{I}_{\mathrm{GMI,opt}}^{(r)}
    =
    \frac{1}{MB_x}
    \sum_{n=1}^{M}\sum_{i=1}^{B_x}
    \psi\left(
        P\left(1-a+a\xi_i(\bm{E}^{(r,n)})\right)
    \right)
    +
    a\eta.
    \label{eq:pnc-opt-gmi-estimator}
\end{equation}
The final estimate and its estimated standard error are
\begin{equation*}
    \hat{I}_{\mathrm{GMI,opt}}
    =
    \frac{1}{J}\sum_{r=1}^{J}
    \hat{I}_{\mathrm{GMI,opt}}^{(r)},
    \qquad
    \hat{\operatorname{SE}}
    =
    \left[
        \frac{1}{J(J-1)}
        \sum_{r=1}^{J}
        \left(
            \hat{I}_{\mathrm{GMI,opt}}^{(r)}
            -
            \hat{I}_{\mathrm{GMI,opt}}
        \right)^2
    \right]^{1/2}.
\end{equation*}
The standard error is computed across independent scramblings,
rather than by treating the points within a Sobol sequence as
independent samples.

The default simulation setting uses $J=8$ independent scramblings,
$M=2^{14}$ points per scrambling, and double-precision arithmetic.
For each fixed $(B_x,c)$, the same points and eigenvalues are reused
over the SNR range from $-5$ to $25$~dB in increments of $0.5$~dB.
Indeed, $\bm{C}$ and $\xi_i(\bm{E})$ do not depend on the SNR;
only $a$ and the scalar arguments of $\psi$ change.
This avoids repeated eigenvalue decompositions and requires neither
posterior sampling nor phase discretization.

\subsection{Posterior moments for the phase noise channel}
We first calculate the posterior mean and the average conditional covariance matrix in the following lemma.
\begin{lemma}[Posterior moments for the phase noise channel]
\label{lem:pnc-posterior-moments}
Consider the phase noise channel in \eqref{eq:channel_phase}, with
$\bX\sim\mathcal{CN}(\boldsymbol{0},P\bI_{B_x})$, $P>0$, and
$\sigma^2>0$. Define
\begin{equation*}
    \bm{D}
    =
    \operatorname{diag}
    \left(e^{-c^2/2},\ldots,e^{-c^2B_x/2}\right).
\end{equation*}
Then $\bY$ is independent of the phase vector $\boldsymbol{\phi}$, and
the posterior mean and the average conditional covariance matrix are given by 
\begin{equation}
    \bm{\mu}(\by)
    =
    \mathbb{E}[\bX\mid\by]
    =
    \frac{P}{P+\sigma^2}\bm{D}\by, \qquad \overline{\bm{\Sigma}}
    :=
    \mathbb{E}_{\bY}
    \left[\operatorname{Cov}[\bX\mid\bY]\right]
    =
    P\bI_{B_x}
    -
    \frac{P^2}{P+\sigma^2}\bm{D}^2.
    \label{eq:pnc-posterior-mean-covariance-matrix}
\end{equation}
Moreover,
\begin{equation}
    \mathbb{E}[\bY\bY^H]
    =
    (P+\sigma^2)\bI_{B_x}, \qquad  \mathbb{E}[\bX\bY^H]
    =
    P\bm{D}.
    \label{eq:self-cross-corr}
\end{equation}
In particular,
\begin{equation*}
    \operatorname{Cov}[\bX\mid\by]
    \succeq
    \frac{P\sigma^2}{P+\sigma^2}\bI_{B_x}
    \succ0.
\end{equation*}
\end{lemma}

\begin{proof}
Write
\begin{equation*}
    \bY=\bm{\Theta}\bX+\bN,
    \qquad
    \bm{\Theta}
    =
    \operatorname{diag}
    \left(e^{i\phi_1},\ldots,e^{i\phi_{B_x}}\right).
\end{equation*}
Since $\bm{\Theta}$ is unitary and the input is white Gaussian,
\begin{equation*}
    \bY\mid\boldsymbol{\phi}
    \sim
    \mathcal{CN}
    \left(\boldsymbol{0},(P+\sigma^2)\bI_{B_x}\right).
\end{equation*}
This distribution does not depend on $\boldsymbol{\phi}$.
Consequently, $\bY$ and $\boldsymbol{\phi}$ are independent,
$p(\boldsymbol{\phi}\mid\by)=p(\boldsymbol{\phi})$, and
\eqref{eq:self-cross-corr} follows.

Conditioned on $\by$ and $\boldsymbol{\phi}$, Gaussian conditioning gives
\begin{equation*}
    \bX\mid\by,\boldsymbol{\phi}
    \sim
    \mathcal{CN}\left(
        \frac{P}{P+\sigma^2}\bm{\Theta}^H\by,
        \frac{P\sigma^2}{P+\sigma^2}\bI_{B_x}
    \right).
\end{equation*}
Since $\phi_i\sim\mathcal{N}(0,c^2i)$,
\begin{equation*}
    \mathbb{E}[e^{-i\phi_i}]
    =
    e^{-c^2i/2},
    \qquad
    \mathbb{E}[\bm{\Theta}^H]=\bm{D}.
\end{equation*}
Taking expectation with respect to
$p(\boldsymbol{\phi}\mid\by)=p(\boldsymbol{\phi})$ yields
\eqref{eq:pnc-posterior-mean-covariance-matrix}.
Similarly, independence of $\bX$, $\bN$, and $\boldsymbol{\phi}$ gives
\begin{equation*}
    \mathbb{E}[\bX\bY^H]
    =
    \mathbb{E}
    \left[\bX\bX^H\bm{\Theta}^H\right]
    =
    P\bm{D},
\end{equation*}
proving \eqref{eq:self-cross-corr}.

By the law of total covariance,
\begin{equation*}
    \begin{aligned}
        \overline{\bm{\Sigma}}
        =
        P\bI_{B_x}
        -
        \mathbb{E}_{\bY}
        \left[\bm{\mu}(\bY)\bm{\mu}^H(\bY)\right]
        =
        P\bI_{B_x}
        -
        \frac{P^2}{(P+\sigma^2)^2}
        \bm{D}\mathbb{E}[\bY\bY^H]\bm{D}
        =
        P\bI_{B_x}
        -
        \frac{P^2}{P+\sigma^2}\bm{D}^2.
    \end{aligned}
\end{equation*}
Finally, applying the law of total covariance conditionally on $\by$,
\begin{equation*}
    \begin{aligned}
        \operatorname{Cov}[\bX\mid\by]
        =
        \mathbb{E}_{\boldsymbol{\phi}\mid\by}
        \left[\operatorname{Cov}[\bX\mid\by,\boldsymbol{\phi}]\right]+
        \operatorname{Cov}_{\boldsymbol{\phi}\mid\by}
        \left(\mathbb{E}[\bX\mid\by,\boldsymbol{\phi}]\right)\succeq
        \frac{P\sigma^2}{P+\sigma^2}\bI_{B_x}.
    \end{aligned}
\end{equation*}
This completes the proof.
\end{proof}

\subsection{Proof for Proposition~\ref{prop:lin-GNNDR-pnc}
(Restricted Vec-GNNDR for the phase noise channel)}

\begin{proof}
By Lemma~\ref{lem:pnc-posterior-moments},
\begin{equation*}
    \bm{\mu}(\by)
    =
    \frac{P}{P+\sigma^2}\bm{D}\by,
    \qquad
    \overline{\bm{\Sigma}}
    =
    \operatorname{diag}(r_1,\ldots,r_{B_x}),
\end{equation*}
where
\begin{equation*}
    r_i=\frac{P\nu_i}{P+\sigma^2},
    \qquad
    P-r_i=\frac{P^2e^{-c^2i}}{P+\sigma^2}.
\end{equation*}
In particular, $0<r_i<P$ for every $i$.

\begin{enumerate}
    \item[(i)] \emph{Constant scalar scaling function.}
    The normalized trace of the average conditional covariance is
    \begin{equation*}
        T
        =
        \frac{1}{B_x}
        \operatorname{Tr}(\overline{\bm{\Sigma}})
        =
        \frac{P\overline{\nu}}{P+\sigma^2},
        \qquad
        P-T=\frac{P^2\eta}{P+\sigma^2}.
    \end{equation*}
    Proposition~\ref{prop:Vec-GNNDR-cssf} therefore gives
    \begin{equation*}
        I_{\mathrm{GMI,cssf}}
        =
        \log\frac{P}{T}
        =
        \log\frac{P+\sigma^2}{\overline{\nu}}.
    \end{equation*}
    Substituting the expressions for $T$, $P-T$, and $\bm{\mu}(\by)$
    into the corresponding optimal processing and scaling functions yields
    \begin{equation*}
        \begin{aligned}
            \bg_{\mathrm{cssf}}(\by)
            =
            \sqrt{\frac{P}{T(P-T)}}\,\bm{\mu}(\by)
            =
            \frac{\bm{D}\by}{\sqrt{\eta\overline{\nu}}}, \qquad
            \bm{f}_{\mathrm{cssf}}(\by)
            =
            \sqrt{\frac{P-T}{PT}}\,\bI_{B_x}
            =
            \sqrt{\frac{\eta}{\overline{\nu}}}\,\bI_{B_x}.
        \end{aligned}
    \end{equation*}

    \item[(ii)] \emph{Constant matrix scaling function.}
    Applying Proposition~\ref{prop:Vec-GNNDR-cmsf} gives
    \begin{equation*}
        \begin{aligned}
            I_{\mathrm{GMI,cmsf}}
            =
            \frac{1}{B_x}
            \log\frac{P^{B_x}}{\det(\overline{\bm{\Sigma}})} =
            \frac{1}{B_x}
            \sum_{i=1}^{B_x}
            \log\frac{P+\sigma^2}{\nu_i}.
        \end{aligned}
    \end{equation*}
    Since $\overline{\bm{\Sigma}}$ is diagonal, the optimal pair can be
    written in the coordinate basis, up to a common permutation of the
    processed output coordinates, as
    \begin{equation*}
        \begin{aligned}
            &\bg_{\mathrm{cmsf}}(\by)
            =
            \operatorname{diag}
            \left(
                \sqrt{\frac{P}{(P-r_i)r_i}}
            \right)_{i=1}^{B_x}
            \bm{\mu}(\by)=
            \operatorname{diag}
            \left(\nu_1^{-1/2},\ldots,\nu_{B_x}^{-1/2}\right)\by,\\
            &\bm{f}_{\mathrm{cmsf}}(\by)
            =
            \operatorname{diag}
            \left(
                \sqrt{\frac{P-r_i}{Pr_i}}
            \right)_{i=1}^{B_x}
            =
            \operatorname{diag}
            \left(
                \frac{e^{-c^2/2}}{\sqrt{\nu_1}},
                \ldots,
                \frac{e^{-c^2B_x/2}}{\sqrt{\nu_{B_x}}}
            \right).
        \end{aligned}
    \end{equation*}

    \item[(iii)] \emph{Linear processing function.}
    By Lemma~\ref{lem:pnc-posterior-moments}, the matrix in
    \eqref{eq:lin-Q} satisfies
    \begin{equation*}
        \bm{Q}
        =
        \frac{P^2}{P+\sigma^2}\bm{D}^2
        =
        P\bI_{B_x}-\overline{\bm{\Sigma}}.
    \end{equation*}
    Hence Proposition~\ref{prop:Vec-GNNDR-lin} gives
    \begin{equation*}
        I_{\mathrm{GMI,lin}}
        =
        \frac{1}{B_x}
        \log\frac{P^{B_x}}{\det(P\bI_{B_x}-\bm{Q})}
        =
        I_{\mathrm{GMI,cmsf}}.
    \end{equation*}
    Moreover, the CMSF-optimal processing function in part~(ii) is already
    linear. Its processing--scaling pair is therefore feasible for the
    linear-processing scheme and attains the same GMI. In particular,
    \begin{equation*}
        \bm{\Gamma}
        =
        \operatorname{diag}
        \left(\nu_1^{-1/2},\ldots,\nu_{B_x}^{-1/2}\right),
        \qquad
        \bm{\Pi}
        =
        \operatorname{diag}
        \left(
            \frac{e^{-c^2i/2}}{\sqrt{\nu_i}}
        \right)_{i=1}^{B_x}.
    \end{equation*}
\end{enumerate}
This completes the proof.
\end{proof}

\subsection{Proof for Proposition~\ref{prop:pnc_id} (Phase noise channel with identity decoding rule)}
\begin{proof}
   Consider the identity decoding rule $d(\bx, \by) = |\by- \bx|_2^2$.
    Then the GMI in \eqref{eq:max_GMI_opt} takes the form
    \begin{equation} \label{eq:GMI_id_appendix}
        I_{\text{GMI, id}} = \max_{\theta<0} \left\{ \theta \left[\sigma^2 + \frac{P}{B_x} \sum_{i=1}^{B_x}  \mathbb{E}_{\theta_i }(2- 2 \cos \theta_i)\right] - \frac{(P+\sigma^2)\theta}{1 - \theta P} + \log(1-P\theta)\right\},
    \end{equation}
    with $\theta_i \sim \cN(0, c^2 i)$. We first evaluate the expectation with respect to $\theta_i$:
    \begin{equation*}
        \mathbb{E}_{\theta_i}(2 - 2 \cos \theta_i) = 2 (1 - e^{-c^2 i / 2}).
    \end{equation*}
    Define $A = \sigma^2 + \frac{2P}{B_x} \sum_{i=1}^{B_x}(1- e^{-c^2 i / 2})$, then the maximization problem for $\theta$ reduces to 
    \begin{equation*}
        h(\theta) = \theta A - \frac{(P+\sigma^2) \theta}{1 - \theta P} + \log(1 - \theta P).
    \end{equation*}
    Taking the derivative of $h(\theta)$ with respect to $\theta$ and setting it to zero yields
    \begin{equation*}
        A- \frac{(P+\sigma^2)}{(1- \theta P)^2} - \frac{P}{1- \theta P} = 0.
    \end{equation*}
    The solution is
    \begin{equation*}
        \theta^{*}=\frac{2 A-P-\sqrt{P^{2}+4 A\left(P+\sigma^{2}\right)}}{2 A P}.
    \end{equation*}
    Substituting $\theta^{*}$ into \eqref{eq:GMI_id_appendix} gives the desired expression, which completes the proof.
\end{proof}

\subsection{Calculation for the high-SNR limit of GMIs}

Fix $B_x\ge1$ and $c>0$, and consider $P/\sigma^2\to\infty$.
By Proposition~\ref{prop:lin-GNNDR-pnc},
\begin{align*}
I_{\mathrm{GMI,cssf}}
=\log\left(
\frac{1+\sigma^2/P}{1-\eta+\sigma^2/P}
\right), \qquad
I_{\mathrm{GMI,cmsf}}
=I_{\mathrm{GMI,lin}}
=\frac{1}{B_x}\sum_{i=1}^{B_x}
\log\left(
\frac{1+\sigma^2/P}{1-e^{-c^2i}+\sigma^2/P}
\right),
\end{align*}
where $\eta=B_x^{-1}\sum_{i=1}^{B_x}e^{-c^2i}\in(0,1)$.
Taking the high-SNR limit gives
\begin{equation}
\label{eq:GMI_cssf_limit}
\lim_{\mathrm{SNR}\to\infty}I_{\mathrm{GMI,cssf}}
=
\log\left(\frac{1}{1-\eta}\right),
\end{equation}
and
\begin{equation}
\label{eq:GMI_cmsf_lin_limit}
\begin{aligned}
\lim_{\mathrm{SNR}\to\infty}I_{\mathrm{GMI,cmsf}}
=\lim_{\mathrm{SNR}\to\infty}I_{\mathrm{GMI,lin}} =\frac{1}{B_x}\sum_{i=1}^{B_x}
\log\left(\frac{1}{1-e^{-c^2i}}\right).
\end{aligned}
\end{equation}

For the identity decoder, define
\[
s=\frac{2}{B_x}\sum_{i=1}^{B_x}
\left(1-e^{-c^2i/2}\right)\in(0,2),
\qquad
\alpha_P=\frac{A}{P}=s+\frac{\sigma^2}{P}.
\]
By Proposition~\ref{prop:pnc_id}, the optimizing parameter satisfies
\begin{equation}
\label{eq:def_kappa}
\begin{aligned}
P\theta^*
=\frac{
2\alpha_P-1-\sqrt{1+4\alpha_P(1+\sigma^2/P)}
}{2\alpha_P} \overset{\mathrm{SNR} \to \infty}{\longrightarrow}
\kappa:=
\frac{2s-1-\sqrt{1+4s}}{2s}.
\end{aligned}
\end{equation}
Since $\alpha_P\to s$ and
$1-\kappa=2/(\sqrt{1+4s}-1)$, substituting
\eqref{eq:def_kappa} into \eqref{eq:GMI_id_appendix} yields
\begin{equation}
\label{eq:GMI_id_limit1}
\begin{aligned}
\lim_{\mathrm{SNR}\to\infty}I_{\mathrm{GMI,id}}
&=s\kappa-\frac{\kappa}{1-\kappa}+\log(1-\kappa)\\
&=s+1-\sqrt{1+4s}
+\log\left(\frac{2}{\sqrt{1+4s}-1}\right).
\end{aligned}
\end{equation}
Thus, all four schemes have finite high-SNR limits for fixed
$B_x$ and $c>0$.

\bibliographystyle{IEEEtran.bst}
\bibliography{ref2}

\end{document}